\documentclass[11pt]{article}
\usepackage{amsmath, amssymb, amsthm, mathtools, listings, inconsolata, enumerate, units}
\usepackage{amsfonts}
\usepackage{color}
\usepackage{tikz}
\usepackage{pgfplots}
\usepackage{listings}
\usepackage{courier}
\usepackage{tcolorbox}
\usepackage{bm}
\usepackage{dsfont}
\usepackage{float}
\usepackage[round]{natbib}
\usepackage[linkcolor=blue,citecolor=blue, colorlinks=true]{hyperref}
\usepackage{booktabs}
\usepackage{geometry}
\usepackage{indentfirst}

\allowdisplaybreaks[1]
\lstdefinestyle{code} {basicstyle = \linespread{1}\ttfamily\small, showstringspaces = false}
\newtheorem{conj}{Conjecture}

\newtheorem{theorem}[conj]{Theorem}
\newtheorem{prop}[conj]{Proposition}
\newtheorem{assumption}{Assumption} 
\newtheorem{lemma}{Lemma} 
\newtheorem{corollary}{Corollary}
\theoremstyle{definition}
\newtheorem{remark}{Remark}

\renewcommand{\b}[1]{\ensuremath{\bm{\mathrm{#1}}}}

\DeclareMathOperator*{\argmin}{arg\,min}
\DeclareMathOperator*{\argmax}{arg\,max}

\newcommand{\be}{\bm{e}}

\newcommand{\bu}{\bm{u}}
\newcommand{\bv}{\bm{v}}

\newcommand{\by}{\bm{y}}

\newcommand{\bA}{\bm{A}}
\newcommand{\bB}{\bm{B}}

\newcommand{\bD}{\bm{D}}
\newcommand{\bE}{\bm{E}}
\newcommand{\bF}{\bm{F}}
\newcommand{\bG}{\bm{G}}
\newcommand{\bH}{\bm{H}}
\newcommand{\bI}{\bm{I}}

\newcommand{\bL}{\bm{L}}
\newcommand{\bM}{\bm{M}}

\newcommand{\bQ}{\bm{Q}}

\newcommand{\bU}{\bm{U}}
\newcommand{\bV}{\bm{V}}
\newcommand{\bW}{\bm{W}}
\newcommand{\bX}{\bm{X}}
\newcommand{\bY}{\bm{Y}}
\newcommand{\bZ}{\bm{Z}}

\newcommand{\cC}{\mathcal{C}}

\newcommand{\cE}{\mathcal{E}}

\newcommand{\cM}{\mathcal{M}}

\newcommand{\cO}{\mathcal{O}}

\newcommand{\cR}{\mathcal{R}}
\newcommand{\cS}{{\mathcal{S}}}

\newcommand{\EE}{\mathbb{E}}

\newcommand{\PP}{\mathbb{P}}

\newcommand{\RR}{\mathbb{R}}

\newcommand{\bbeta}{\bm{\beta}}
\newcommand{\bgamma}{\bm{\gamma}}
\newcommand{\bepsilon}{\bm{\epsilon}}
\newcommand{\bdelta}{\bm{\delta}}

\newcommand{\btheta}{\bm{\theta}}

\newcommand{\bomega}{\bm{\omega}}

\newcommand{\bGamma}{\bm{\Gamma}}
\newcommand{\bDelta}{\bm{\Delta}}
\newcommand{\bTheta}{\bm{\Theta}}

\newcommand{\bOmega}{\bm{\Omega}}

\usepackage{comment}
\usepackage{tikz}

\DeclarePairedDelimiter{\floor}{\lfloor}{\rfloor}
\usepackage{enumitem}
\usepackage[]{algorithm2e}

\def\beq{\begin{equation}\begin{aligned}[b]}

\def\eeq{\end{aligned}\end{equation}}

\usepackage{comment}
\def\Cov{\text{Cov}}   
\def\PP{\mathbb{P}}
\def\EE{\mathbb{E}}
\def\RR{\mathbb{R}}
\def\wh{\widehat}

\def\bI{{\bm I}}
\def\wt{\widetilde}
\def\rank{\textrm{rank}}
\def\i{\infty}
\def\diag{\textrm{diag}}
\def\g{\gamma}
\def\l12{\ell_1/\ell_2}

\def\op{{\rm op}}
\def\tr{{\rm tr}}
\def\sw{\Sigma_W}
\def\se{\Sigma_E}
\def\lbdIj{\lambda_1^{(j)}}
\def\lbdj{\lambda_2^{(j)}}
\def\seps{\Sigma_{\epsilon}}

\providecommand{\norm}[1]{\|#1\|}
\providecommand{\bnorm}[1]{\big\|#1\big\|}

\newcommand*{\supp}{\mathrm{supp}}

\usepackage{caption}
\title{Inference in  High-dimensional Multivariate Response Regression with Hidden Variables}
\author{Xin Bing \thanks{Department of 
Pure Mathematics and Mathematical Statistics, University of Cambridge. E-mail: \texttt{xb228@cam.ac.uk}}
~~~~~ Wei Cheng \thanks{Center for Computational Molecular Biology, Brown University. E-mail: \texttt{wei\_cheng1@brown.edu}}
~~~~~ Huijie Feng \thanks{Microsoft.  E-mail: \texttt{huijiefeng@microsoft.com}.}
~~~~~Yang Ning \thanks{Department of Statistics and Data Science, Cornell University.  E-mail: \texttt{yn265@cornell.edu}.}}

\date{}

\begin{document}

\maketitle

\begin{abstract}
    This paper studies the inference of the regression coefficient matrix under multivariate response linear regressions in the presence of hidden variables. 
    A novel procedure for constructing confidence intervals of entries of the coefficient matrix is proposed. Our method first utilizes the multivariate nature of the responses by estimating and adjusting the hidden effect to construct an initial estimator of the coefficient matrix. By further deploying a low-dimensional projection procedure to reduce the bias introduced by the regularization in the previous step, a refined estimator is proposed and shown to be asymptotically normal. The asymptotic variance of the resulting estimator is derived with closed-form expression and can be consistently estimated. In addition, we propose a testing procedure for the existence of hidden effects and provide its theoretical justification. Both our procedures and their analyses are valid even when the feature dimension and the number of responses exceed the sample size. Our results are further backed up via extensive simulations and a real data analysis. 
\end{abstract}

\section{Introduction}\label{sec_intro}

Multivariate response linear regression is a widely used approach of discovering the association between a response vector $Y$ and a feature vector $X$ in a variety of applications \citep{anderson_book}. Oftentimes, there may exist some unobservable, hidden, variables $Z$ that correlate with both the response $Y$ and the feature $X$. For example, in genomics studies, $Y$ typically represents different gene expressions, $X$ contains a set of exposures (e.g. levels of treatment), and $Z$ corresponds to the unobserved batch effect  \citep{Leek2008,Gagnon2012}. In causal inference, one can interpret $X$ as the multiple causes of $Y$ and treat $Z$ as confounders, which are unobserved due to cost constraint or ethical issue \citep{silva2006learning,janzing2018detecting,wang2019blessings}. Since $X$ and $Z$ are often correlated, ignoring the hidden variables $Z$ in the regression model may lead to spurious association between $X$ and $Y$. Therefore, accounting for the existence of such hidden variables is critical to draw valid scientific conclusions.

This paper studies the following multivariate response linear regression with hidden variables,
\begin{equation}\label{model}
    Y = \bTheta^T X + \bB^T Z + E,
\end{equation}
where $Y\in \RR^m$ is the multivariate response, $X\in \RR^p$ is the random vector of $p$ observable features while $Z\in \RR^K$ is the random vector of $K$ unobservable, hidden, variables, that are possibly correlated with $X$. The number of hidden variables $K$ is unknown and is assumed to be no greater than the number of responses $m$. The random vector $E\in \RR^m$ is the additive noise independent of $X$ and $Z$. Assume the observed data $(\bY, \bX) \in (\RR^{n\times m}, \RR^{n\times p})$ consist of $n$ i.i.d. samples $(\bY_i, \bX_i)$, for $i\in [n] := \{1,\ldots,n\}$, from model (\ref{model}). 
Throughout the paper,  we focus on the high-dimensional setting, that is both $m$ and $p$ can grow with the sample size $n$. Without loss of generality, we assume $\EE(X)=\b0$ and $\EE(Z)=\b0$ as we can always center the data $\bY$ and $\bX$.

In model (\ref{model}), the coefficient matrix $\bTheta \in \RR^{p\times m}$ encodes the association between $X$ and $Y$ after adjusting the hidden variables $Z$, and is of our primary interest. More precisely, for any given $i\in[p]$ and $j\in[m]$, we are interested in constructing confidence intervals for $\Theta_{ij}$, or equivalently, testing the following hypothesis:
\begin{equation}\label{def_target}
    H_{0,\Theta_{ij}}: ~  \Theta_{ij} = 0, \qquad \textrm{versus} \qquad H_{1,\Theta_{ij}}:~ \Theta_{ij}\ne 0.
\end{equation}
Our secondary interest is to answer the question that whether the $j$th response $Y_j$ is affected by any of the hidden variables. Since each column $\bB_j\in \RR^K$ of the matrix $\bB = (\bB_1, \ldots, \bB_m)$ corresponds to the coefficient of the hidden effects of $Z$ on $Y_j$, we can answer the above question by testing the hypothesis:
\beq\label{def_target_B}
    H_{0, B_j}: \bB_{j} = \b0,\qquad \textrm{versus} \qquad H_{1,B_j}: \bB_j \ne \b0.
\eeq
In particular, if the null hypothesis $H_{0, B_j}$ is rejected, then the effect of the hidden variables $Z$ on $Y_j$ is significant, suggesting the necessity of adjusting the hidden effects for modelling $Y_j$.

Since we allow $X$ and $Z$ to be correlated in (\ref{model}), we can decouple their dependence via the $L_2$ projection of $Z$ onto the linear space of $X$:
\beq\label{model_ZX}
    Z = \bA^T X + (Z - \bA^T X) := \bA^T X + W,
\eeq
where $\bA = (\EE[XX^T])^{-1}\EE[XZ^T]\in \RR^{p\times K}$ and $W=Z - \bA^T X$ satisfies $\EE[WX^T]=\b0$. While $W$ and $X$ are uncorrelated, we do not require them to be independent. In other words, (\ref{model_ZX}) does not imply that $X$ and $Z$ follow a linear regression model. Indeed, our framework allows any nonlinear dependence structure between $X$ and $Z$ and is therefore model free for the joint distribution of $(X,Z)$. Under such decomposition, 
the original model (\ref{model}) can be rewritten as 
\beq\label{model_linear}
    Y = (\bTheta + \bA\bB)^T X + \epsilon
\eeq
where the new error term $\epsilon \coloneqq \bB^T W + E$ has zero mean and is uncorrelated with $X$. Before we elaborate how we make inference on $\Theta_{ij}$ and $\bB_j$, we start with a brief review of the related literature.



\subsection{Related literature}

Surrogate variable analysis (SVA) has been widely used to estimate and make inference on $\bTheta$ under model (\ref{model}) for genomics data \citep{Leek2008,Gagnon2012}. Recent progress has been made in \cite{Lee2017,wang2017,McKennan19} towards both developing new methodologies and understanding the theoretical properties of the existing approaches.  However, all existing SVA-related approaches rely on the ordinary least squares (OLS) between $\bY$ and $\bX$ to estimate $\bTheta + \bA\bB$ in (\ref{model_linear}), hence are only feasible when the feature dimension, $p$, is small comparing to the sample size $n$. As researchers tend to collect far more features than before due to advances of modern technology, there is a need of developing new method which allows the feature dimension $p$ to grow with, or even exceed, the sample size $n$. 

More recently, \cite{bing2020adaptive} studied the estimation of $\bTheta$ under model (\ref{model}). Their proposed procedure assumes a row-wise sparsity structure on $\bTheta$ and is suitable for $p$ that is potentially greater than $n$. Despite the advance on the estimation aspect, conducting inference on $\bTheta$ remains an open problem when $p$ is larger than $n$. The extra difficulty of making inference comparing to estimation in the high-dimensional regime is already visible in the ideal scenario, the sparse linear regression models, without any hidden variable, see  \cite{zhangzhang2014,vandegeer2014,belloni2015uniform,Javanmard2018,ning2017general}, among many others. 
Inference of the linear coefficient in the presence of hidden variables, to the best of our knowledge,  is only studied in  \cite{guo2020doubly} for the univariate case $\by = \bX\btheta + \bZ\bbeta +\bepsilon$ where $\by\in \RR^n$ is the univariate response, $\bX\in \RR^{n\times p}$ consists of the high-dimensional feature and $\bZ\in \RR^{n\times K}$ represents the hidden confounders. By further assuming $\bX = \bZ\bGamma^T + \bW'$ for some loading matrix $\bGamma$ and additive error $\bW'$ 
independent of $\bZ$, \cite{guo2020doubly} proposed a doubly debiased lasso procedure for making inference on entries of $\btheta$. Our situation differs from theirs in that we have multiple responses. By borrowing strength across multivariate responses, we are able to remove the hidden effects without assuming any model between $\bX$ and $\bZ$. Moreover, combining multiple responses provides additional information on the coefficient matrix, $\bB$, of the hidden variable, which not only helps to remove the hidden effects in our estimation procedure for $\bTheta$, but also enables us to test and quantify the hidden effects for each response.

In model  (\ref{model_linear}), when $\bTheta$ is sparse and the matrix $\bL \coloneqq \bA\bB$ has a small rank $K$, our problem is related to the recovery of an additive decomposition of a sparse matrix and a low-rank matrix, as studied by \cite{chandrasekaran2012latent,Candes,Hsu2011}, just to name a few. In order to identify and estimate  $\bTheta$, \cite{chandrasekaran2012latent} proposed a penalized $M$-estimator under  certain incoherence conditions between $\bTheta$
 and $\bL$. By contrast, our identifiability conditions (see, Section \ref{sec_id}) differ significantly from theirs, hence leading to a completely different procedure for estimation. Furthermore, this strand of works only focus on estimation while our interest in this paper is about inference.

\subsection{Main contributions}

   Our first contribution is in establishing an identifiability result of $\bTheta$ in Theorem \ref{thm_ident} of Section \ref{sec_id} under model (\ref{model}) when the entries of $E$ in (\ref{model}) are allowed to be correlated, that is, $\se:=\Cov(E)$ is non-diagonal. To the best of our knowledge, the existing literature only studies the identifiability of $\bTheta$ when $\se$ is diagonal, see, for instance, \cite{Lee2017,wang2017,McKennan19,bing2020adaptive}. In Section \ref{sec_id} we also discuss different sets of conditions under which $\bTheta$ can be identified asymptotically as $m\to \infty$ when $\se$ is non-diagonal. 
   
   Our second contribution is to propose a new procedure in Section \ref{sec_method} for constructing confidence intervals of $\Theta_{ij}$ that is suitable even when $p$ is larger than $n$. Our procedure consists of four steps: the first step in Section \ref{sec_pred} estimates the coefficient matrix $(\bTheta + \bA\bB)$ in (\ref{model_linear}); the second step in Section \ref{sec_est_B} estimates $\bB$, the coefficient matrix of the hidden variables, using the residual matrix from the first step; the third step uses the estimate of $\bB$ to remove the hidden effect and construct an initial estimator $\wh\Theta_{ij}$ of $\Theta_{ij}$, while our final step constructs the refined estimator $\wt\Theta_{ij}$ of $\Theta_{ij}$ by removing the bias of $\wh\Theta_{ij}$ due to the high-dimensional regularization (see, Section \ref{sec_est_Theta}). The resulting estimate $\wt\Theta_{ij}$ is further used  to construct confidence intervals of $\Theta_{ij}$ and to test the hypothesis  (\ref{def_target}) in Section \ref{sec_est_Theta}. Finally, in Section \ref{sec_method_infer_B}, we further propose a $\chi^2$-based statistic for testing the null hypothesis $\bB_j = \b0$ for any given $j$. 
   
   Our third contribution is to provide statistical guarantees for the aforementioned procedure. Our main result, stated in Theorem \ref{thm_asymp_normal} of Section \ref{sec_theory_ASN}, shows that our estimator $\wt\Theta_{ij}$ of $\Theta_{ij}$ satisfies $\sqrt{n}(\wt \Theta_{ij}-\Theta_{ij}) = \xi + \Delta$ where $\xi$ is normally distributed, conditioning on the design matrix, and $\Delta$ is asymptotically negligible as $n\to \i$. In Section \ref{sec_effciency}, we further show that $\wt\Theta_{ij}$ is asymptotically efficient in the Gauss-Markov sense, and its asymptotic variance can be consistently estimated. Combining these results justifies the usage of our proposed procedure in Section \ref{sec_est_Theta} for making inference on $\Theta_{ij}$.  In the proof of Theorem \ref{thm_asymp_normal}, an important intermediate result we derived is the (column-wise) uniform $\ell_2$ convergence rate of our estimator $\wh \bB$, which is stated in Theorem \ref{thm_rates_B}. On top of this result, we further establish the asymptotic normality of $\wh\bB_j$ for any $j\in[m]$ with explicit expression of the asymptotic variance in Theorem \ref{thm_B_asn}. The result provides theoretical guarantees for the $\chi^2$-based statistic in Section \ref{sec_method_infer_B} for testing $\bB_j = \b0$.  
   

   The remainder of this paper is organized as follows. In Section \ref{sec_id} we establish the identifiability result of $\bTheta$. Section \ref{sec_method} contains the methodology of making inference on $\Theta_{ij}$ and $\bB_j$. Statistical guarantees are provided in Section \ref{sec_theory}. 
   Simulation studies are presented in Section \ref{sec_sim} while the real data analysis is shown in Section \ref{sec_real_data}.

    \paragraph{Notation.}
	For any set $S$, we write $|S|$ for its cardinality. For any positive integer $d$, we write $[d] = \{1,2,\ldots,d\}$.
	For any vector $v\in \RR^d$ and some real number $q\ge 0$, we define its $\ell_q$ norm as $\|v\|_q = (\sum_{j=1}^d |v_j|^q)^{1/q}$. For any matrix $M \in \RR^{d_1 \times d_2}$, $I \subseteq [d_1]$ and $J\subseteq [d_2]$, we write $M_{IJ}$ as the $|I| \times |J|$ submatrix of $M$ with row and column indices corresponding to $I$ and $J$, respectively. In particular, $M_{I\cdot}$ denotes the $|I|\times d_2$ submatrix and $M_J$ denotes the $d_1\times |J|$ submatrix. Further write $\|M\|_{p,q} = (\sum_{j=1}^{d_1} \|M_{j\cdot}\|_q^p)^{1/p}$ and denote by $\|M\|_{\op}$, $\|M\|_F$ and $\|M\|_\infty$, respectively, the operator norm, the Frobenius norm and the element-wise sup-norm of $M$. For any matrix $M$, we write $\lambda_{k}(M)$ for its $k$th largest singular value. We use $\bI_d$ to denote the $d\times d$ identity matrix and $\b0$ to denote the vectors with entries all equal to zero. We use $\be_1,\ldots,\be_d$ to denote the canonical basis in $\RR^d$. For any two sequences $a_n$ and $b_n$, we write $a_n \lesssim b_n$  if there exists some positive constant $C$ such that $a_n \le Cb_n$ for any $n$. We let $a_n \asymp b_n$ stand for $a_n\lesssim b_n$ and $b_n \lesssim a_n$. Denote $a\vee b=\max (a,b)$ and $a\wedge b=\min(a,b)$. 

\section{Identifiability of $\Theta$}\label{sec_id}

In this section, we establish conditions under which $\bTheta$ in model (\ref{model}) is identifiable when $Z$ is correlated with $X$ and the entries of $E$ are possibly correlated. 

Recall that model (\ref{model}) can be rewritten as (\ref{model_linear}). By regressing $Y$ onto $X$, one can identify 
\beq\label{def_F}
    \bF = \bTheta + \bA\bB. 
\eeq
The main challenge in identifying $\bTheta$ is that we need to further separate $\bTheta$ and $\bA\bB$ in the matrix $\bF$. 
The existing literature \citep{wang2017,Lee2017,McKennan19,bing2020adaptive} leverages the following decomposition of the residual covariance matrix of $\epsilon = \bB^TW + E$ from (\ref{model_linear})
\beq\label{def_Sigma_eps}
    \seps  = \bB^T \sw \bB + \se,
\eeq
to recover the row space of $\bB\in \RR^{K\times m}$. Here we write $\sw=\Cov(W)$ and $\se=\Cov(E)$. The decomposition (\ref{def_Sigma_eps}) is ensured by the independence assumption between $E$ and $W$. 
When $\se$ is diagonal and under suitable conditions on $\bB$ and $\sw$, the row space of $\bB$ can be identified from  (\ref{def_Sigma_eps}) either via PCA or the heteroscedastic PCA \citep{bing2020adaptive}, or via maximizing the quasi-likelihood under a factor model \citep{wang2017}. The recovered row space of $\bB$ is further used towards identifying $\bTheta$.

Our model differs from the existing literature in that we allow $\se$ to be non-diagonal, in which case the identifiability conditions in \cite{wang2017} and \cite{bing2020adaptive} are no longer applicable. For non-diagonal $\se$, we adopt the following conditions,
\begin{equation}\label{ident_conds}
    \lambda_K\left({1\over m}\bB^T \sw  \bB\right) \ge c,\qquad \|\se\|_{\op} = o(m), \qquad \textrm{as }m\to \i,
\end{equation}
where $c$ is a positive constant and $\lambda_K(M)$ denotes the $K$th largest eigenvalue of a symmetric matrix $M$. Under (\ref{ident_conds}), the space spanned by the first $K$ eigenvectors of $\seps$ recovers the row space of $\bB$ asymptotically as $m\to \i$.  This is an immediate result of the Davis-Kahan Theorem \citep{DavisKahan}, and has been widely used in the literature of factor models, see, for instance, \cite{fan2013large}. 

Given the row space of $\bB$, we can identify the projection matrices $P_B = \bB^T(\bB\bB^T)^{-1}\bB$ and $P_{B}^{\perp} = \bI_m - P_B$. Multiplying $P_{B}^{\perp}$ on both sides of equation (\ref{model}), we have
\begin{equation}\label{eq_pby}
    P_{B}^{\perp} Y = (\bTheta P_{B}^{\perp})^T X + P_{B}^{\perp}E,
\end{equation}
from which we recover $\bTheta P_{B}^{\perp}$ by 
\begin{equation}\label{def_Theta_PB_comp}
    \bTheta P_{B}^{\perp} = [\Cov(X)]^{-1}\Cov(X,  P_{B}^{\perp} Y).
\end{equation}
From $\bTheta P_{B}^{\perp}  = \bTheta - \bTheta P_{B}$, we have that $\bTheta$ can be recovered if $\bTheta P_B$ becomes negligible as $m\to \i$. Requiring $\bTheta P_B$ being small is common in the existing literature \citep{Lee2017,wang2017,bing2020adaptive}. We adopt the condition of assuming $\bTheta P_B$ small in terms of row-wise $\ell_1$ norm. The following theorem formally establishes the identifiability of $\bTheta$. As revealed in the proof of Theorem \ref{thm_ident}, $\|\bTheta_{i\cdot}\|_1=o(m)$ together with the other conditions therein ensures $(\bTheta P_B)_{ij} = o(1)$.


\begin{theorem}\label{thm_ident}
    Under model (\ref{model}), assume (\ref{ident_conds}) and 
    \beq\label{cond_ident_Theta}
        \max_{1\le j\le m}\bB_j^T\sw\bB_j = O(1),\qquad \max_{1\le i\le p}\|\bTheta_{i\cdot}\|_1 = o(m),\qquad \textrm{as }m\to \i.
    \eeq 
    Then $\bTheta$ can be recovered from the first two moments of $(X, Y)$ asymptotically as $m\to \i$.
\end{theorem}

The first requirement of (\ref{cond_ident_Theta}) is a regularity condition which holds, for instance, if $\sw\in\RR^{K\times K}$ has bounded eigenvalues and each column $\bB_j\in \RR^K$ of $\bB$ is bounded in $\ell_2$-norm. The second condition in (\ref{cond_ident_Theta}) requires the $\ell_1$-norm of each row of $\bTheta \in \RR^{p\times m}$ is of smaller order of $m$. This is the case if $\bTheta$ has bounded entries and each row of $\bTheta$ is sufficiently sparse. Such a sparsity assumption is reasonable in many applications, for instance, in genomics \citep{wang2017,McKennan19}.

\begin{remark}[Alternative identifiability conditions of $P_B$]\label{rem_ident_B}
   Condition (\ref{ident_conds}) assumes the spiked eigenvalue structure of $\seps$ in (\ref{def_Sigma_eps}) and is a common identifiability condition in the factor model when $m$ is large (see, \cite{fan2013large,Bai-factor-model-03}). We refer to Remark \ref{rem_cond_B} for more discussions on (\ref{ident_conds}). 
   Alternatively, another line of work studies the unique decomposition of the low rank and sparse decomposition under the so-called rank-sparsity incoherence conditions, \cite{Candes,Chandrasekaran,Hsu2011}, just to name a few. For instance, \citet[Theorem 1]{Hsu2011} showed that $\bB^T \sw \bB$ and $\se$ are identifiable from $\seps$ if 
   \beq\label{ident_cond_sparse_rank}
        \|\se\|_{\i,0} \|\bU_B\|_{\i,2}^2 \le c 
   \eeq
   for some small constant $0< c<1$. Here $\bU_B$ contains the right $K$ singular vectors of $\bB \in \RR^{K\times m}$. Once $\bB^T \sw \bB$ is identified, we can recover $P_B$ via PCA. Our identifiability results in Theorem \ref{thm_ident} still hold if (\ref{ident_conds}) is replaced by (\ref{ident_cond_sparse_rank}). 
\end{remark}

\begin{remark}[Other identifiability conditions of $\bTheta$]\label{rem_ident_cond_Theta}
    In the SVA literature, provided that $P_B$ is known, there are other sufficient conditions under which $\bTheta$ is identifiable. One type of such condition is called {\em negative controls} which assumes that, for a known set $S\subseteq[m]$ with $|S| \ge K$, 
        $$
        \bTheta_S = \b0\quad \text{and} \quad \text{rank}(\bB_{S}) =K.
        $$
   In words, there is a known set of responses that are not associated with any of the features in the multivariate response model (\ref{model}). Another condition considered in \cite{wang2017} requires the sparsity of $\bTheta$ in a similar spirit to (\ref{cond_ident_Theta}). It is assumed that, for some integer $K\le r\le m$, 
        \[
            \max_{j\in[p]}\left\|\bTheta_{j\cdot}\right\|_0 \le \floor{(m-r)/2},\qquad \rank(\bB_S) = K, \quad \forall~ S\subseteq [m] \textrm{ with }|S| = r.
        \]
Intuitively, the above condition also puts restrictions on the sparsity of $\bB$, as the submatrix of $\bB$ may have rank smaller than $K$ if $\bB$ is too sparse.
    Our identifiability results in Theorem \ref{thm_ident} still hold if condition (\ref{cond_ident_Theta}) is replaced by any of these conditions. 
\end{remark}

\section{Methodology}\label{sec_method}

In this section we describe our procedure of making inference on $\Theta_{ij}$ and $\bB_j$ for a given $i\in[p]$ and $j\in[m]$.  Recall that $(\bY_{i\cdot}, \bX_{i\cdot})$, for $1\le i\le n$, are i.i.d. copies of $(Y,X)$ from model (\ref{model}). Let $(\bY, \bX)$ denote the data matrix. 
For constructing confidence intervals of $\Theta_{ij}$ and testing the hypothesis (\ref{def_target}), our procedure consists of three main steps: (1) estimate the best linear predictor $\bX\bF$ in Section \ref{sec_pred} with $\bF$ defined in (\ref{def_F}), (2) estimate the residual $\bepsilon = \bY -\bX\bF$ and the matrix $\bB$ in Section \ref{sec_est_B}, (3) estimate $\bTheta_j$ and construct the final estimator of $\Theta_{ij}$ in Section \ref{sec_est_Theta}. Finally, we discuss how to make inference on $\bB_j$ in Section \ref{sec_method_infer_B}.

\subsection{Estimation of $XF$}\label{sec_pred}

Recall from (\ref{def_F}) that  $\bF$ has the additive decomposition of $\bTheta$ and $\bA\bB$. Estimating $\bF$ is challenging when the number of features $p$ exceeds the sample size $n$ without additional structure on $\bTheta$.  We thus consider the following parameter space of $\bTheta$
\begin{equation}\label{def_space_Theta}
     \cM(s_n, M_n) := 
    \left\{
        \bM\in\RR^{p\times m}:  \sum_{j=1}^p 1\{\|\bM_{j\cdot}\|_2\ne 0\} \le s_n,
        \max_{1\le j\le p} \norm{\bM_{j\cdot}}_1 \leq M_n
    \right\}
\end{equation}
for some integer $1\le s_n \le p$ and some sequence $M_n > 0$ that both possibly grow with $n$. As a result, any $\bTheta \in \cM(s_n, M_n)$  has at most $s_n$ non-zero rows and, for each of these non-zero rows, its $\ell_1$-norm is controlled by the sequence $M_n$. Existence of zero rows
is a popular sparsity structure in multivariate response regression \citep{yuanlin} and is also appealing for feature selection, while the structure of row-wise $\ell_1$ norm is needed in view of the identifiability condition (\ref{cond_ident_Theta}). 

Since the submatrix of $\bTheta \in\cM(s_n, M_n)$ corresponding to the non-zero rows may further have different sparsity patterns, we propose to estimate each column of $\bF$ separately. Specifically, we estimate $\bF$ by $\wh{\bF} = (\wh \bF_1,\dotso,\wh \bF_m) \in \RR^{p\times m}$ where, for each $j \in [m]$, $\wh\bF_j = \wh{\btheta}^{(j)} + \wh{\bdelta}^{(j)}$ is obtained by solving 
	\beq\label{def_est_F_j}
    \wh{\btheta}^{(j)},~\wh{\bdelta}^{(j)} = \argmin_{\btheta,\bdelta \in \RR^p}\frac{1}{n}\norm{\bY_{j} - \bX(\btheta + \bdelta)}_2^2 + \lambda_1^{(j)}\norm{\btheta}_1 + \lambda_2^{(j)}\norm{\bdelta}_2^2.
	\eeq
for some tuning parameters $\lambda_1^{(j)}, \lambda_2^{(j)} \ge 0$. Computationally, for any given $\lbdIj$ and $\lbdj$, solving (\ref{def_est_F_j}) is as efficient as solving a lasso problem (see, \cite{chernozhukov2017} or Lemma \ref{lem_solution} in Appendix \ref{app_theory_fit}). We discuss in details practical ways of selecting $\lambda_1^{(j)}$ and $\lambda_2^{(j)}$ in Section \ref{sec_cv}.

 Procedure (\ref{def_est_F_j}) is known as lava \citep{chernozhukov2017} and is designed to capture both the sparse signal $\bTheta_j$ and the dense signal $\bA\bB_j$ via respectively the lasso penalty and the ridge penalty. When columns of $\bTheta$ share the same sparsity pattern, \cite{bing2020adaptive} proposed a variant of (\ref{def_est_F_j}) to estimate $\bF$ jointly via the group lasso penalty together with the multivariate ridge penalty. To allow different sparsity patterns in columns of $\bTheta$ and, more importantly, to provide a sharp column-wise control of $\bX\wh \bF_j - \bX\bF_j$ for our subsequent inference on $\Theta_{ij}$, we opt for estimating $\bF$ column-by-column.



\subsection{Estimation of $B$}\label{sec_est_B}
    In this section, we discuss the estimation of $\bB$.  
    Our procedure first estimates the residual matrix $\bepsilon \coloneqq \bY - \bX\bF\in \RR^{n\times m}$ by 
    \begin{equation}\label{def_est_epsilon}
        \wh\bepsilon = \bY - \bX \wh \bF 
    \end{equation}
    with $\wh\bF$ obtained from (\ref{def_est_F_j}). To estimate $\bB$, notice that $\bepsilon = \bW \bB + \bE$ follows a factor model with $\bB$ being the loading matrix and $\bW$ being the latent factor matrix, should we observe $\bepsilon$. We therefore propose to estimate $\bB$ by the following approach commonly used in the factor analysis \citep{SW2002,Bai-factor-model-03,fan2013large} via the plug-in estimate $\wh\bepsilon$. Specifically, write the SVD of the normalized $\wh\bepsilon$ as
    \begin{equation}\label{def_svd_epsilon}
        {1 \over \sqrt{nm}}\wh\bepsilon ~ = ~\sum_{k=1}^{m} d_k \bu_k \bv_k^T,
    \end{equation}
    where $\bU_K = (\bu_1,\ldots, \bu_K)\in \RR^{n\times K}$ and $\bV_K = (\bv_1,\ldots,\bv_K)\in \RR^{m\times K}$ denote, respectively, the left and right singular vectors corresponding to $d_1 \ge d_2 \ge \cdots \ge d_K$. Further write $\bD_K = \diag(d_1,\ldots, d_K)$. We propose to estimate $\bB$ and $\bW$ by 
    \begin{align*}
        &(\wh\bB, \wh\bW) = \argmin_{\bB,\bW} {1\over nm}\left\|
         \wh\bepsilon - \bW\bB 
        \right\|_F^2,\\\nonumber
        &\textrm{subject to}\quad {1\over n}\bW^T\bW = \bI_K,\quad {1\over m}\bB\bB^T\textrm{ is diagonal}.
    \end{align*}
    It is well known (see, for instance, \cite{Bai-factor-model-03}) that the above problem leads to the following solution
    \begin{equation}\label{def_est_BW}
        \wh\bB^T = \sqrt{m}~ \bV_K\bD_K,\qquad \wh\bW = \sqrt{n}~ \bU_K.
    \end{equation}
    We assume $K$ is known for now and defer its selection to Section \ref{sec_K}. 

\subsection{Estimation and inference of $\Theta$}\label{sec_est_Theta}

Without loss of generality, we let $\Theta_{11}$ be the parameter of our interest.
To make inference of $\Theta_{11}$, we first construct an initial estimator of $\bTheta_1\in\RR^p$ via $\ell_1$ regularization after removing the hidden effects, and then obtain our final estimator of $\Theta_{11}$ by removing the bias due to the $\ell_1$-regularization in the first step. For this reason, our final estimator of $\Theta_{11}$ is doubly debiased. 

Write $\wt \by = \bY\wh{P}_B^\perp \be_1$ with $\wh P_B^{\perp} := \bI_m - \wh\bB^T(\wh\bB\wh\bB^T)^{-1}\wh\bB = \bI_m - \bV_K\bV_K^T$ from (\ref{def_est_BW}). 
In view of (\ref{eq_pby}), we propose to use the solution of the following lasso problem as the initial estimator of $\bTheta_1$,
	\beq\label{def_Thetaj_hat}
	\wh{\bTheta}_1 = \argmin_{\btheta\in\RR^p}\frac{1}{n}\bnorm{\wt\by - \bX\btheta}_2^2 + \lambda_3\norm{\btheta}_1.
	\eeq 
Here $\lambda_3\ge 0$ is some tuning parameter. As seen in (\ref{eq_pby}), using the projected response $\wt \by = \bY\wh{P}_B^\perp \be_1$ in the above lasso problem removes the bias due to the hidden variables. 

While the $\ell_1$-regularization reduces the variance of the resulting estimator, it  introduces extra bias that needs to be adjusted in order to further make inference of $\Theta_{11}$. To reduce this bias due to the $\ell_1$ regularization, our final estimator of $\Theta_{11}$ is proposed as follows,
	\beq\label{def_Theta_11}
	\wt{\Theta}_{11} = \wh{\Theta}_{11} + \wh{\bomega}_1^T\frac{1}{n}\bX^T(\wt\by - \bX\wh{\bTheta}_1)
	\eeq
where $\wh\bomega_1 \in \RR^p$ is the estimate of the first column $\bOmega_1$ of $\bOmega \coloneqq \Sigma^{-1}$ with $\Sigma=\Cov(X)$. There are several ways of estimating $\bOmega_1$, for instance, \cite{zhangzhang2014,javanmard14,vandegeer2014}. In this paper, we follow the node-wise lasso procedure in \cite{zhangzhang2014} and \cite{vandegeer2014} to obtain $\wh{\bomega}_1$. Specifically, let 
\beq\label{formula_nodewise}
    \wh\bgamma_1 =  \argmin_{\bgamma \in \RR^{p-1}} {1\over n}\left\|
        \bX_1 - \bX_{-1}\bgamma 
    \right\|_2^2 + \wt \lambda\|\bgamma\|_1
\eeq
for some tuning parameter $\wt \lambda\ge 0$, where $\bX_{-1}\in\RR^{n\times(p-1)}$ is the submatrix of $\bX$ with the first column removed. We write 
\beq\label{def_tau_1}
    \wh\tau_1^2 = {1\over n}\bX_1^T(\bX_1 - \bX_{-1}\wh\bgamma_1)
\eeq
and define 
\beq\label{def_est_omega}
    \wh\bomega_1^T = {1\over \wh\tau_1^2}\begin{bmatrix}
     1 & -\wh\bgamma_1^T
    \end{bmatrix},
\eeq
as the estimator of $\bOmega_1$. In Theorem \ref{thm_asymp_normal} of Section \ref{sec_theory_ASN}, we show that, conditioning on the design matrix, $\sqrt{n}(\wt{\Theta}_{11} - \Theta_{11})$ is asymptotically normal with mean zero and variance $\sigma_{E_1}^2 \wh{\bomega}_1^T\wh\Sigma\wh{\bomega}_1$, where $\sigma_{E_1}^2:=[\se]_{11}$ and $\wh\Sigma = n^{-1}\bX^T\bX$.

In light of  this result, we can test the hypothesis $H_{0,\Theta_{11}}: \Theta_{11} = 0$ versus $H_{1,\Theta_{11}}: \Theta_{11} \neq 0$, via the following test statistic
\beq\label{def_U_hat}
\wh U_n^{(11)} = \sqrt{n}~\wt\Theta_{11}/\sqrt{\wh\sigma_{E_1}^2 \wh{\bomega}_1^T\wh\Sigma\wh{\bomega}_1},
\eeq
with $\wh\sigma_{E_1}^2$ being an estimator of $\sigma_{E_1}^2$, defined as 
\begin{equation}\label{def_est_variance}
    \wh \sigma_{E_1}^2 = {1\over n}(\wh\bepsilon_1-\wh\bW\wh\bB_1)^T(\wh\bepsilon_1-\wh\bW\wh\bB_1)
\end{equation}
with $\wh\bepsilon$, $\wh\bB$ and $\wh\bW$ obtained from (\ref{def_est_epsilon}) and (\ref{def_est_BW}).
For any given significance level $\alpha\in(0,1)$, we reject the null hypothesis if $|\wh U_n^{(11)}| > k_{\alpha/2}$, where $k_{\alpha/2}$ is the $(1-\alpha/2)$ quantile of $N(0,1)$. Equivalently, we can also construct a $(1 - \alpha)\times 100\%$ confidence interval for $\Theta_{11}$ as
    \beq \label{CI_def}
    \left(\wt{\Theta}_{11} - k_{\alpha/2}\sqrt{\wh\sigma_{E_1}^2\wh{\bomega}_1^T\wh\Sigma\wh{\bomega}_1 / n},~~
    \wt{\Theta}_{11} + k_{\alpha/2}\sqrt{\wh\sigma_{E_1}^2\wh{\bomega}_1^T\wh\Sigma\wh{\bomega}_1 / n}
    \right).
    \eeq

\subsection{Hypothesis testing of the hidden effect}\label{sec_method_infer_B}
 
In practice, it is also of interest to test whether or not some response $Y_j$, for $1\le j\le m$, is affected by any of the hidden variables $Z$. If the effect of the hidden variables $Z$ is indeed significant, ignoring the hidden variables in the regression analysis may yield biased estimators and incorrect conclusion. In this case, the use of our hidden variable model (\ref{model}) is strongly preferred, as adjusting the hidden effects for modelling $Y_j$ is critical. 

Without loss of generality, we take $j = 1$. The hypothesis testing problem (\ref{def_target_B}) becomes $H_{0,B_1}: \bB_{1} = \b0$ versus $H_{1,B_1}: \bB_1 \ne \b0$. We propose to use the following test statistic
    \begin{equation}\label{def_R_hat}
        \wh R_{n}^{(1)} = n \wh \bB_1^T \wh \bB_1 / \wh\sigma_{E_1}^2
    \end{equation}
with $\wh\bB$  and $\wh\sigma_{E_1}^2$  obtained from (\ref{def_est_BW}) and (\ref{def_est_variance}), respectively. While $\wh \bB$ depends on the regularized estimator lava in (\ref{def_est_F_j}) via the estimated residuals, an interesting phenomenon is that there is no need to further debias the estimator $\wh \bB$ for inference. In Theorem \ref{thm_B_asn}, we show that the estimator $\wh \bB_j$ is asymptotically normal and the test statistic $\wh R_{n}^{(1)}$ converges in distribution to the $\chi^2$ distribution with degrees of freedom equal to $K$ under the null. Thus, given any significance level $\alpha\in(0,1)$, we reject the null hypothesis if $\wh R_{n}^{(1)} > c_{\alpha}$, where $c_{\alpha}$ is the $(1-\alpha)$ quantile of the $\chi^2$ distribution with degrees of freedom equal to $K$.

\section{Theoretical analysis}\label{sec_theory}

In this section, we provide theoretical guarantees for our procedure in Section \ref{sec_method}. Section \ref{sec_ass} contains our main assumptions. The asymptotic normality of $\wt\Theta_{11}$ is established in Section \ref{sec_theory_ASN} while its efficiency and the consistent estimation of its asymptotic variance are discussed in Section \ref{sec_effciency}. The statistical guarantees for  $\wh\bB$ are shown in Sections \ref{sec_theory_B}.

\subsection{Assumptions}\label{sec_ass}
Throughout our analysis, we assume that $m$ and $p$ both grow with $n$ and the number of hidden variables, $K$, is fixed. Our analysis can be extended to the case where $K$ grows with $n$ coupled with more involved conditions. We start with the following blanket distributional assumptions on $W$ and $E$.
\begin{assumption}\label{ass_error}
	Let $\gamma_w$ and $\gamma_e$ denote some finite positive  constants.
	 Assume $\Sigma_W^{-1/2}W$ is a $\g_w$ sub-Gaussian random vector \footnote{A centered random vector $X\in \RR^d$ is $\g$ sub-Gaussian if $
		\EE[\exp(\langle u, X\rangle)] \le \exp(\|u\|_2^2\g^2/2)
		$ for any $u\in\RR^d$.} with $\Sigma_W={\rm Cov}(W)$.
	Assume $\se^{-1/2}E$ is a $\g_e$ sub-Gaussian random vector with $\se = {\rm Cov}(E)$. 
\end{assumption}

Our analysis requires the following regularity conditions on $\bB$, $\sw$ and $\se$. 
    \begin{assumption}\label{ass_B_Sigma}
        Assume there exist some positive finite constants $c_W\le C_W$, $c_B\le C_B$, $C_E$ and $c_\epsilon$  such that 
        \begin{enumerate}
            \item[(a)] $c_W\le \lambda_K(\sw) \le \lambda_1(\sw) \le C_W$;
            \item[(b)] $\max_{1\le j\le m}\|\bB_j\|_2^2\le C_B$, $\lambda_K(\bB\bB^T) \ge c_B m$;
            \item[(c)] $\lambda_1(\se) \le C_E$;
            \item[(d)] $\min_{1\le j\le m} \left(\bB_j^T \sw \bB_j + [\se]_{jj}\right) \ge c_\epsilon$.
        \end{enumerate}
    \end{assumption}

    \begin{remark}\label{rem_cond_B}
        Assumption \ref{ass_B_Sigma} is slightly stronger than the identifiability condition (\ref{ident_conds}) and the first condition in (\ref{cond_ident_Theta}). They are all commonly used regularity conditions in the literature of factor analysis \citep{Bai-Ng-K,Bai-factor-model-03,SW2002,Bai-Ng-forecast,fan2013large,Ahn-2013,fan2017} as well as in the related SVA literature \citep{Lee2017,wang2017}. In particular, condition $\lambda_K(\bB\bB^T) \ge c_B m$ is known as the pervasive assumption which holds, for instance, if a (small) proportion of columns of $\bB$ are i.i.d. realizations of a $K$-dimensional sub-Gaussian random vector whose covariance matrix has bounded eigenvalues \citep{guo2020doubly}. 
    \end{remark}

    We also need conditions on the design matrix $\bX$. Recall that $s_n$ is defined in (\ref{def_space_Theta}).
	
    \begin{assumption}\label{ass_X}
        Assume the rows of $\bX$ are i.i.d. realizations of the random vector $X\in \RR^p$ with $\Sigma:={\rm Cov}(X)$ satisfying 
        $$
            \max_{1\le j\le p}\Sigma_{jj} \le C,\qquad c \le \lambda_{\min}(\Sigma) \le \sup_{S\subseteq[p]: |S|\le s_n}\lambda_{\max}(\Sigma_{SS})\le C
        $$ for some absolute constants $0<c<C<\i$. Further assume $X\sim N_p(0, \Sigma)$. 
    \end{assumption}
    
    Assumption \ref{ass_X} is borrowed from \cite{vandegeer2014} to analyze the theoretical properties of $\wh\bomega_1$ via the node-wise lasso approach in (\ref{def_est_omega}). As commented there, the Gaussianity in Assumption \ref{ass_X} is not essential and can be relaxed to that $X$ is a sub-Gaussian or bounded random vector.

    Since our whole inference procedure for $\Theta_{11}$ starts with the estimation of $\bX\bF$ from (\ref{def_est_F_j}), the estimation error of $\bX\wh\bF$ plays a critical role throughout our analysis. While upper bounds of the rate of convergence of $\|\bX\wh\bF_j - \bX\bF_j\|_2$ have been established in \cite{chernozhukov2017}, we provide  a uniform bound in Appendix \ref{app_theory_fit} by showing that, with probability tending to one, the following holds uniformly over $j\in[m]$, 
    \begin{equation}\label{def_Rem_j}
            {1\over n}\|\bX\wh\bF_j - \bX\bF_j\|_2^2 ~ \lesssim ~  Rem_{1,j} + Rem_{2,j}(\bdelta_j) + Rem_{3,j}(\btheta_j). 
    \end{equation}
    Here we write $\bF_j = \btheta_j + \bdelta_j$ with $\btheta_j \coloneqq \bTheta_j$ and $\bdelta_j \coloneqq \bA\bB_j$. The terms $Rem_{1,j}$, $Rem_{2,j}(\bdelta_j)$ and $Rem_{3,j}(\btheta_j)$ all depend on the design matrix $\bX$ and their exact expressions are stated in Appendix \ref{app_theory_fit}. 
    For ease of presentation, we resort to a deterministic upper bound of the right hand side of (\ref{def_Rem_j}).
    
    \begin{assumption}\label{ass_initial}
        There exists a positive (deterministic) sequence $r_{n} = o(1)$ such that with probability tending to one as $n\to\i$, 
        \[
            \max_{1\le j\le m}\Bigl[Rem_{1,j} + Rem_{2,j}(\bdelta_j) + Rem_{3,j}(\btheta_j)\Bigr]\le  r_{n}.
        \]
    \end{assumption}
    Our subsequent theoretical results naturally depend on $r_{n}$, for which we provide the explicit rate later in Corollary \ref{cor_ASN} of Section \ref{sec_theory_ASN}.  Notice that Assumption \ref{ass_initial} together with (\ref{def_Rem_j}) readily implies 
    \[
        \lim_{n\to \i}\PP\left\{
        \max_{1\le j\le m}{1\over n}\|\bX\wh \bF_j - \bX\bF_j\|_2^2 \lesssim r_n
        \right\} = 1.
    \]

    \subsection{Asymptotic normality of $\wt\Theta_{11}$}\label{sec_theory_ASN}

    In this section, we establish our main result: the asymptotic normality of our estimator $\wt \Theta_{11}$ from (\ref{def_Theta_11}). To this end, we first study the convergence rate of the initial estimator $\wh\bTheta_1$ defined in (\ref{def_Thetaj_hat}). Recall from (\ref{def_Theta_PB_comp}) that the estimand of $\wh\bTheta_1$ is   $\bar\bTheta_1 :=  \bTheta P_{B}^{\perp} \be_1$ which satisfies 
    \[
        \|\bar\bTheta_1\|_0 = \|\bTheta P_{B}^{\perp} \be_1\|_0  \le s_n,
    \]
    implied by (\ref{def_space_Theta}).
    The following lemma states the $\ell_1$ convergence rate of $\wh\bTheta_1 - \bar\bTheta_1$, whose proof can be found in Appendix \ref{app_proof_thm_Theta}. Recall that $M_n$ is defined in (\ref{def_space_Theta}) and $r_{n}$ is defined in  Assumption \ref{ass_initial}.

	\begin{lemma}\label{thm_Theta_simple_rates}
	    Under Assumptions \ref{ass_error} -- \ref{ass_initial}, assume $M_n = o(m)$, $\|{\rm Cov}(Z)\|_{\rm op} = \cO(1)$, $\log m = o(n)$ and $s_n\log p = o(n)$. By choosing 
	    $$
	        \lambda_3 \gtrsim \sqrt{\max_{1\le j\le p}\wh\Sigma_{jj}}\sqrt{\log p\over n}
	   $$ in (\ref{def_Thetaj_hat}), with probability tending to one as $n\to \infty$,
	    \beq\label{rate_Theta_td_simp}
            \|\wh \bTheta_1 -  \bar\bTheta_1\|_{1} & ~ \lesssim ~ s_n\sqrt{\log p\over n} + \left({s_nM_n\over m} + \sqrt{s_n}\right)\left(\sqrt{\log m \over n\wedge m} + r_n\right).
        \eeq
	\end{lemma}
	Condition $M_n = o(m)$ is needed here to ensure that $\bTheta$ is identifiable (see, Section \ref{sec_id}). It can be replaced by any other identifiability conditions in Remark \ref{rem_ident_cond_Theta}. Recall that $Z\in \RR^K$ and $K$ is fixed, $\|\Cov(Z)\|_{\op} = \cO(1)$ is a mild regularity condition. The requirement $s_n\log p = o(n)$ is also mild as we explained below. 
	
	The first term on the right hand side of (\ref{rate_Theta_td_simp}) is known as the optimal rate of estimating a $s_n$-sparse coefficient vector in standard linear regression. Therefore,  $s_n\sqrt{\log p} = o(\sqrt n)$ is the minimal requirement for consistently estimating $\bar\bTheta_1$ in $\ell_1$-norm.
	The second term  stems from the error of estimating $P_{B}$, or in fact, of estimating $\bB$ (see, Theorem \ref{thm_rates_B} in Section \ref{sec_theory_B}). For instance, when $\bX\bF$ can be estimated with a fast rate, that is, $r_n$ is sufficiently small, then (\ref{rate_Theta_td_simp}) can be simplified to
   \[
        \|\wh \bTheta_1 -  \bar\bTheta_1\|_{1}  \lesssim  s_n\sqrt{\log p\over n} + {s_nM_n\over m}\sqrt{\log m \over n\wedge m} + \sqrt{s_n\log m \over n\wedge m}.
   \]
   The above rate becomes faster as $m$ increases. In particular, when $n = \cO(m)$, we recover the optimal rate (up to a multiplicative logarithmic factor)
   \[
    \|\wh \bTheta_1 -  \bar\bTheta_1\|_{1}  = \cO_\PP\left(s_n\sqrt{\log (p\vee m)\over n} \right).
   \]

    Armed with the guarantees of the initial estimator $\wh\bTheta_1$,  our following main result shows that $\sqrt{n}(\wt \Theta_{11} - \Theta_{11})$ is asymptotically normal with
    a closed-form expression of the asymptotic variance. Its proof can be found in Appendix \ref{app_thm_asymp_normal}. 
    Recall that $\bOmega = \Sigma^{-1}$ is the precision matrix of $X$. Since $\wt\Theta_{11}$ depends on the estimate of $\bOmega_1\in\RR^p$, our analysis requires  $\bOmega_1$ to be  sparse. Let $s_\Omega = \|\bOmega_1\|_0$ denote the sparsity of $\bOmega_1$.

    \begin{theorem}\label{thm_asymp_normal}
        Under Assumptions \ref{ass_error} --  \ref{ass_initial}, assume $E_1\sim N(0,\sigma_{E_1}^2)$, $\|{\rm Cov}(Z)\|_{\rm op} = \cO(1)$, $ (s_n\vee s_\Omega)\log(p)\log(m) = o(n)$ and $s_n \log p = o(\sqrt n)$. 
        Further assume
        \begin{align}\label{cond_Mn}
        & M_n \sqrt{n} = o(m),\\\label{cond_rn}
        &\norm{\bA_{1\cdot}}_2 \sqrt{\log m}+\left(\|\bA_{1\cdot}\|_2 \sqrt n + \sqrt{(s_n\vee s_\Omega)\log p}\right)r_n  = o(1).
        \end{align}
        By choosing $\wt \lambda \asymp \sqrt{\log p/n}$ in (\ref{def_est_omega}), one has
        \[\sqrt{n}(\wt{\Theta}_{11} - \Theta_{11}) = \zeta + \Delta,\]
            where \[
            \zeta \mid \bX \sim N(0, \sigma_{E_1}^2 \wh{\bomega}_1^T\wh\Sigma\wh{\bomega}_1),\qquad |\wh{\bomega}_1^T\wh\Sigma\wh{\bomega}_1 - \Omega_{11}| = o_{\PP}(1),\qquad \Delta = o_{\PP}(1).\]
    \end{theorem}

    Theorem \ref{thm_asymp_normal} shows that the difference between $\wt\Theta_{11}$ and $\Theta_{11}$ scaled by $\sqrt{n}$ is decomposed into two terms, $\zeta$ and $\Delta$, where, conditioning on $\bX$, $\zeta$ follows a Gaussian distribution with zero mean and variance $\sigma_{E_1}^2 \wh{\bomega}_1^T\wh\Sigma\wh{\bomega}_1$, and $\Delta$ is asymptotically negligible. Indeed, $\Delta = o_{\PP}(1)$ holds uniformly over $\bTheta \in \cM(s_n,M_n)$ in (\ref{def_space_Theta}), so that we can use Theorem \ref{thm_asymp_normal} to construct honest confidence intervals for $\Theta_{11}$, as long as $\sigma_{E_1}^2$ can be consistently estimated.


    \begin{remark}[Discussions of conditions in Theorem \ref{thm_asymp_normal}]
    The Gaussianity assumption of $E_1$ is not essential. In fact, our proof states that 
    $\zeta = \wh\bomega_1^T\bX^T\bE_1 / \sqrt n$. Therefore, when $E_1$ is not Gaussian,  one can still obtain $\sqrt{n}(\wt{\Theta}_{11} - \Theta_{11}) \mid \bX\to_d N(0,\sigma_{E_1}^2 \wh{\bomega}_1^T\wh\Sigma\wh{\bomega}_1)$ provided that the Lindeberg's condition for the central limit theorem holds.

    The condition $s_{\Omega}\log p = o(n)$ ensures the consistency of the node-wise Lasso estimator $\wh\bomega_1$, see \cite{vandegeer2014}. We require an extra logarithmic factor of $m$ here due to the union bounds over $j\in[m]$ for estimating $\bX\bF_j$. 
    Condition $s_n \log p = o(\sqrt n)$ puts restriction on the number of non-zero rows in $\bTheta$. It is a rather standard condition for making inference of the coefficient in high-dimensional regressions \citep{javanmard14,vandegeer2014,zhangzhang2014}. As discussed after Lemma \ref{thm_Theta_simple_rates}, it is also the minimum requirement for consistently estimating $\bar\bTheta_1$ in $\ell_1$-norm.  
    
    Condition (\ref{cond_Mn}) is concerned with the magnitude of each row of $\bTheta$ in $\ell_1$ norm and is a strengthened version of the identifiability condition (\ref{cond_ident_Theta}). Recall that the estimand of the initial estimator $\wh\bTheta_1$ is   $\bar\bTheta_1 :=  \bTheta P_{B}^{\perp} \be_1$ rather than $\bTheta_1$. The condition is used to ensure that the bias term for estimating $\Theta_{11}$, defined as  $\Theta_{11}-\bar\Theta_{11}=\bTheta_{1\cdot}^T P_B\be_1$, is asymptotically negligible.  Condition (\ref{cond_Mn})  holds, for instance, when the rows of $\bTheta$ are sufficiently sparse and the order of $m$ is comparable or larger than $n$, see \cite{McKennan19,wang2017}. 

    Finally, condition (\ref{cond_rn}) puts restriction on the $\ell_2$ norm of $\bA_{1\cdot}$ as well as on the order of $r_n$. To aid intuition of this condition, we provide explicit rates of $r_n$ under two common scenarios in the high-dimensional setting. As seen in Corollary \ref{cor_ASN} below, the requirement of $r_n$ again hinges on the magnitude of $\bA$ which quantifies the correlation between the observable feature $X$ and the hidden variable $Z$. 
    We refer to Remark \ref{rem_A} for detailed discussions of conditions on $\bA$.
    \end{remark}

    The following corollary provides explicit rates of $r_n$ under two common scenarios in the high-dimensional settings, depending on the magnitude of $\|\Sigma\|_{\op}$.
    
    \begin{corollary}\label{cor_ASN}
Assume that Assumptions \ref{ass_error} --  \ref{ass_X} hold. 
    \begin{enumerate}
        \item[(1)] Suppose $p>n$ and $\|\Sigma\|_{\rm op} = \cO(1)$. Assume $(s_n\vee s_\Omega)\log^2(p\vee m) = o(n)$, 
        \beq\label{cond_A_op}
            \|\bA\|_{\rm op}^2 = o\left({1\over \sqrt{(s_n\vee s_\Omega)\log p}}\right)
        \eeq
        and $\|\bA_{1\cdot}\|_2 = o(\sqrt{(s_n\vee s_\Omega)\log p/ n})$. 
        Then Assumption \ref{ass_initial} holds with 
        \beq\label{rate_rnj_case1}
            r_{n} = \cO\left(\|\bA\|_{\op}^2 + {s_n \log (p\vee m)\over n}\right),\quad \forall\ 1\le j\le m
        \eeq
        and  condition (\ref{cond_rn}) holds. 
        
        \item[(2)] Suppose $p>n$, $\|\Sigma\|_{\op}\asymp p$ and $\tr(\Sigma)=\cO(p)$. Assume $s_n(s_n\vee s_\Omega) \log^2(p\vee m) = o(n)$ and 
        $
            \|\bA\|_{\op}^2 = \cO(1/p).
        $
        Then Assumption  \ref{ass_initial} holds with 
        \[
            r_{n} = \cO\left(\sqrt{s_n\log (p\vee m) \over n}\right).
        \]
        Furthermore, condition (\ref{cond_rn}) holds as well. 
    \end{enumerate}
\end{corollary}


\begin{remark}[Discussions of conditions on $\bA$]\label{rem_A}
    We first explain why restriction on the magnitude of $\bA$ is necessary in the high-dimensional regime ($p > n$).
     For any $j\in [m]$, recall that $\|\bA\bB_j\|_2^2 = \|\bdelta_j\|_2^2$ and  consider the regression $\bY_j = \bX\bdelta_j + \bepsilon_j$ with $\btheta_j = \b0$. Even in this simplified scenario, since $\bdelta_j$ is a dense $p$-dimensional vector, its consistent estimation 
    requires  $\|\bdelta_j\|_2=o(1)$ when $p$ is larger than $n$ \citep{Hsu2014,chernozhukov2017,cevid2018spectral}. Therefore, one would expect that 
    $\|\bdelta_j\|_2^2 = o(1)$ is necessary for consistent estimation of $\bX\bF_j$ for each $1\le j\le m$. The uniform bound over $1\le j\le m$, together with $\lambda_K(\bB)\gtrsim \sqrt{m}$, in turn implies
    \beq\label{bd_A_op_minimax}
        \|\bA\|_{\op}^2 = o(1).
    \eeq
    Therefore, consistent estimation of $\bX\bF$ in high-dimensional scenario necessarily requires small $\|\bA\|_{\op}^2$. Recall that $\bA = \Sigma^{-1}\Cov(X,Z)$ with $\Sigma = \Cov(X)$. A small $\|\bA\|_{\op}^2$ means either (a) the observable feature $X$ and the hidden variable $Z$ are weakly correlated, or (b) $\Sigma$ has spiked eigenvalues. We comment on these two cases separately below. 

    Scenario (1) of Corollary \ref{cor_ASN} corresponds to (a). 
    When there is a finite number of observable feature $X$ correlated with the hidden variable $Z$, we have 
    $\|\bA\|_{\op}^2 = \cO(\rho)$ where $\rho=\max_{1\le j\le m,1\le k\le K}\textrm{Corr}(X_j, Z_k)$. Condition (\ref{cond_A_op}) holds if $\rho = o(1/\sqrt{(s_n\vee s_\Omega)\log p})$. 
    In addition, $\|\bA_{1\cdot}\|_2 = o(\sqrt{(s_n\vee s_\Omega)\log p/ n})$ holds, for instance, when either the rows of $\bA$ are balanced in the sense that $\|\bA_{1\cdot}\|_2 = \cO(\|\bA\|_{\op}/\sqrt{p})$ or $\max_{1\le k\le K}\textrm{Corr}(X_1, Z_k) = o(\sqrt{(s_n\vee s_\Omega)\log p/n})$.
    
    Scenario (2) of Corollary \ref{cor_ASN} corresponds to (b) where $\Sigma$ has a fixed number of spiked eigenvalues. One instance is when $X$ follows from a factor model $X = \bGamma F + W'$ where $F\in\RR^r$ is the factor and the loading matrix $\bGamma\in\RR^{p\times r}$ satisfies $\lambda_r(\bGamma) \gtrsim \sqrt p$ with $r < p$. \citet[Section 3.4]{bing2020adaptive} provides examples of this model under which $\|\Sigma\|_{\op} = \cO(p)$, $\tr(\Sigma)= \cO(p)$ and $\|\bA\|_{\op}^2 = \cO(1/p).$
\end{remark}

    \subsection{Efficiency and consistent estimation of the asymptotic variance}\label{sec_effciency}
    
    From Theorem \ref{thm_asymp_normal}, our estimator $\wt\Theta_{11}$ has the asymptotic variance $\sigma_{E_1}^2\Omega_{11}/n$, which, according to the Gauss-Markov theorem, is the same asymptotic variance of the best linear unbiased  estimator (BLUE) of $\Theta_{11}$ in the classical low-dimensional setting without any hidden variables. Therefore, our estimator $\wt \Theta_{11}$ is efficient in this Gauss-Markov sense. In fact, even when there exist hidden variables $Z$, $\sigma_{E_1}^2\Omega_{11}/n$ is also the minimal variance of all unbiased estimators in the low-dimensional setting. Indeed, when $Z$ is observable, the Gauss-Markov theorem states that the oracle BLUE of $\Theta_{11}$ has the asymptotic variance 
    \[
        {\sigma_{E_1}^2\over n}
        ~ \be_1^T \begin{bmatrix}
           \Sigma & \Cov(X,Z) \\ 
           \Cov(Z,X) & \Cov(Z)
        \end{bmatrix}^{-1} \be_1 = {\sigma_{E_1}^2\over n}\left(
        \Omega_{11} + \bA_{1\cdot}^T \sw^{-1}\bA_{1\cdot}
        \right).
    \]
    Here the equality uses the  block matrix inversion formula, the definition $\bA = \Sigma^{-1}\Cov(X,Z)$ and $\sw = \Cov(Z) - \Cov(Z,X)\Sigma^{-1} \Cov(X,Z)$. Comparing to $\sigma_{E_1}^2\Omega_{11} / n$, the term $\bA_{1\cdot}^T \sw^{-1} \bA_{1\cdot}$ represents the efficiency loss due to the hidden variables. However, in the high-dimensional setting with $\|\bA_{1\cdot}\|_2 = o(1)$ (together with $\Omega_{11} \ge c$ and $\lambda_K(\sw) \ge c_W$), this efficiency loss becomes negligible and the asymptotic variance in the above display reduces to $\sigma_{E_1}^2\Omega_{11} / n$. 
    
    In the high-dimensional regime, if one treats model (\ref{model}) as a semi-parametric model $Y_1 = \Theta_{11}X_1 + G(X_{-1}, Z) + E_1$ for some unknown function $G:\RR^{p-1}\times \RR^K \to \RR$ with $Z$ being observable, our estimator $\wt\Theta_{11}$ of $\Theta_{11}$ is semi-parametric efficient according to Theorem 2.3 and Lemma 2.1 in \cite{vandegeer2014}.\\

    Our proposed test statistic in (\ref{def_U_hat}) and confidence intervals  in (\ref{CI_def}) require to estimate $\sigma_{E_1}^2$. 
    The following proposition ensures that the proposed estimator $\wh\sigma_{E_1}^2$ in (\ref{def_est_variance}) is consistent. Consequently, an application of the Slutsky's theorem coupled with Theorem \ref{thm_asymp_normal} justifies the validity of our test statistic and confidence intervals in Section \ref{sec_est_Theta}.
    
    \begin{prop}\label{prop_sigma_E}
        Under conditions of Theorem \ref{thm_asymp_normal}, $\wh \sigma_{E_1}^2$ defined in (\ref{def_est_variance}) satisfies 
        \[
                |\wh \sigma_{E_1}^2 - \sigma_{E_1}^2| = o_\PP(1).
        \]  
    \end{prop}

    \subsection{Rate of convergence and asymptotic normality of $\wh B$}\label{sec_theory_B}
    
    Towards establishing the theoretical guarantees of $\wt\Theta_{11}$ in the previous section, one intermediate, but important, step is to sharply characterize the error of estimating $P_B$, or equivalently, $\bB$. In this section, we first present the convergence rate of our estimator $\wh\bB$ in (\ref{def_est_BW}). Then, we establish the asymptotic normality of $\wh\bB$ to test the hypothesis (\ref{def_target_B}).
    
    First notice that, without further restrictions, $\bW$ and $\bB$ are not identifiable even one has direct access to $\bepsilon = \bW\bB+\bE$. This can be seen by constructing $\bW' = \bW Q$ and $\bB' = Q^{-1}\bB$ for any invertible matrix $Q\in \RR^{K\times K}$ such that $\bW\bB = \bW'\bB'$. To quantify the estimation error of $\wh\bB$, we introduce the following rotation matrix \citep{bai2020simpler},
    \begin{equation}\label{def_H0}
        \bH_0^T = {1\over nm}\bW^T\bW \bB \wh\bB^T \bD_K^{-2} \in \RR^{K\times K}
    \end{equation}
    with $\bD_K$ defined in (\ref{def_svd_epsilon})\footnote{If $\bD_K$ is not invertible, we use its Moore-Penrose inverse instead.}.
    Further define 
    \begin{equation}\label{def_B_tilde}
        \wt \bB = \bH_0 \bB \in \RR^{K\times m}.
    \end{equation}
    Since $\wt \bB = (nm)^{-1}\bD_K^{-2}\wh\bB (\bB^T\bW^T\bW\bB)$ only depends on the data and the identifiable quantity $\bB^T\bW^T\bW\bB$, $\wt \bB$ is well-defined.

    The following theorem provides the uniform $\ell_2$ convergence rate of $\wh\bB_j - \wt\bB_j$ over $1\le j\le m$. Recall that $M_n$ is defined in (\ref{def_space_Theta}) and $r_{n}$ is defined in  Assumption \ref{ass_initial}.
    
    \begin{theorem}\label{thm_rates_B}
        Under Assumptions \ref{ass_error}, \ref{ass_B_Sigma},  \ref{ass_initial} and $M_n = o(m)$, with probability tending to one as $n\to \infty$,
        one has
        \begin{equation}\label{eq_thm_rates_B}
            \max_{1\le j\le m}\|\wh\bB_j - \wt\bB_j\|_2 \lesssim \sqrt{\log m\over n\wedge m} + r_{n}.            
        \end{equation}
    \end{theorem}

    The first term on the right hand side of (\ref{eq_thm_rates_B}) is the error rate of estimating $\bB$ when $\bepsilon=\bY - \bX\bF$ is known, while the second term corresponds to the error of estimating $\bepsilon$ by $\wh\bepsilon = \bY-\bX\wh\bF$.
    If $\bepsilon = \bW\bB + \bE \in \RR^{n\times m}$ were observed, theoretical guarantees of $\wh\bB$ and $\wh\bW$ from (\ref{def_est_BW}) for diverging $n$ and $m$  have been thoroughly studied in the literature of factor models \citep{Bai-factor-model-03,Bai-Ng-forecast,fan2013large}. Our results reduce to the existing results in this case with $r_n = 0$. The logarithmic factor of $m$ comes from establishing the union bound over $j\in[m]$. The appearance of $m$ in the denominator of bound (\ref{eq_thm_rates_B}) also reflects the benefit of having a large $m$, the so-called blessing of dimensionality \citep{Bai-factor-model-03,fan2013large}. When one only has access to $\wh\bepsilon$ instead of $\bepsilon$, the analysis becomes more challenging. Specifically, since $\wh\bepsilon= \bW\bB + \wt\bE$ with $\wt\bE := \bE+\wh\bepsilon - \bepsilon$, one can view $\wh\bepsilon$ as a factor model with the factor component $\bW\bB$ and the error $\wt\bE $. The difficulty of establishing Theorem \ref{thm_rates_B} lies in characterizing the dependence between $\wt\bE$ and $\bW\bB$, as $\wh\bepsilon$ depends on the data hence also depends on $\bW$ in a complicated way.\\

    In addition to the rates of convergence, the following theorem provides the asymptotic normality of $\wh \bB_j$ for any  $1\le j\le m$. 
    
    \begin{theorem}\label{thm_B_asn}
        Under the same conditions of Theorem \ref{thm_rates_B}, assume  $s_n\log(p\vee m) = o(\sqrt{n})$, $\|\se\|_{\i,1}=\cO(1)$, $ \sqrt{n} = o(m/\log (m))$ and 
        \begin{equation}\label{cond_r_asn}
            \|\bA\|_{\op}^2\max\left\{ n \|\bA\bB_j\|_2^2,~ s_n\log (p\vee m), ~ \sqrt{n\log m\over m}\right\}= o(1).
        \end{equation}
        Then for any $1\le j\le m$, one has 
        \[
            \sqrt{n}(\wh\bB_j - \wt\bB_j) \overset{d}{\longrightarrow} N_K(\b0, \sigma_{E_1}^2 \bI_K),\qquad \textrm{as }n\to\infty.        
        \]
    \end{theorem}
    For the same reason, since we do not impose any identifiability conditions for $\bB$, our estimator $\wh \bB_j$ is not centered around $\bB_j$ but rather its rotated version $\wt\bB_j = \bH_0\bB_j$ \citep{Bai-factor-model-03,bai2020simpler}. We emphasize that this rotation does not impede us from testing $\bB_j = \b0$. 
    Specifically, Theorem \ref{thm_B_asn} implies that for any $1\le j\le m$, under the null hypothesis $\bB_j = \b0$, 
    \[
       n \wh \bB_j^T \wh \bB_j / \sigma_{E_j}^2 \overset{d}{\longrightarrow} \chi^2_K,\qquad \textrm{as }n\to\infty.        
    \]
    provided that 
    \begin{equation}\label{cond_A_infer_B}
        \|\bA\|_{\op}^2 \max\left\{ s_n\log (p\vee m), ~ \sqrt{n\log(m)/m}\right\}= o(1).
    \end{equation}
    Since $\sigma_{E_1}^2$ can be consistently estimated as shown in Proposition \ref{prop_sigma_E} of Section \ref{sec_effciency}, this justifies the validity of our testing statistic $\wh R_n^{(1)}$ in (\ref{def_R_hat}) of Section \ref{sec_method_infer_B}.
    In case one is willing to assume additional identifiability conditions on $\bB$, such as those in \cite{Bai-Ng-forecast}, the rotation matrix $\bH_0$ becomes the identity matrix asymptotically \citep{bai2020simpler}. 
    
    In the following, we comment on the conditions in Theorem \ref{thm_B_asn}. 
    To allow a non-diagonal $\se$, the inferential result on $\bB$ requires $\|\se\|_{\i,1}=\cO(1)$, a stronger condition than Assumption \ref{ass_B_Sigma} (c), as well as $\log (m) \sqrt{n} = o(m)$. These conditions are commonly assumed in the analysis of factor models \citep{Bai-factor-model-03,Bai-Ng-forecast,bai2020simpler}, and can be dropped if $\se$ is proportional to the identity matrix, as remarked in \citet[Theorem 6]{Bai-factor-model-03}. 
    Condition (\ref{cond_r_asn}) is needed to ensure that the error
    of estimating $\bepsilon$ by $\wh\bepsilon$ is negligible. For the similar reason, if $\se$ is proportional to the identity matrix, the requirement $\|\bA\|_{\op}^2\sqrt{n\log(m)/m}=o(1)$ can be removed. In general, condition (\ref{cond_r_asn}) holds, for instance, if $\sqrt{n/m} = \cO(s_n\log(p\vee m))$, 
    \begin{equation}\label{eq_condition_B}
          \|\bA\|_{\op}^2 = o\left({1\over s_n\log (p\vee m)}\right),\qquad  \|\bA\|_{\op}^2\|\bA\bB_j\|_2^2 = o\left(
       {1\over n}
    \right).  
    \end{equation}
    We reiterate that for testing the hypothesis $\bB_j = \b0$, the condition $\|\bA\|_{\op}^2\|\bA\bB_j\|_2^2 = o(1/n)$ holds automatically. We refer to Corollary \ref{cor_ASN} for the discussion on the first condition in (\ref{eq_condition_B}).\\

    \begin{remark}[Comparison with \cite{guo2020doubly}]
       As briefly mentioned in the Introduction, \cite{guo2020doubly} consider  the univariate model $y = X^T \btheta + Z^T\bbeta +\epsilon$ and propose  a doubly debiased lasso procedure for making inference on entries of $\btheta$, say $\theta_1$, in the presence of hidden confounders $Z\in\RR^{K}$. Although both their estimator of $\theta_1$ and our estimator of $\Theta_{11}$ are shown to be efficient in the Gauss-Markov sense (i.e. the same asymptotic variance), the analyses are carried under different modelling assumptions. For instance, different from our model, \cite{guo2020doubly} additionally assume $X = \bGamma Z + W'$ with some additive error $W'$ that is independent of $Z$. They also assume all $K$ singular values of the loading matrix $\bGamma$ to be of order $\sqrt{p}$. Consequently, the $L_2$-projection matrix $\bA = (\EE[XX^T])^{-1}\EE[XZ^T]$ satisfies $\|\bA\|_{\op}^2=\cO(1/p)$ and the residual vector $W=Z - \bA^T X$ satisfies $\|\sw\|_{\op} = \cO(1/p)$. By contrast, from Corollary \ref{cor_ASN} and its subsequent remark, our analysis does not necessarily require  $\|\bA\|_{\op}^2=\cO(1/p)$. This could be understood as the benefits of having multivariate responses. On the other hand, we require parts (a) and (b) in Assumption \ref{ass_B_Sigma} and the latter does not hold under the conditions on $X$ and $\bGamma$ in \cite{guo2020doubly}. Finally, due to the multivariate nature of the responses, we are able to conduct inference on $\bB$ to test the existence of hidden confounders, whereas, in the univariate case, \cite{guo2020doubly} does not study such inference problems on $\bbeta$. 
    \end{remark}

\section{Practical considerations and simulation study}\label{sec_prac_and_sim}

In this section we first discuss two practical considerations of our procedure: selection of the number of hidden variables  $K$ in Section \ref{sec_K} and selection of tuning parameters in Section \ref{sec_cv}. We then evaluate the finite sample performance of the proposed inferential method via synthetic datasets in Section \ref{sec_sim}.

\subsection{Selection of the number of hidden variables}\label{sec_K}

    Recall that $\bepsilon = \bW\bB + \bE$ follows a factor model with $K$ latent factors  (corresponding to $\bW$) if $\bepsilon$ were observed. We propose to select $K$ based on the estimate $\wh\bepsilon$ in (\ref{def_est_epsilon}) of $\bepsilon$. Specifically, we adopt the criterion in \cite{bing2020adaptive} that selects $K$ by 
    \begin{equation}\label{select_K}
	\wh K = \argmax_{j\in \{1,2,\ldots, \bar K\}}   d_j / d_{j+1},
	\end{equation}
	where $d_1 \ge d_2 \ge \cdots$ are the singular values of $\wh\bepsilon/\sqrt{nm}$ in (\ref{def_svd_epsilon}) and $\bar K$ is a pre-specified number, for example, $\bar K =\floor{(n\wedge m)/2}$ \citep{lam2012} with $\floor{x}$ standing for the largest integer that is no greater than $x$. Criterion (\ref{select_K}) is first proposed by \cite{lam2012} for selecting the number of latent factors in factor models. It is related with the ``elbow'' approach of selecting the number of components in PCA. In our current context, both theoretical and empirical justifications of the criterion (\ref{select_K}) have been provided in \cite{bing2020adaptive}. On the other hand, there exist other methods of selecting $K$ for which we refer to \cite{Lee2017,wang2017,bing2020adaptive}.
	
\subsection{Selection of tuning parameters}\label{sec_cv}
    We describe how to practically select the tuning parameters in our procedure of making inference of $\Theta_{11}$.
    
    The estimation of $\bX\bF$ in (\ref{def_est_F_j}) requires the selection of $\lbdIj$ and $\lbdj$ for $j\in [m]$. Their theoretical orders are stated in Theorem \ref{thm_pred} of Appendix \ref{app_theory_fit}. In practice, one could choose them over a two-way grid of $\lbdIj$ and $\lbdj$ via cross-validation (CV) by minimizing the mean squared prediction error on a validation set (for instance, by using the $k$-fold CV). When the dimensions $p$ and $m$ are large, such two-way grid search might be computationally burdensome. \citet[Appendix E.3]{bing2020adaptive} proposed a faster way of selecting $\lbdIj$ and $\lbdj$. For the reader's convenience, we restate it here. Pick any $j\in[m]$. We start with a grid $\mathcal{G}$ of $\lbdj$ and for each $\lbdj \in \mathcal{G}$, we set 
        \[
            \lbdIj(\lbdj) = c_0 \sqrt{\max_{1\le j\le p} M_{jj}(\lbdj)}\left(\sqrt{m \over n} + \sqrt{2\log p \over n}\right)
        \]
        where 
        $\bM(\lbdj) = n^{-1} \bX^T Q^2_{\lbdj} \bX$ with 
        $Q_{\lbdj} = \bI_n - \bX(\bX^T\bX + n\lbdj\bI_p)^{-1}\bX^T$
        and 
        $c_0>0$ is some universal constant (our simulation reveals good performance for $c_0 = 1$). This choice of $\lbdIj(\lbdj)$ is based on its theoretical order in Theorem \ref{thm_pred} of Appendix \ref{app_theory_fit}. We then use $5$-fold cross validation to select $\lambda_2^{(j)*}$ which gives the smallest mean squared error of the predicted values. Fixing $\lambda_2^{(j)*}$, the optimization problem in (\ref{crit_Theta}) becomes a group-lasso problem and we propose to select $\lbdIj$ via $5$-fold cross validation (for instance, the \textsf{cv.glmnet} package in \textsc{R}). 
        
    The initial estimator $\wh\bTheta_1$ of $\bTheta_1$ in (\ref{def_Thetaj_hat}) requires another tuning parameter $\lambda_3$. As (\ref{def_Thetaj_hat}) solves a standard lasso problem, we propose to select $\lambda_3$ via $5$-fold cross validation implemented in the \textsf{cv.glmnet} package in \textsc{R}. 
    
    Finally, recall that we use the node-wise lasso procedure in (\ref{def_est_omega}) for estimating the first column of the precision matrix $\bOmega$. We propose to select $\wt \lambda$ in (\ref{def_est_omega}) by 5-fold CV as well.

\subsection{Simulations}\label{sec_sim}

    In this section we conduct extensive simulations to verify the performance of our developed inferential tools for testing $\Theta_{ij} = 0$ and $\bB_j = \b0$.

\paragraph{Data generating mechanism:}
The data generating process is as follows. 
For generating the design matrix, we simulate $\bX_i$ i.i.d. from $N_p(\b0,\Sigma)$ where $\Sigma_{jk} = (-1)^{j+k}\cdot (0.5)^{|j-k|}$ for all $j,k\in [p]$. We simulate $A_{jk}\sim \eta\cdot N(0.5,0.1)$ and $B_{kl}\sim N(0.1,1)$ for $j \in [p]$, $k\in[K]$, $l\in [m]$ where the parameter $\eta$ controls the magnitude of entries of $\bA$. 
To generate $\bTheta$, for given integers $s$ and $s_m$,  
we sample entries of the top  left  $s\times s_m$ submatrix of $\bTheta$ i.i.d. from $N(2,0.1)$ and set all other entries of $\bTheta$ to zero. The number of non-zero rows of $\bTheta$ is set to $s = 3$ while the sparsity of each non-zero row is fixed as $s_m = 10$. Next, we generate i.i.d.  $\bZ_i = \bA^T\bX_i + \bW_i$ with $\bW_i \sim N_K(\b0,3^2\bI_K)$. Finally, we generate i.i.d. $\bY_i = \bTheta^T\bX_i + \bB^T\bZ_i + \bE_i$ with $\bE_i\sim N_p(\b0,\bI_m)$.

Throughout the simulation, we fix  $n = 200$, $K=3$ and consider $p \in \{50, 250\}$, $m\in \{20,50,100\}$ and $\eta \in \{0.2,1\}$. Each setting is repeated 25 times without further specification. 

\paragraph{Procedures under comparison:}  For our proposed procedure, we select tuning parameters in the way we described in Section \ref{sec_cv}. To concentrate on the comparison of inference, we use the true $K$ as input (our simulation reveals that $K$ can be consistently estimated by (\ref{select_K}) in almost all settings). For comparison, we also consider the following approaches.
\begin{itemize}
    \item Desparsified method (DSpar) implemented in the ``hdi'' package in \textsc{R},
    \item Decorrelated Score (DScore) test implemented in the ``ScoreTest''  package\footnote{\url{https://github.com/huijiefeng/ScoreTest}}  in \textsc{R},
    \item Doubly Debiased Lasso (DDL) method proposed by \citet{guo2020doubly}\footnote{\url{https://github.com/zijguo/Doubly-Debiased-Lasso}}.
\end{itemize}



\paragraph{Testing on $\bTheta$:}

We evaluate the performance of conducting hypothesis testing on $\bTheta$ by using all four methods in each combination setting of $p \in \{50, 250\}$, $m\in \{20,50,100\}$ and $\eta \in \{0.2,1\}$. 
To introduce the metrics we use, for each generated $\bTheta$, we let $\cS = \{(i,j): \Theta_{ij}\neq 0\}$ denote the support of $\bTheta$ and $\cS^c$ denote its complement. By fixing the significance level at $\alpha = 0.05$, we compute the the empirical Type I error and the empirical Power for each method, defined as
\beq\nonumber
&\text{Type I error} = \frac{1}{|\cS^c|} \sum_{(i,j)\in \cS^c} 1\left\{
\textrm{Reject the null $H_{0,\Theta_{ij}}$}
\right\}\\\nonumber
&\text{Power} = \frac{1}{|\cS|} \sum_{(i,j)\in \cS}
1\left\{
\textrm{Reject the null $H_{0,\Theta_{ij}}$}
\right\}
\eeq 

Table \ref{table_error} reports the averaged Type I errors and Powers for all four methods in each setting\footnote{Since \cite{guo2020doubly} only provides guarantees of DDL for large $p$, we only compare with DDL in the high-dimensional scenarios. Due to the long running time of DDL, we only report its performance for $m = 20$ and $p = 250$.}. As we can see, when $\eta = 0.2$ so that the magnitude of hidden effects is relatively small, in both low ($p=50$) and high ($p=250$) dimensional settings, the averaged Type I errors of all methods are generally close to the nominal level 0.05, while the proposed method achieves higher Powers. When $\eta = 1.0$ so that the magnitude of hidden effects is relatively large, in the low dimensional setting $p = 50$, the averaged Type I errors of the proposed approach are much lower and closer to the nominal level than all other methods.  
On the other hand, in the high dimensional setting $p = 250$, despite all methods have similar Type I errors,   our proposed approach yields much higher Powers.

\begin{table}[t]
 	\caption{The averaged Type I errors and Powers at significance level 0.05 for the proposed method, DSpar, DScore  and DDL}\label{table_error}
 	\begin{center}
 		\begin{tabular}{ c c c c c c c c c}
 			\hline
 			$p$ &Metric&Method&\multicolumn{3}{c}{$\eta = 0.2$}&\multicolumn{3}{c}{$\eta = 1.0$}  \\
 			&&& $m = 20$ & $m = 50$ & $m = 100$ & $m = 20$ & $m = 50$ & $m = 100$\\
 			\hline
 			50&Type I error&Proposed&0.057&0.072&0.085&0.117&0.102&0.104\\
 			&&DSpar&0.060&0.059&0.064&0.338&0.313&0.282\\
 			&&DScore&0.054&0.060&0.051&0.367&0.361&0.348\\
 			&&DDL&
 			 - & - & - & - & - & - \\
 		\hline	&Power&Proposed&1.000&1.000&1.000&0.929&1.000&1.000\\
 			&&DSpar&0.970&0.866&0.941&0.924&0.957&0.757\\
 			&&DScore&0.982&0.916&0.934&0.908&0.857&0.942\\
 			&&DDL & 
 			- & - & - & - & - & - \\
 			\hline
 			250&Type I error&Proposed&0.051&0.076&0.063&0.089& 0.097 &0.116\\
 			&&DSpar&0.058&0.059&0.054&0.110&0.114&0.111\\
 			&&DScore&0.045&0.046&0.052&0.105&0.104&0.109\\
 			&&DDL&0.098&-&-&0.114&-&-\\
 		\hline	&Power&Proposed&1.000&1.000&1.000&0.998&1.000&0.998\\
 			&&DSpar&0.934&0.88&0.954&0.580&0.602&0.729\\
 			&&DScore&0.913&0.856&0.883&0.663&0.683&0.702\\
 			&&DDL&0.893&-&-&0.691&-&-\\
 			\hline
		\end{tabular}
 	\end{center}	
 \end{table}

We further demonstrate how the empirical Type I error and Power of different methods change as the signal strength varies. To this end, we generate $\bTheta$ by setting its non-zero entries to $r$ with $r$ varying within $\{0.05,0.07,0.1,0.2,0.3,0.5,1,1.5,2.0\}$. We consider $p = 50$, $m = 20$ and $\eta\in\{0.2,1\}$. For each choice of $r$ and $\eta$, we repeat generating the data and computing Type I errors and Powers 25 times. Figure \ref{fig:power} depicts how the averaged Type I errors and Powers change as $r$ increases for different methods. 
When $\eta = 0.2$, the averaged Type I errors of all methods are similar and close to 0.05 but our proposed approach has much higher Powers than the other two methods over the whole range of the signal strength. When $\eta = 1.0$, it is clear that both DSpar and DScore fail to control the Type I errors whereas our proposed method not only controls the Type I error but also has much higher Powers as the signal strength increases. Figure \ref{fig:power} together with the results from Table \ref{table_error} suggests the superiority of our proposed approach over the compared methods.   

\begin{figure}[H]
    \centering
    \includegraphics[width = 16cm, height = 10cm]{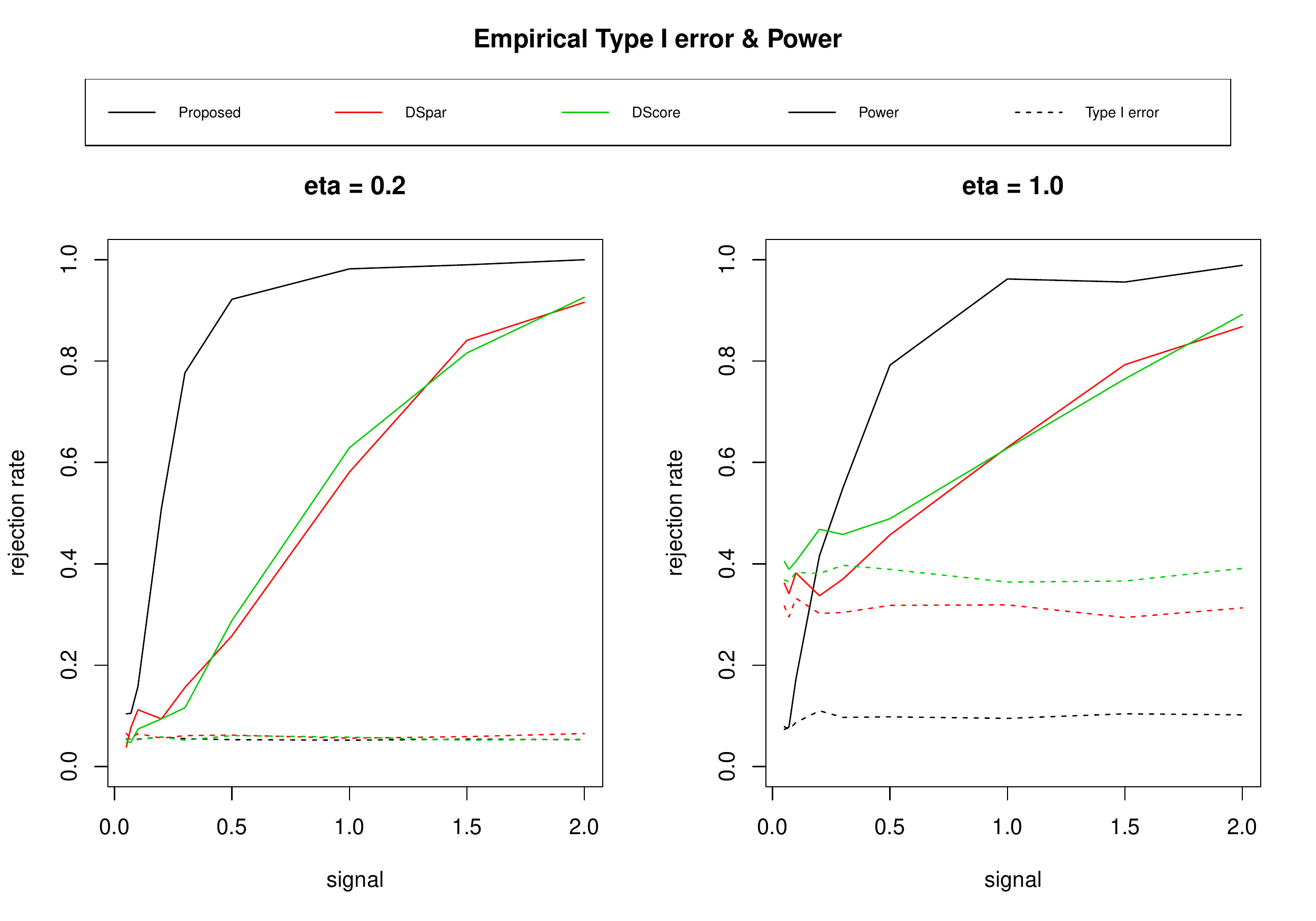}
    \caption{The average Type I errors and Powers with varying magnitude of the nonzero coefficients of $\bTheta$. The black, red and green lines represent the proposed approach, DSpar and DScore, respectively. The solid lines depict the averaged Powers while the dashed lines represent the averaged Type I errors.}
    \label{fig:power}
\end{figure}

\paragraph{Testing on $\bB$:}
We proceed to evaluate the empirical performance of our proposed method for testing the hypothesis $H_{0,B_j}: \bB_j = \b0$ versus $H_{1,B_j}: \bB_j \ne  \b0$. We adopt the same data generating process as described in the beginning except that we set $\bB_j = \b0$ for each $j \in \{1,\dotso,b_m\}$. Here $b_m$ controls the number of zero columns of $\bB$ and is chosen from $\{5, 10\}$. We also consider  $p = 50$, $\eta = 0.1$ and vary $m$ within $\{20, 50, 100\}$. 
Similarly, we calculate the empirical Type I error and the empirical Power as
\beq
&\text{Type I error} = \frac{1}{b_m} \sum_{j=1}^{b_m}1\left\{\textrm{Reject the null $H_{0,B_{j}}$}\right\},\\
&\text{Power} = \frac{1}{(m - b_m)} \sum_{j= b_m+1}^m 1\left\{\textrm{Reject the null $H_{0,B_{j}}$}\right\}.
\eeq
We repeat 100 times for each scenario. Table \ref{table_error_B} contains the averaged Type I errors and Powers of our procedure in all settings. The Type I errors are not far from the nominal level 0.05 and get closer to it as $m$ increases while the Powers are close to one in all settings. These findings are in line of our Theorem \ref{thm_B_asn}.

\begin{table}[t]
 	\caption{The averaged Type I errors and Powers at significance level 0.05 for the proposed method of testing $H_{0,B_j}: \bB_j = \b0$ versus $H_{1,B_j}: \bB_j\ne \b0$.}\label{table_error_B}
 	\begin{center}
 		\begin{tabular}{ c c c c c c c}
 			\hline
 			Metric &\multicolumn{3}{c}{$b_m = 5$}&\multicolumn{3}{c}{$b_m = 10$}\\
 			& $m = 20$ & $m = 50$& $m = 100$& $m = 20$ & $m = 50$& $m = 100$\\
 			\hline
 			Type I error&0.072&0.064&0.062&0.063&0.041&0.058\\
 		    \hline	
 		   Power&0.989&1.000&0.998&1.000&0.988&0.999\\
 			\hline
		\end{tabular}
 	\end{center}	
 \end{table}

\section{Analysis on the stock mouse dataset}\label{sec_real_data}

In this section, we validate our method on the heterogenous stock mouse dataset \citep{valdar2006genome} from Wellcome Trust Centre for Human Genetics. 
This dataset contains $129$ continuous phenotypes that can be categorized into six categories: Behavior, Diabetes, Ashma, Immunology, Haemotology and Biochemistry. The dataset also contains around $10,000$ 
Single Nucleotide Polymorphisms (SNPs) for each mouse. One primary interest is to discover significant associations between the SNPs and the phenotypes.  Since both phenotypes and genotypes are measured by different experimenters at different time points and the mice are from different generations and families \citep{valdar2006genome}, we expect the existence of unknown hidden effects, such as batch effects. We thus deploy  our proposed method for finding significant entries of $\bTheta$ by adjusting the potential hidden effects. 

To preprocess the data, since the measured phenotypes and SNPs vary for different groups of mice,
we only consider the mice that should have all phenotypes measured. 
Meanwhile, we only keep the SNPs that have been measured by these retained mice. Finally, since there exists  different levels of missingness among the phenotypes, we remove those phenotypes with percentage of missing values greater than $5\%$ and impute the missing values of the remaining phenotypes by using the average of their $20$-nearest neighbors. After the data preprocessing,  we obtain a data set that has  $n = 810$ mice, $p = 10,346$ measured SNPs and $m = 104$ recorded phenotypes.

To deploy our method, we first use the procedure in Section \ref{sec_K} to find $\wh K=28$ for this dataset and then apply our procedure in (\ref{sec_est_Theta}) to test the significance of each entry of $\bTheta$. The tuning parameters are chosen in the way as described in Section \ref{sec_cv}. To account for multiple testing problem, 
we apply the Bonferroni correction at 0.05 significant level. 
For comparison, we also run both DSpar and DScore (see,  Section \ref{sec_sim}) with the same correction. 
To interpret and validate the discovered significant associations, we map the SNPs to either annotated genes or intergenic regions.

On the one hand, our approach and the other two methods detect some common meaningful signals. For example, in Diabetes related phenotypes, such as Insulin, both our method and DSpar find the SNP \textit{rs4213255} to be significant. 
This SNP is mapped to gene \textit{repro33} which has been shown to be associated with endocrine and exocrine glands \citep{goldfine1997endocrine} that directly mediates insulin level. Another SNP that is found by both our method 
and DSpar 
to be significant for an immunology phenotype 
is \textit{rs13476136} whose corresponding gene \textit{Tli1} (T lymphoma induced 1) has been demonstrated to directly affect immunology \citep{wielowieyski1999tli1, blake2003mgd, smith2019mouse, krupke2017mouse}. 
Furthermore, significance of the SNP \textit{rs3713052} is discovered for a Haemotology related phenotype (Haem.LICabs) by all three methods, and this SNP is mapped into the intergenic region between the gene \textit{Gm39049} and the gene \textit{Tenm4}. 
Although the function of this intergenic region is unclear to us, the \textit{Tenm4} gene has been found to associate with the hematopoietic system \citep{blake2003mgd, smith2019mouse, krupke2017mouse}.

On the other hand, there exist many meaningful  associations that are only identified to be significant by our method. For instance, the SNP \textit{rs6290322}
is only found to be significant by our method for a Diabetes related phenotype (Glucose). It has been shown that the mapped gene \textit{gro57} of this SNP is associated with several Diabetic phenotypes
\citep{blake2003mgd, smith2019mouse, krupke2017mouse}. Our method also finds the
SNP \textit{rs3141314} to be significant 
for a Haemotology phenotype (Haem.PLT, platelet count). This SNP is mapped to gene \textit{hlb258} which is known to be functional related with the blood phenotypes \citep{blake2003mgd, smith2019mouse, krupke2017mouse}. In addition, several SNPs such as $rs3711203$ and $rs3725230$ are only found by our method to be significant 
for multiple immunological phenotypes. These SNPs are all mapped to gene \textit{slck} (slick hair gene) which directly effects the integumentary system \citep{blake2003mgd, smith2019mouse, krupke2017mouse}. The integumentary system including the skin and corresponding appendages acts as a physical barrier between outside environment and internal environment hence plays an important role in the immune system.

Overall, our method finds more meaningful and significant SNPs than the other two methods. Specifically, for each method, we record the numbers of significant SNPs for each phenotype and report the summary statistics of these numbers in Table \ref{real data}. We also run our testing procedure in Section \ref{sec_method_infer_B} for $\bB$ and all the test statistics are very large ($>427$ for all phenotypes), suggesting the existence of strong hidden effects. Although DSpar and Dscore are able to detect a few signals that are sufficiently large without adjusting the hidden effects, to find more weak/moderate yet meaningful signals, our proposed approach appears to be more effective. 

\begin{table}[ht]
 	\caption{Summary statistics of the numbers of significant SNPs over all phenotypes by using different methods.}\label{real data}
 	\begin{center}
 		\begin{tabular}{c c c c c}
 			\hline
 			Method& Min & Mean & Median & Max\\
 			\hline
 			Ours & 7 & 21.77 & 21 & 43\\
 		    \hline	
 		   DSpar & 0 & 1.77 & 0 & 39\\
 			\hline
 		   DScore & 0 & 0.09& 0  & 5\\
 			\hline
		\end{tabular}
 	\end{center}	
 \end{table}

    \newpage

    {\small 
    \setlength{\bibsep}{0.85pt}{
    \bibliographystyle{plainnat}
    \bibliography{ref}

\begin{thebibliography}{46}
\providecommand{\natexlab}[1]{#1}
\providecommand{\url}[1]{\texttt{#1}}
\expandafter\ifx\csname urlstyle\endcsname\relax
  \providecommand{\doi}[1]{doi: #1}\else
  \providecommand{\doi}{doi: \begingroup \urlstyle{rm}\Url}\fi

\bibitem[Ahn and Horenstein(2013)]{Ahn-2013}
Seung~C. Ahn and Alex~R. Horenstein.
\newblock Eigenvalue ratio test for the number of factors.
\newblock \emph{Econometrica}, 81\penalty0 (3):\penalty0 1203--1227, 2013.

\bibitem[Anderson(1984)]{anderson_book}
T.~W. Anderson.
\newblock \emph{An introduction to multivariate statistical analysis}.
\newblock Wiley Series in Probability and Statistics. Wiley, 1984.

\bibitem[Bai(2003)]{Bai-factor-model-03}
Jushan Bai.
\newblock Inferential theory for factor models of large dimensions.
\newblock \emph{Econometrica}, 71\penalty0 (1):\penalty0 135--171, 2003.

\bibitem[Bai and Ng(2002)]{Bai-Ng-K}
Jushan Bai and Serena Ng.
\newblock Determining the number of factors in approximate factor models.
\newblock \emph{Econometrica}, 70\penalty0 (1):\penalty0 191--221, 2002.

\bibitem[Bai and Ng(2008)]{Bai-Ng-forecast}
Jushan Bai and Serena Ng.
\newblock Forecasting economic time series using targeted predictors.
\newblock \emph{Journal of Econometrics}, 146\penalty0 (2):\penalty0 304 --
  317, 2008.
\newblock Honoring the research contributions of Charles R. Nelson.

\bibitem[Bai and Ng(2020)]{bai2020simpler}
Jushan Bai and Serena Ng.
\newblock Simpler proofs for approximate factor models of large dimensions.
\newblock \emph{arXiv preprint arXiv:2008.00254}, 2020.

\bibitem[Belloni et~al.(2015)Belloni, Chernozhukov, and
  Kato]{belloni2015uniform}
Alexandre Belloni, Victor Chernozhukov, and Kengo Kato.
\newblock Uniform post-selection inference for least absolute deviation
  regression and other z-estimation problems.
\newblock \emph{Biometrika}, 102\penalty0 (1):\penalty0 77--94, 2015.

\bibitem[Bickel et~al.(2009)Bickel, Ritov, and Tsybakov]{bickel2009}
Peter~J. Bickel, Ya’acov Ritov, and Alexandre~B. Tsybakov.
\newblock Simultaneous analysis of lasso and dantzig selector.
\newblock \emph{Ann. Statist.}, 37\penalty0 (4):\penalty0 1705--1732, 08 2009.
\newblock \doi{10.1214/08-AOS620}.

\bibitem[Bing et~al.(2019)Bing, Bunea, and Wegkamp]{bing2020inference}
Xin Bing, Florentina Bunea, and Marten Wegkamp.
\newblock Inference in interpretable latent factor regression models.
\newblock \emph{arXiv e-prints}, pages arXiv--1905, 2019.

\bibitem[Bing et~al.(2020)Bing, Ning, and Xu]{bing2020adaptive}
Xin Bing, Yang Ning, and Yaosheng Xu.
\newblock Adaptive estimation of multivariate regression with hidden variables.
\newblock \emph{arXiv preprint arXiv:2003.13844}, 2020.

\bibitem[Bing et~al.(2021)Bing, Bunea, Strimas-Mackey, and
  Wegkamp]{bing2020prediction}
Xin Bing, Florentina Bunea, Seth Strimas-Mackey, and Marten Wegkamp.
\newblock Prediction under latent factor regression: Adaptive pcr,
  interpolating predictors and beyond.
\newblock \emph{Journal of Machine Learning Research}, 22\penalty0
  (177):\penalty0 1--50, 2021.

\bibitem[Blake et~al.(2003)Blake, Richardson, Bult, Kadin, and
  Eppig]{blake2003mgd}
Judith~A Blake, Joel~E Richardson, Carol~J Bult, Jim~A Kadin, and Janan~T
  Eppig.
\newblock Mgd: the mouse genome database.
\newblock \emph{Nucleic acids research}, 31\penalty0 (1):\penalty0 193--195,
  2003.

\bibitem[Cand\`{e}s et~al.(2011)Cand\`{e}s, Li, Ma, and Wright]{Candes}
Emmanuel~J. Cand\`{e}s, Xiaodong Li, Yi~Ma, and John Wright.
\newblock Robust principal component analysis?
\newblock \emph{J. ACM}, 58\penalty0 (3):\penalty0 11:1--11:37, June 2011.
\newblock ISSN 0004-5411.
\newblock \doi{10.1145/1970392.1970395}.

\bibitem[{\'C}evid et~al.(2018){\'C}evid, B{\"u}hlmann, and
  Meinshausen]{cevid2018spectral}
Domagoj {\'C}evid, Peter B{\"u}hlmann, and Nicolai Meinshausen.
\newblock Spectral deconfounding via perturbed sparse linear models.
\newblock \emph{arXiv preprint arXiv:1811.05352}, 2018.

\bibitem[Chandrasekaran et~al.(2011)Chandrasekaran, Sanghavi, Parrilo, and
  Willsky]{Chandrasekaran}
Venkat. Chandrasekaran, Sujay. Sanghavi, Pablo~A. Parrilo, and Alan~S. Willsky.
\newblock Rank-sparsity incoherence for matrix decomposition.
\newblock \emph{SIAM Journal on Optimization}, 21\penalty0 (2):\penalty0
  572--596, 2011.
\newblock \doi{10.1137/090761793}.

\bibitem[Chandrasekaran et~al.(2012)Chandrasekaran, Parrilo, and
  Willsky]{chandrasekaran2012latent}
Venkat Chandrasekaran, Pablo~A Parrilo, and Alan~S Willsky.
\newblock Latent variable graphical model selection via convex optimization.
\newblock \emph{The Annals of Statistics}, pages 1935--1967, 2012.

\bibitem[Chernozhukov et~al.(2017)Chernozhukov, Hansen, and
  Liao]{chernozhukov2017}
Victor Chernozhukov, Christian Hansen, and Yuan Liao.
\newblock A lava attack on the recovery of sums of dense and sparse signals.
\newblock \emph{Ann. Statist.}, 45\penalty0 (1):\penalty0 39--76, 02 2017.
\newblock \doi{10.1214/16-AOS1434}.

\bibitem[Davis and Kahan(1970)]{DavisKahan}
Chandler Davis and W.~M. Kahan.
\newblock The rotation of eigenvectors by a perturbation. iii.
\newblock \emph{SIAM Journal on Numerical Analysis}, 7\penalty0 (1):\penalty0
  1--46, 1970.
\newblock \doi{10.1137/0707001}.

\bibitem[Fan et~al.(2013)Fan, Liao, and Mincheva]{fan2013large}
Jianqing Fan, Yuan Liao, and Martina Mincheva.
\newblock Large covariance estimation by thresholding principal orthogonal
  complements.
\newblock \emph{Journal of the Royal Statistical Society: Series B (Statistical
  Methodology)}, 75\penalty0 (4):\penalty0 603--680, 2013.

\bibitem[Fan et~al.(2017)Fan, Xue, and Yao]{fan2017}
Jianqing Fan, Lingzhou Xue, and Jiawei Yao.
\newblock Sufficient forecasting using factor models.
\newblock \emph{Journal of Econometrics}, 201\penalty0 (2):\penalty0 292 --
  306, 2017.

\bibitem[Gagnon-Bartsch and Speed(2012)]{Gagnon2012}
Johann~A. Gagnon-Bartsch and Terence~P. Speed.
\newblock {Using control genes to correct for unwanted variation in microarray
  data}.
\newblock \emph{Biostatistics}, 13\penalty0 (3):\penalty0 539--552, 11 2012.

\bibitem[Goldfine et~al.(1997)Goldfine, German, Tseng, Wang, Bolaffi, Chen,
  Olson, and Rothman]{goldfine1997endocrine}
Ira~D Goldfine, Michael~S German, Hsien-Chen Tseng, Juemin Wang, Janice~L
  Bolaffi, Je-Wei Chen, David~C Olson, and Stephen~S Rothman.
\newblock The endocrine secretion of human insulin and growth hormone by
  exocrine glands of the gastrointestinal tract.
\newblock \emph{Nature biotechnology}, 15\penalty0 (13):\penalty0 1378--1382,
  1997.

\bibitem[Guo et~al.(2020)Guo, {\'C}evid, and B{\"u}hlmann]{guo2020doubly}
Zijian Guo, Domagoj {\'C}evid, and Peter B{\"u}hlmann.
\newblock Doubly debiased lasso: High-dimensional inference under hidden
  confounding and measurement errors.
\newblock \emph{arXiv e-prints}, pages arXiv--2004, 2020.

\bibitem[{Hsu} et~al.(2011){Hsu}, {Kakade}, and {Zhang}]{Hsu2011}
D.~{Hsu}, S.~M. {Kakade}, and T.~{Zhang}.
\newblock Robust matrix decomposition with sparse corruptions.
\newblock \emph{IEEE Transactions on Information Theory}, 57\penalty0
  (11):\penalty0 7221--7234, Nov 2011.
\newblock ISSN 1557-9654.
\newblock \doi{10.1109/TIT.2011.2158250}.

\bibitem[Hsu et~al.(2014)Hsu, Kakade, and Zhang]{Hsu2014}
Daniel Hsu, Sham~M. Kakade, and Tong Zhang.
\newblock Random design analysis of ridge regression.
\newblock \emph{Found. Comput. Math.}, 14\penalty0 (3):\penalty0 569--600, June
  2014.
\newblock ISSN 1615-3375.
\newblock \doi{10.1007/s10208-014-9192-1}.

\bibitem[Janzing and Sch{\"o}lkopf(2018)]{janzing2018detecting}
Dominik Janzing and Bernhard Sch{\"o}lkopf.
\newblock Detecting confounding in multivariate linear models via spectral
  analysis.
\newblock \emph{Journal of Causal Inference}, 6\penalty0 (1), 2018.

\bibitem[{Javanmard} and {Montanari}(2014)]{javanmard14}
Adel {Javanmard} and Andrea {Montanari}.
\newblock {Confidence intervals and hypothesis testing for high-dimensional
  regression}.
\newblock \emph{{J. Mach. Learn. Res.}}, 15:\penalty0 2869--2909, 2014.
\newblock ISSN 1532-4435; 1533-7928/e.

\bibitem[Javanmard and Montanari(2018)]{Javanmard2018}
Adel Javanmard and Andrea Montanari.
\newblock {Debiasing the lasso: Optimal sample size for Gaussian designs}.
\newblock \emph{The Annals of Statistics}, 46\penalty0 (6A):\penalty0 2593 --
  2622, 2018.
\newblock \doi{10.1214/17-AOS1630}.

\bibitem[Krupke et~al.(2017)Krupke, Begley, Sundberg, Richardson, Neuhauser,
  and Bult]{krupke2017mouse}
Debra~M Krupke, Dale~A Begley, John~P Sundberg, Joel~E Richardson, Steven~B
  Neuhauser, and Carol~J Bult.
\newblock The mouse tumor biology database: a comprehensive resource for mouse
  models of human cancer.
\newblock \emph{Cancer research}, 77\penalty0 (21):\penalty0 e67--e70, 2017.

\bibitem[Lam and Yao(2012)]{lam2012}
Clifford Lam and Qiwei Yao.
\newblock Factor modeling for high-dimensional time series: Inference for the
  number of factors.
\newblock \emph{Ann. Statist.}, 40\penalty0 (2):\penalty0 694--726, 04 2012.

\bibitem[Lee et~al.(2017)Lee, Sun, Wright, and Zou]{Lee2017}
Seunggeun Lee, Wei Sun, Fred~A. Wright, and Fei Zou.
\newblock {An improved and explicit surrogate variable analysis procedure by
  coefficient adjustment}.
\newblock \emph{Biometrika}, 104\penalty0 (2):\penalty0 303--316, 04 2017.
\newblock ISSN 0006-3444.
\newblock \doi{10.1093/biomet/asx018}.

\bibitem[Leek and Storey(2008)]{Leek2008}
Jeffrey~T. Leek and John~D. Storey.
\newblock A general framework for multiple testing dependence.
\newblock \emph{Proceedings of the National Academy of Sciences}, 105\penalty0
  (48):\penalty0 18718--18723, 2008.
\newblock ISSN 0027-8424.
\newblock \doi{10.1073/pnas.0808709105}.

\bibitem[McKennan and Nicolae(2019)]{McKennan19}
Chris McKennan and Dan Nicolae.
\newblock {Accounting for unobserved covariates with varying degrees of
  estimability in high-dimensional biological data}.
\newblock \emph{Biometrika}, 106\penalty0 (4):\penalty0 823--840, 09 2019.
\newblock ISSN 0006-3444.
\newblock \doi{10.1093/biomet/asz037}.

\bibitem[Ning and Liu(2017)]{ning2017general}
Yang Ning and Han Liu.
\newblock A general theory of hypothesis tests and confidence regions for
  sparse high dimensional models.
\newblock \emph{The Annals of Statistics}, 45\penalty0 (1):\penalty0 158--195,
  2017.

\bibitem[{Rudelson} and {Zhou}(2013)]{rz13}
M.~{Rudelson} and S.~{Zhou}.
\newblock Reconstruction from anisotropic random measurements.
\newblock \emph{IEEE Transactions on Information Theory}, 59\penalty0
  (6):\penalty0 3434--3447, June 2013.
\newblock ISSN 1557-9654.
\newblock \doi{10.1109/TIT.2013.2243201}.

\bibitem[Silva et~al.(2006)Silva, Scheines, Glymour, Spirtes, and
  Chickering]{silva2006learning}
Ricardo Silva, Richard Scheines, Clark Glymour, Peter Spirtes, and
  David~Maxwell Chickering.
\newblock Learning the structure of linear latent variable models.
\newblock \emph{Journal of Machine Learning Research}, 7\penalty0 (2), 2006.

\bibitem[Smith et~al.(2019)Smith, Hayamizu, Finger, Bello, McCright, Xu,
  Baldarelli, Beal, Campbell, Corbani, et~al.]{smith2019mouse}
Constance~M Smith, Terry~F Hayamizu, Jacqueline~H Finger, Susan~M Bello,
  Ingeborg~J McCright, Jingxia Xu, Richard~M Baldarelli, Jonathan~S Beal,
  Jeffrey Campbell, Lori~E Corbani, et~al.
\newblock The mouse gene expression database (gxd): 2019 update.
\newblock \emph{Nucleic acids research}, 47\penalty0 (D1):\penalty0 D774--D779,
  2019.

\bibitem[Stock and Watson(2002)]{SW2002}
James~H Stock and Mark~W Watson.
\newblock Forecasting using principal components from a large number of
  predictors.
\newblock \emph{Journal of the American Statistical Association}, 97\penalty0
  (460):\penalty0 1167--1179, 2002.
\newblock \doi{10.1198/016214502388618960}.

\bibitem[Valdar et~al.(2006)Valdar, Solberg, Gauguier, Burnett, Klenerman,
  Cookson, Taylor, Rawlins, Mott, and Flint]{valdar2006genome}
William Valdar, Leah~C Solberg, Dominique Gauguier, Stephanie Burnett, Paul
  Klenerman, William~O Cookson, Martin~S Taylor, J~Nicholas~P Rawlins, Richard
  Mott, and Jonathan Flint.
\newblock Genome-wide genetic association of complex traits in heterogeneous
  stock mice.
\newblock \emph{Nature genetics}, 38\penalty0 (8):\penalty0 879--887, 2006.

\bibitem[van~de Geer et~al.(2014)van~de Geer, Bühlmann, Ritov, and
  Dezeure]{vandegeer2014}
Sara van~de Geer, Peter Bühlmann, Ya’acov Ritov, and Ruben Dezeure.
\newblock On asymptotically optimal confidence regions and tests for
  high-dimensional models.
\newblock \emph{Ann. Statist.}, 42\penalty0 (3):\penalty0 1166--1202, 06 2014.
\newblock \doi{10.1214/14-AOS1221}.

\bibitem[Vershynin(2012)]{vershynin_2012}
Roman Vershynin.
\newblock \emph{Introduction to the non-asymptotic analysis of random
  matrices}, page 210–268.
\newblock Cambridge University Press, 2012.
\newblock \doi{10.1017/CBO9780511794308.006}.

\bibitem[Wang et~al.(2017)Wang, Zhao, Hastie, and Owen]{wang2017}
Jingshu Wang, Qingyuan Zhao, Trevor Hastie, and Art~B. Owen.
\newblock Confounder adjustment in multiple hypothesis testing.
\newblock \emph{Ann. Statist.}, 45\penalty0 (5):\penalty0 1863--1894, 10 2017.
\newblock \doi{10.1214/16-AOS1511}.

\bibitem[Wang and Blei(2019)]{wang2019blessings}
Yixin Wang and David~M Blei.
\newblock The blessings of multiple causes.
\newblock \emph{Journal of the American Statistical Association}, 114\penalty0
  (528):\penalty0 1574--1596, 2019.

\bibitem[Wielowieyski et~al.(1999)Wielowieyski, Brennan, and
  Jongstra]{wielowieyski1999tli1}
Andrzej Wielowieyski, Laurie~A Brennan, and Jan Jongstra.
\newblock Tli1, a resistance locus for carcinogen-induced t-lymphoma.
\newblock \emph{Mammalian genome}, 10\penalty0 (6):\penalty0 623--627, 1999.

\bibitem[Yuan and Lin(2006)]{yuanlin}
M.~Yuan and Y.~Lin.
\newblock Model selection and estimation in regression with grouped variables.
\newblock \emph{J. Roy. Statist. Soc. Ser. B}, 68:\penalty0 49--67, 2006.

\bibitem[Zhang and Zhang(2014)]{zhangzhang2014}
Cun-Hui Zhang and Stephanie~S. Zhang.
\newblock Confidence intervals for low dimensional parameters in high
  dimensional linear models.
\newblock \emph{Journal of the Royal Statistical Society: Series B (Statistical
  Methodology)}, 76\penalty0 (1):\penalty0 217--242, 2014.
\newblock \doi{https://doi.org/10.1111/rssb.12026}.

\end{thebibliography}
    }
    }

    \newpage

    
    \appendix 
    
    \section{Column-wise $\ell_2$ convergence rates of $X\wh F-XF$}\label{app_theory_fit}
   
    We first provide theoretical guarantees of $\bX\wh\bF-\bX\bF$ under the fixed design matrix $\bX$ as the analysis is still valid for random design by first conditioning on $\bX$. 
	Recall from model (\ref{model}) that $W$ is uncorrelated with $X$. To simplify the analysis under the fixed design scenario, we assume the independence between $X$ and $W$ in order to derive the deviation bounds of their cross product. We expect that the same theoretical guarantees hold under $\Cov(X,W)=0$ by using more complicated arguments.

    Recall that $\wh\bF = (\wh\bF_1, \ldots, \wh\bF_m)$ with $\wh\bF_j$ obtained from solving (\ref{def_est_F_j}) for $1\le j\le m$. The following lemma characterizes the solution $\wh\bF_j = \wh\btheta^{(j)} + \wh\bdelta^{(j)}$. It is proved in \cite{chernozhukov2017}.

    \begin{lemma}\label{lem_solution}
        For any $1\le j\le m$, let $(\wh \btheta^{(j)}, \wh \bdelta^{(j)})$ be any solution of (\ref{def_est_F_j}), and denote
		\begin{equation}\label{def_P_Q_lbd2}
		P_{\lbdj} = \bX\left(\bX^T\bX + n \lbdj \bI_p\right)^{-1}\bX^T,\qquad Q_{\lbdj} = \bI_n - P_{\lbdj}.
    	\end{equation}
		for any $\lbdj\ge 0$ such that $P_{\lbdj}$ exists. 
		Then $\wh \btheta^{(j)}$ is the solution of the following problem
		\begin{equation}\label{crit_Theta}
		\wh  \btheta^{(j)} = \arg\min_{\btheta\in \RR^p} {1\over n}\left\|Q_{\lbdj}^{1/2}(\bY_j - \bX\btheta)
		\right\|_2^2 + \lbdIj \|\btheta\|_{1},
		\end{equation}
		and $\wh \bdelta^{(j)}=(\bX^T\bX + n \lbdj \bI_p)^{-1}\bX^T(\bY_j - \bX\wh\btheta^{(j)})$, where $Q_{\lbdj}^{1/2}$ is the principal matrix square root of $Q_{\lambda_2}$. Moreover, we have
		\begin{equation}\label{fit}
		\bX\wh \bF_j = \bX\left(\wh \btheta^{(j)} + \wh \bdelta^{(j)}\right) = P_{\lbdj}\bY_j + Q_{\lbdj}\bX\wh \btheta^{(j)}. 
		\end{equation}
	\end{lemma}
    To analyze $\wh\bF_j$, we first introduce the Restricted Eigenvalue (RE) \citep{bickel2009}. For some given constant $\alpha \ge 1$ and integer $1\le s\le p$, define
		\begin{equation}\label{RE_X}
		\kappa(s, \alpha) = \min_{S\subseteq [p],|S|\le s}~\min_{\Delta \in \cC(S, \alpha)}{\|\bX \Delta\|_2\over \sqrt{n}\|\Delta_{S\cdot}\|_2},
		\end{equation}
	where  $\cC(S, \alpha) := \{\Delta\in \RR^{p} \setminus \{\b0\}: \alpha\|\Delta_{S}\|_{1}\ge \|\Delta_{S^c}\|_{1}\}$. For $1\le j\le m$ and the $j$th response regression, define
	\begin{equation}\label{def_V_eps}
	 \sigma_j^2 = \g_w^2 \bB_j^T\Sigma_W \bB_j + \g_e^2 \sigma_{E_j}^2
	\end{equation}
	where $\g_w$ and $\g_e$ are the sub-Gaussian constants defined in Assumption \ref{ass_error} and $\sigma_{E_j}^2 = [\se]_{jj}$. Write $M^{(j)} = n^{-1}\bX^T Q_{\lbdj}^2\bX$ with $Q_{\lbdj}$ defined in (\ref{def_P_Q_lbd2}). Recall that $\wh\Sigma = n^{-1}\bX^T\bX$ and its eigenvalue are $\Lambda_1\ge \Lambda_2 \ge \cdots \ge \Lambda_q>0$ with $q = \rank(\bX)$. Further recall $s_n$ is defined in (\ref{def_space_Theta}). 
	The following theorem provides the $\ell_2$ convergence rate of $\bX\wh\bF_j - \bX\bF_j$ uniformly over $1\le j\le m$. 
	
	\begin{theorem}\label{thm_pred}
		Under Assumptions \ref{ass_error}, assume $\kappa(s_n,4)>0$ and choose 
		\begin{equation}\label{rate_lbd1}
		\lambda_{1}^{(j)} = 4\sigma_j\sqrt{6\max_{1\le i\le p}M_{ii}^{(j)}}\sqrt{\log (p\vee m) \over n}
		\end{equation}
		and any $\lambda_2^{(j)} \ge 0$
		in (\ref{def_est_F_j}) such that $P_{\lambda_2^{(j)}}$ exists. 
		With probability $1-2(p\vee m)^{-1} - m^{-1}$, 
		\begin{align*}
		{1\over n}\left\|\bX \wh \bF_j - \bX \bF_j\right\|_2^2 
		&\lesssim \inf_{\substack{(\btheta_0, \bdelta_0):\\
				\btheta_0 + \bdelta_0 = \bF_j}} \Bigl[Rem_{1,j} + Rem_{2,j}(\bdelta_0) +  Rem_{3,j}(\btheta_0)\Bigr]
		\end{align*}
		holds uniformly over $1\le j\le m$, 
		where 
		\begin{align*}
		&Rem_{1,j} = \left(\tr\left(P_{\lambda_2^{(j)}}^2\right) + \left\|P_{\lambda_2^{(j)}}^2\right\|_{\op}\log m  \right){\sigma_j^2 \over n}\\
		&Rem_{2,j}(\bdelta_0) = \lambda_2^{(j)}~ \bdelta_0^T \wh \Sigma(\wh \Sigma + \lambda_2^{(j)} \bI_p)^{-1}\bdelta_0\\
		&Rem_{3,j}(\btheta_0) =  
		{\lbdj(\Lambda_1 + \lambda_2^{(j)}) \over (\Lambda_q+\lambda_2^{(j)})^2}\left(\max_{1\le i\le p}\wh\Sigma_{ii}\right) {s_0\log (p\vee m) \over \kappa^2(s_n,4)}{\sigma_j^2\over n}.
		\end{align*}
	\end{theorem}
	\begin{proof}
    	Theorem \ref{thm_pred} can be proved by using the line of arguments in the proof of Theorem 4 in \cite{bing2020adaptive} except for working on the following event 
	    \beq\label{def_event_lbd} 
	        \cE := \bigcap_{i=1}^p\bigcap_{j=1}^m\left\{
	            \left|
	             \bX_i^T Q_{\lbdj} \bepsilon_j
	            \right| \le {n\over 4}\lbdIj
	        \right\}
	    \eeq 
	    with $\lbdIj$ defined in (\ref{rate_lbd1}). To establish $\PP(\cE)$, pick any $1\le i\le p$ and $1\le j\le m$. 
		We first note that, by the independence of  $\bepsilon_{tj}$ for $1\le t\le n$, 
		$\bepsilon_j^TQ_{\lbdj}\bX_i$ is sub-Gaussian with sub-Gaussian parameter
		\[
		    \sigma_j\sqrt{\bX_i^TQ_{\lbdj}^2 \bX_i} = \sigma_j\sqrt{nM_{ii}^{(j)}}.
		\]
		Thus, the basic tail inequality of sub-Gaussian random variable yields  
		\[
		\PP\left\{
		\left|\bX_i^T Q_{\lbdj} \bepsilon_j\right| > t\sigma_j\sqrt{nM_{ii}^{(j)}}
		\right\}\le 2e^{-t^2/2},\quad \text{for all }t\ge0.
		\]
		Choose $t = \sqrt{6\log(p\vee m)}$ and take the union bounds over $1\le i\le p$ and $1\le j\le m$ to obtain 
		$\PP(\cE) \ge 1 - 2(p\wedge m)^{-1}.$
	\end{proof}
    
     We remark that Theorem \ref{thm_pred} in particular holds for the true $\btheta_j = \bTheta_j$ and $\bdelta_j = \bA\bB_j$, for $1\le j\le m$, whenever they are identifiable. 
     
    

    \section{Main proofs}
    
    \subsection{Proof of Theorem \ref{thm_ident}: identifiability}
        From model (\ref{model_linear}) and noting that $\Cov(X, \epsilon) = \b0$, $\bTheta + \bA\bB$ can be identified from $[\Cov(X)]^{-1}\Cov(X,Y)$, and so is $\seps$. Let $\bU_K\in\RR^{m\times K}$ denote the first $K$ eigenvectors of $\seps$. An application of the Davis Kahan Theorem yields
        \[
            \|\bU_K\bU_K^T - P_B\|_{\op} \le {\sqrt{2}\|\se\|_{\op} \over \lambda_K(\bB^T\sw\bB)} = o(1)
        \]
        under condition (\ref{ident_conds}). Thus, $P_B^{\perp}$ is recovered asymptotically and so is $\bTheta P_B^{\perp}  = (\bTheta + \bA\bB)P_B^{\perp}$. Finally, for each $1\le i\le p$ and $1\le j\le m$, since under condition (\ref{cond_ident_Theta}),
        \beq\label{bd_bias}
            |\bTheta_{i\cdot}^TP_B\be_j| =~ & \left|\bTheta_{i\cdot}^T\bB^T\sw^{1/2}\left(\sw^{1/2}\bB\bB^T\sw^{1/2}\right)^{-1}\sw^{1/2}\bB\be_j\right|\\
            \leq ~ &\norm{\bTheta_{i\cdot}}_1\left\|\bB^T\sw^{1/2}\left(\sw^{1/2}\bB\bB^T\sw^{1/2}\right)^{-1}\right\|_{\i,2}\left\|\sw^{1/2}\bB\be_j\right\|_2\\
            \leq ~ &\norm{\bTheta_{i\cdot}}_1\max_{1\leq \ell \leq m}\norm{\sw^{1/2}\bB_\ell}_2[\lambda_{K}(\bB^T \sw \bB)]^{-1}\norm{\sw^{1/2}\bB_j}_2\\
            =~&\cO\left({\norm{\bTheta_{i\cdot}}_1\over m}\right),
        \eeq 
        we conclude that 
        \[
            \bTheta_{ij} = [\bTheta P_B^{\perp}]_{ij} + [\bTheta P_B]_{ij} = [\bTheta P_B^{\perp}]_{ij} + o(1).
        \]
        This completes the proof. \qed

              \subsection{Proof of Theorem \ref{thm_rates_B}: The uniform convergence rate of $\wh B_j$}\label{app_proof_thm_B}
       
       Recall from (\ref{def_svd_epsilon}) that 
        \[
            {1\over nm}\wh\bepsilon^T\wh\bepsilon = \bV\bD^2\bV^T.
        \]
        We work on the intersection of the events
        \begin{align}\label{def_event_F}
            \cE_F &:= \left\{
            \max_{1\le j\le m}{1\over n}\|\bX\wh\bF_j- \bX\bF_j\|_2^2 \lesssim r_n \right\},\\\label{def_event_D}
            \cE_D &:= \left\{
                \sqrt{c_Wc_B} \lesssim  \lambda_K(\bD_K)\le \lambda_1(\bD_K) \lesssim
               \sqrt{C_WC_B}
            \right\},    
        \end{align}
        with $r_n$ defined in Assumption \ref{ass_initial} and $c_B, C_B, c_W, C_W$ defined in Assumption \ref{ass_B_Sigma}.
        Lemma \ref{lem_D_K} and Assumption \ref{ass_initial} guarantee that $\lim_{n\to\i}\PP(\cE_F\cap \cE_D) = 1$.
        
        By (\ref{def_est_BW}), observe that 
        \begin{align*}
            {1\over nm}\wh\bepsilon^T\wh\bepsilon \wh\bB^T = \bV\bD^2\bV^T\sqrt{m}\bV_K\bD_K = \wh\bB^T\bD_K^2.
        \end{align*}
        Plugging 
        \begin{equation}\label{def_bDelta}
            \wh\bepsilon = \bY -  \bX\wh\bF = \bepsilon + \underbrace{\bX\bF - \bX\wh\bF}_{\bDelta}
        \end{equation}
        into the above display yields 
        \begin{align*}
            {1\over nm}\left(
            \bepsilon^T\bepsilon + \bepsilon^T\bDelta + \bDelta^T\bepsilon + \bDelta^T\bDelta
            \right) \wh \bB^T \bD_K^{-2} = \wh\bB^T.
        \end{align*}
        Since 
        \[
             {1\over nm}
            \bepsilon^T\bepsilon   = {1\over nm}\left(
            \bB^T\bW^T\bW\bB + \bB^T\bW^T\bE + \bE^T\bW\bB + \bE^T\bE 
            \right),
        \]
        using the definition in (\ref{def_H0}) gives
        \begin{align}\label{display_B_hat_BH}\nonumber
            &\wh \bB^T - \bB^T\bH_0^T\\ 
            &= {1\over nm}\left(
            \bB^T\bW^T\bE + \bE^T\bW\bB + \bE^T\bE + \bepsilon^T\bDelta + \bDelta^T\bepsilon + \bDelta^T\bDelta
            \right) \wh \bB^T \bD_K^{-2}\\\nonumber
            &= {1\over n\sqrt m}\left(
            \bB^T\bW^T\bE + \bE^T\bW\bB + \bE^T\bE + \bepsilon^T\bDelta + \bDelta^T\bepsilon + \bDelta^T\bDelta
            \right)\bV_K\bD_K^{-1},
        \end{align}
        where we used (\ref{def_est_BW}) in the last step. Pick any $1\le j\le m$ and multiply both sides of the above display by $\be_j$. We proceed to bound each corresponding terms on the right hand side. 
        
        First, invoking Lemma \ref{lem_quad_terms} and $\cE_D$ gives 
        \begin{align*}
            \left\|\be_j^T\bB^T\bW^T\bE\bV_K\bD_K^{-1}\right\|_2 \lesssim \|\bB_j^T\bW^T\bE\|_2 \lesssim \sqrt{nm\log m}
        \end{align*}
        with probability at least $1-8m^{-1}$. Similarly, we obtain 
        \begin{align*}
            {1\over n\sqrt m}\left\|\be_j^T\left(
            \bB^T\bW^T\bE + \bE^T\bW\bB + \bE^T\bE 
            \right)\bV_K\bD_K^{-1}\right\|_2 \lesssim \sqrt{\log m \over n\wedge m}.
        \end{align*}
        On the other hand, Lemma \ref{lem_quad_terms_Delta} together with Assumption \ref{ass_initial} ensures that, with probability $1- 8m^{-1}$,
        \begin{align}\label{bd_quad_Deltas}
            &{1\over n\sqrt m}\left\|\be_j^T\left(
            \bepsilon^T\bDelta + \bDelta^T\bepsilon + \bDelta^T\bDelta
            \right)\bV_K\bD_K^{-1}\right\|_2\\\nonumber
            & \lesssim ~ \sqrt{r_n}\sqrt{Rem_{1,j} + Rem_{2,j}(\bdelta_j) + Rem_{3,j}(\btheta_j)} + r_{n,1} +  \sqrt{r_{n,2}\log (m) \over n}+ r_{n,3}\sqrt{1\over n}\\\nonumber
            & \lesssim ~ r_n
       \end{align}
        uniformly over $1\le j\le m$. Here, for convenience, we write 
        \begin{align}\label{def_r_n_k}
            &r_{n,1} = \max_{1\le j\le m}Rem_{1,j},\quad r_{n,2} = \max_{1\le j\le m}Rem_{2,j}(\bdelta_j),\quad r_{n,3} = \max_{1\le j\le m}Rem_{3,j}(\btheta_j).
        \end{align}
        Collecting the previous three displays concludes the desired rate. The proof is completed by noting that $m = m(n) \to \infty$ whence the probabilities tend to one as $n\to \infty$. 
        \qed

    \subsection{Proof of Lemma \ref{thm_Theta_simple_rates}: $\ell_1$ convergence rate of the initial estimator $\wh\Theta_1$}\label{app_proof_thm_Theta}

        Recall $\wh \Sigma = n^{-1}\bX^T \bX$ and $\kappa(s_n,4)$ is defined in (\ref{RE_X}). Define the following event 
        \begin{align}\label{def_event_X}
            \cE_{\bX} := \left\{\kappa(s_n, 4) \ge c,~ \max_{1\le j\le p}\wh\Sigma_{jj}\le C, ~  {1\over \sqrt n}\|\bX\bTheta\|_{2,1} \le C' M_n\sqrt{s_n}, ~  {1\over \sqrt n}\|\bX\bA\|_{\op} \le C' \right\}
        \end{align}
        for  some finite constants $C\ge c>0$ and  $C'>0$. Lemma \ref{lem_X} in Appendix \ref{app_lemmas_Theta} proves that $\lim_{n\to \infty}\PP(\cE_{\bX}) = 1$ under the conditions of Theorem \ref{thm_Theta_simple_rates}. 
        Recall $r_n$ from Assumption \ref{ass_initial}. Define  
        \begin{equation}\label{def_eta_bar}
            \eta_{n} = \sqrt{\log m \over n\wedge m} +  r_n.
        \end{equation}
        Further recall $\wt\bB$ and $\bH_0$ are defined in (\ref{def_B_tilde}) and (\ref{def_H0}). We work on the event
        \begin{equation}\label{def_event_misc}
            \cE_{\bX} \cap \left\{
                \|(\wt \bB - \wh\bB)\wh P_B^{\perp}\be_1\|_2 \lesssim \eta_n
            \right\} \cap \left\{
                \|(\wh P_B-P_B)\be_1\|_\i \lesssim  {\eta_n \over m}
            \right\}\cap \left\{
                \lambda_K(\bH_0) \gtrsim c_H
            \right\}
       \end{equation}
        which, according to Lemmas \ref{lem_X}, \ref{lemma_technical} and \ref{lemma_PB_error}, holds with probability tending to one. 
        
        Recall that $\bar \bTheta_1 = \bTheta P_{B}^{\perp} \be_1$.
        Starting with 
        \[
           \frac{1}{n}\bnorm{\wt\by - \bX\wh  \bTheta_1}_2^2 + \lambda_3\norm{\wh \bTheta_1}_1 \le \frac{1}{n}\bnorm{\wt\by - \bX\bar \bTheta_1}_2^2 + \lambda_3\norm{\bar \bTheta_1}_1,
        \]
        work out the squares to obtain
        \begin{align*}
            {1\over n}\left\|
                \bX (\wh\bTheta_1 - \bar \bTheta_1)
            \right\|_2^2 \le {2\over n}\left|
            \langle  \bX (\wh\bTheta_1 - \bar \bTheta_1), \wt \by  - \bX  \bar \bTheta_1
            \right|
            +\lambda_3\norm{\bar \bTheta_1}_1 -  \lambda_3\norm{\wh \bTheta_1}_1.
        \end{align*}
        By noting that 
        \begin{align*}
            \wt\by - \bX  \bar \bTheta_1 & = \left[\bX(\bTheta + \bA\bB) + \bW\bB + \bE\right] \wh P_B^{\perp}\be_1 - \bX  \bTheta P_B^{\perp}\be_1\\
            &= \bX\bA\bB\wh P_B^{\perp}\be_1+ \bW\bB\wh P_B^{\perp}\be_1 + \bE \wh P_B^{\perp}\be_1  + \bX  \bTheta(\wh P_B^{\perp}-P_{B}^{\perp}) \be_1
        \end{align*}
        and by writing $\bDelta = \wh\bTheta_1 - \bar\bTheta_1$, 
        we have 
        \begin{align*}
            {2\over n}\left|
            \langle  \bX \bDelta, \wt \by  - \bX  \bar \bTheta_1
            \right| &\le  {2\over n}\left| \be_1^T \wh P_B^{\perp}\bE^T  \bX \bDelta\right|+ {2\over n}\left\|\bX\bDelta\right\|_2 Rem\\
            &\le {2\over n}\left\| \be_1^T \wh P_B^{\perp}\bE^T  \bX\right\|_\i \|\bDelta\|_1 + {2\over n}\left\|\bX\bDelta\right\|_2 Rem.
        \end{align*}
        where 
        \[
           Rem = {1\over \sqrt n}\left\|
            \bX\bA\bB\wh P_B^{\perp}\be_1+ \bW\bB\wh P_B^{\perp}\be_1 +\bX  \bTheta (P_B - \wh P_B) \be_1\right\|_2.
        \]
        Provided that 
        \begin{equation}\label{bd_lbd3}
            \left\| \be_1^T \wh P_B^{\perp}\bE^T  \bX\right\|_\i \le {n\over 4}\lambda_3, 
        \end{equation}
        from the fact that $\|\bar\bTheta_1\|_0 \le s_n$, 
        using $\norm{\bar\bTheta_1}_1 - \norm{\wh \bTheta_1}_1 \le \|\bDelta_S\|_1 + \|\bDelta_{S^c}\|_1$ with $S := \supp(\bar\bTheta_1)$ and $|S| \le s_n$ gives 
        \begin{align*}
             {1\over n}\left\|
                \bX \bDelta
            \right\|_2^2 \le {2\over n}\left\|\bX\bDelta\right\|_2 Rem + {3 \over 2}\lambda_3\|\bDelta_S\|_1 - {1 \over 2}\lambda_3\|\bDelta_{S^c}\|_1.
        \end{align*}
        We now bound from above $Rem$. By recalling that $\wt\bB = \bH_0 \bB$,  
        \begin{align*}
            {1\over \sqrt n}\left\|
            \bX\bA\bB\wh P_B^{\perp}\be_1\right\|_2 &= {1\over \sqrt n}\left\|
            \bX\bA\bH_0^{-1}(\wt\bB - \wh\bB)\wh P_B^{\perp}\be_1\right\|_2\\
            &\le  {1\over \sqrt n}\left\|
            \bX\bA\bH_0^{-1}\right\|_{\op}\left\|(\wt\bB-\wh\bB)\wh P_B^{\perp}\be_1\right\|_2\\
            &\lesssim {1\over \sqrt n}\left\|
            \bX\bA\right\|_{\op}\eta_n & \textrm{by }(\ref{def_event_misc})\\
            &\lesssim \eta_n & \textrm{by }(\ref{def_event_X}).
        \end{align*}
        By (\ref{def_event_misc}), we also have 
        \begin{align*}
             {1\over \sqrt n}\left\|\bX  \bTheta (\wh P_B -P_B)\be_1\right\|_2 &\le 
              {1\over \sqrt n}\left\|\bX\bTheta\right\|_{2,1} \left\|(\wh P_B -P_B)\be_1\right\|_\i \lesssim  {M_n\sqrt{s_n} \over m}\eta_n.
        \end{align*}
        Together with Lemma \ref{lem_W_eigens}, we also have
        \begin{align*}
            {1\over \sqrt n}\left\|
            \bW\bB\wh P_B^{\perp}\be_1\right\|_2 & \lesssim 
            {1\over \sqrt n}\left\|
            \bW\right\|_{\op}\left\|(\wt \bB - \wh\bB)\wh P_B^{\perp}\be_1\right\|_2 \lesssim \eta_n
        \end{align*}
        with probability $1 - 2e^{-n}$. 
        We thus conclude that with the same probability, on the event (\ref{def_event_misc}), 
        \[
            Rem \lesssim  \eta_n \left(1 + {M_n\sqrt{s_n}\over m}\right).
        \]
        Following the same line of arguments as the proof of Theorem 6 in \cite{bing2020adaptive}, it is straightforward to show that, on the event (\ref{def_event_misc}) and for any $\lambda_3$ such that (\ref{bd_lbd3}) holds,
        \begin{equation}\label{rate_Theta_td}
		    \|\wh \bTheta_1 - \bar\bTheta_1\|_{1} 
    		\lesssim \max\left\{{\lambda_3},  ~ {(\wt \lambda_3)^2 \over \lambda_3 }\right\}{s_n\over \kappa^2(s_n,4)},
		\end{equation}
		holds with probability $1-2e^{-n}$,    		where
		\begin{equation}\label{def_lbd3_td}
		\wt\lambda_3 =  \eta_n  \left(1 + {M_n\sqrt{s_n}\over m}\right) {\kappa(s_n,4) \over \sqrt{s_n}}.
	    \end{equation}
	    It remains to show (\ref{bd_lbd3}) holds with probability tending to one for 
        any 
        \begin{align}\label{def_event_lbd3}
		    \lambda_3  \ge  \bar\lambda_3  \asymp  
		    \sigma_{E_1}\sqrt{\max_{1\le j\le p} \wh\Sigma_{jj}}\sqrt{\log p\over n}.
		\end{align}
		If this holds, then observe that (\ref{def_event_lbd3}), (\ref{rate_Theta_td}) and (\ref{def_lbd3_td}) readily imply 
		\beq\label{rate_Theta_td_prime}
        \|\wh \bTheta_1 -  \bar\bTheta_1\|_{1} 
		  &\lesssim (\bar \lambda_3 \vee \wt\lambda_3) {s_n \over \kappa^2(s_n,4)}\\
		  &\lesssim  s_n\sqrt{\log p\over n} +  \left(\sqrt{s_n} + {M_n s_n\over m}\right)\eta_n
        \eeq
        by choosing $\lambda_3$ appropriately. The result immediately follows from (\ref{def_eta_bar}).

		To prove (\ref{bd_lbd3}) holds for any $\lambda_3\ge \bar\lambda_3$, note that
        \begin{align*}
            \left\|\be_1^T \wh P_B^{\perp}\bE^T  \bX\right\|_\i &\le \left\| \be_1^T \bE^T  \bX\right\|_\i + \left\| \be_1^T \wh P_B\bE^T  \bX\right\|_\i\\
            &\le \left\| \be_1^T \bE^T  \bX\right\|_\i + \left\| \be_1^T\wh P_{B}\right\|_2\left\|\bE^T  \bX\right\|_{2,\i}.
        \end{align*}
        Since $\bE_1^T\bX_j$ is $\g_e\sqrt{n\wh\Sigma_{jj}[\se]_{11}}$ sub-Gaussian, the sub-Gaussian tail probability together with union bounds over $1\le j\le p$ yields
        \[
            \PP\left\{
            \left\| \be_1^T \bE^T  \bX\right\|_\i \le 2\g_e\sqrt{n\log p}\sqrt{[\se]_{11}\max_{1\le j\le p}\wh\Sigma_{jj}}
            \right\} \ge 1-2p^{-1}.
        \]
        Furthermore, noting that 
        \[
            \left\|\bE^T  \bX\right\|_{2,\i}^2 = \max_{1\le j\le p}\bX_j^T\bE \se^{-1/2}\se \se^{-1/2}\bE\bX_j
        \]
        and $\bX_j\bE\se^{-1/2}$ is $\g_e\sqrt{n\wh\Sigma_{jj}}$ sub-Gaussian, 
        an application of Lemma \ref{lem_quad} with union bounds over $1\le j\le p$ gives 
        \[
            \PP\left\{
                \left\|\bE^T  \bX\right\|_{2,\i}^2 \le \g_e^2 n \max_{1\le j\le p}\wh\Sigma_{jj}\left(
                \sqrt{\tr(\se)} + \sqrt{4\|\se\|_{\op} \log p}
                \right)^2
            \right\} \ge 1-p^{-1}.
        \]
        By part (E) of Lemma \ref{lemma_technical}, we conclude that 
        \[
            \PP\left\{
             {1\over n}\left\|\be_1^T \wh P_B^{\perp}\bE^T  \bX\right\|_\i 
            \lesssim \g_eC_E\sqrt{\max_{1\le j\le p}\wh\Sigma_{jj}}\sqrt{\log p\over n} 
            \right\}  \ge 1- 3p^{-1}
        \]
        where 
        \[
            C_E = \sqrt{[\se]_{11}} + \sqrt{\tr(\se)\over m\log p} + \sqrt{\|\se\|_{\op}\over m} \lesssim 1.
        \]
        This completes the proof. \qed

    \subsection{Proof  of Theorem \ref{thm_asymp_normal}: asymptotic normality of $\wt \Theta_{11}$}\label{app_thm_asymp_normal}
    
    Recall that $\bar\bTheta_1 = \bTheta P_B^{\perp}\be_1$ so that $\bar\Theta_{11} = \be_1^T \bTheta P_B^{\perp}\be_1$. 
    By the definition of $\wt{\Theta}_{11}$ and $\bar{\Theta}_{11}$, we have
    \beq
    \wt{\Theta}_{11} - \bar\Theta_{11} & = \wh{\Theta}_{11} - \bar\Theta_{11}+ \wh{\bomega}_1^T\frac{1}{n}\bX^T(\wt\by - \bX\wh{\bTheta}_1)\\
    & =  \underbrace{(\be_1 - \frac{1}{n}\bX^T\bX\wh{\bomega}_1)^T(\wh{\bTheta}_1 - \bar \bTheta_1)}_{I_1} + \wh{\bomega}_1^T\frac{1}{n}\bX^T(\wt\by_1 - \bX\bar\bTheta_1)\\
    & = I_1 + \wh{\bomega}_1^T\frac{1}{n}\bX^T\left[(\bX(\bTheta + \bA\bB) + \bW \bB + \bE)\wh{P}_{B}^\perp \be_1 - \bX\bTheta P_B^\perp \be_1\right] \\
    & = I_1 + \underbrace{\wh{\bomega}_1^T\frac{1}{n}\bX^T\bX\bTheta(\wh{P}_B^\perp - P_B^\perp)\be_1}_{I_2}
    + \underbrace{\wh{\bomega}_1^T\frac{1}{n}\bX^T\bX \bA\bB \wh{P}_B^\perp \be_1}_{I_3}\\
    &~~~~~~~~~~~~~~~~~~+  \underbrace{\wh{\bomega}_1^T\frac{1}{n}\bX^T\bW \bB \wh{P}_B^\perp \be_1}_{I_4}
    +\underbrace{ \wh{\bomega}_1^T\frac{1}{n}\bX^T\bE \wh{P}_B^\perp \be_1}_{I_5}\\
    =& I_1 + I_2 + I_3 + I_4 + I_5.
    \eeq 
    In what follows, we will characterize $I_1$ through $I_5$, respectively. For simplicity, define 
    \begin{equation}\label{def_xi}
        \xi_n = s_n\sqrt{\log p\over n} + \left({s_nM_n\over m} + \sqrt{s_n}\right)\left(\sqrt{\log m \over n} + r_n\right)
    \end{equation}
    such that $\|\wh\bTheta_1 - \bar\bTheta_1\|_1 = \cO_{\PP}(\xi_n)$ from Theorem \ref{thm_Theta_simple_rates}.
    \begin{itemize}
	\item For $I_1$, the KKT condition of (\ref{formula_nodewise}) implies that \citep{vandegeer2014} 
    \[
    \left\|\frac{1}{n}\bX^T\bX \wh{\bomega}_1 - \be_1\right\|_\infty\leq \frac{\wt{\lambda}}{2\wh\tau_1^2},
    \]
    which, together with Lemma \ref{lemma_nodewise} and Theorem \ref{thm_Theta_simple_rates}, yields 
	\beq\label{bd_I1}
	|I_1| \leq \norm{\wh{\bTheta}_1 - \bar\bTheta_1}_1\norm{\be_1 - \frac{1}{n}\bX^T\bX\wh{\bomega}}_{\infty} = \cO_{\PP}\left(\xi_n \sqrt{\log p\over n}\right). 
	\eeq 
	
	
	\item For $I_2$, direct calculation gives us
	\beq\nonumber
	I_2 &=  (\be_1 - \frac{1}{n}\bX^T\bX \wh{\bomega}_1)^T    \bTheta( P_B - \wh{P}_B)\be_1 + \bTheta_{1\cdot}^T(P_B - \wh{P}_B)\be_1\\
	&=  I_{21} + I_{22}.
	\eeq 
	Recall that $\eta_n$ is defined in (\ref{def_eta_bar}). We have
	\beq \nonumber
	I_{21}\leq \norm{\be_1 - \frac{1}{n}\bX^T\bX\wh{\bomega}_1}_\infty
	\norm{\bTheta}_{1,1}\norm{(P_B - \wh{P}_B)\be_1}_\infty = \cO_{\PP}\left( \frac{s_n M_n\eta_n}{m}\sqrt{\log p\over n}\right),
	\eeq
	where the last step follows from Lemma \ref{lemma_PB_error}, Lemma \ref{lemma_nodewise} and $\|\bTheta\|_{1,1} \le s_n\|\bTheta\|_{\i,1} \le s_nM_n$ from (\ref{def_space_Theta}).
	Similarly, we can show that
	\beq \nonumber
	|I_{22}| \leq \norm{\bTheta_{1\cdot}}_1 \norm{(P_B - \wh{P}_B)\be_1}_\infty = 
	\cO_{\PP}\left(\frac{M_n\eta_n}{m}\right),
	\eeq
	and therefore 
	\beq\label{bd_I2}
	|I_2| = \cO_{\PP}\left(\left(1 + s_n\sqrt{\log p\over n}\right)\frac{M_n\eta_n}{m}\right) =  \cO_{\PP}\left(\frac{M_n\eta_n}{m}\right).
	\eeq 
	\item For $I_3$,  recall from (\ref{def_H0}) and (\ref{def_B_tilde}) that 
    $\bA\bB = \wt\bA\wt\bB:=(\bA\bH_0^{-1})(\bH_0\bB)$ on the event 
    \[
         \cE_{H} = \left\{
            c_H\lesssim \lambda_K(\bH_0) \le \lambda_1(\bH_0) \lesssim C_H
         \right\}
    \]
    with $c_H$ and $C_H$ defined in Lemma \ref{lemma_technical}.  On the event $\cE_H$, we obtain
	\beq\nonumber
	|I_3| &=  |\wh{\bomega}_1^T\frac{1}{n}\bX^T\bX \wt\bA \wt \bB \wh{P}_B^\perp \be_1|\\
	&\leq \norm{\wh{\bomega}_1^T\frac{1}{n}\bX^T\bX \wt\bA}_2\norm{(\wt \bB - \wh\bB)\wh{P}_B^\perp \be_1}_2 & \textrm{by }\wh{\bB}\wh{P}_B^\perp = \b0\\
	&\lesssim  c_H^{-1}\norm{\wh{\bomega}_1^T\frac{1}{n}\bX^T\bX \bA}_2\norm{(\wt \bB - \wh\bB)\wh{P}_B^\perp \be_1}_2.
	\eeq 
	Notice that $\lim_{n\to\infty}\PP(\cE_H) = 1$ and $\norm{(\wt \bB - \wh\bB)\wh{P}_B^\perp \be_1}_2 = \cO_{\PP}(\bar\eta)$ from parts (A) and (D) of Lemma \ref{lemma_technical}, respectively. We bound from above $\norm{\wh{\bomega}_1^T\frac{1}{n}\bX^T\bX \bA}_2$ as 
	\beq\nonumber
	\norm{\wh{\bomega}_1^T\frac{1}{n}\bX^T\bX \bA}_2
	&\leq\norm{(\be_1 - \frac{1}{n}\bX^T\bX\wh{\bomega}_1)^T\bA}_2 + \norm{\bA_{1\cdot}}_2\\
	& = \cO_{\PP}\left(\sqrt{s_\Omega\log p \over n}\right) + \norm{\bA_{1\cdot}}_2
	\eeq 
	where the last step uses Lemma \ref{lemma_nodewise_A}. We thus conclude
	\beq\label{bd_I3}
	    |I_3| = \cO_{\PP}\left(\eta_n\sqrt{s_\Omega \log p\over n} + \eta_n\norm{\bA_{1\cdot}}_2\right).
	 \eeq 

	\item For $I_4$, on the event $\cE_H$ and by writing $\wt\bW = \bW\bH_0^{-1}$, 
	\beq\nonumber
	    |I_4| \leq  \norm{\wh{\bomega}_1^T\frac{1}{n}\bX^T\wt\bW}_2\norm{(\wt \bB - \wh\bB)\wh{P}_B^\perp \be_1}_2 \lesssim  c_H^{-1}\norm{\wh{\bomega}_1^T\frac{1}{n}\bX^T\bW}_2\cO_{\PP}(\eta_n).
	\eeq 
	Note that, conditioning on $\bX$, $\wh{\bomega}_1^T\bX^T\bW\sw^{-1/2}\in\RR^K$ is $\gamma_w\sqrt{\wh\bomega_1^T\bX^T\bX\wh\bomega_1}$ sub-Gaussian random vector. An application of Lemma \ref{lem_quad} yields, for all $t>0$,
	\begin{align*}
	    \PP\left\{
	        \norm{\wh{\bomega}_1^T\bX^T\bW}_2^2 > \gamma_w^2(\wh\bomega_1^T\bX^T\bX\wh\bomega_1)\left(
	        \sqrt{\tr(\sw)} + \sqrt{2\|\sw\|_{\op}t}
	        \right)^2
	    \right\} \le e^{-t}.
	\end{align*}
	Note that 
	\beq\label{bd_omegaXXomega}
	    {1\over n}\wh\bomega_1^T\bX^T\bX\wh\bomega_1 &\le \Omega_{11} + \left|\wh\bomega_1^T {1\over n}\bX^T\bX \wh\bomega_1 - \Omega_{11}\right|\\
	    &=\cO_{\PP}\left(
	    \Omega_{11} + \sqrt{s_\Omega \log p\over n}
	    \right) & \textrm{ by Lemma \ref{lemma_nodewise}}\\
	    & = \cO_{\PP}(\Omega_{11})
	\eeq
	by using $s_\Omega \log p = o(n)$ and $\Omega_{11} \ge \Sigma_{11}^{-1} \ge C^{-1}$ from Assumption \ref{ass_X}. By also noting that 
	\beq\label{bd_Omega_11}
	    \Omega_{11} \le {1\over \lambda_{\min}(\Sigma)} = \cO(1)
	\eeq
	from Assumption \ref{ass_X}, from $\tr(\sw) \le K\|\sw\|_{\op} = \cO(1)$ and (\ref{bd_omegaXXomega}), we conclude 
	\beq\nonumber
	    \left\|\wh{\bomega}_1^T\frac{1}{n}\bX^T\bW\right\|_2 = \cO_{\PP}\left(
	     1/\sqrt{n}
	    \right).
	\eeq
	Hence 
	\beq\label{bd_I4}
	    I_4 =\cO_{\PP}\left(
	    \eta_n \over \sqrt{n}
	    \right).
	\eeq 
	
	\item For $I_5$, by definition
	\beq\nonumber
	\wh{\bomega}_1^T\frac{1}{n}\bX^T\bE \wh{P}_B^\perp \be_1 =& 
	\wh{\bomega}_1^T\frac{1}{n}\bX^T\bE P_B^\perp \be_1 + \wh{\bomega}_1^T\frac{1}{n}\bX^T\bE(P_B - \wh{P}_B)\be_1\\ 
	:=& I_{51} + I_{52}.
	\eeq 
	It's easy to see that $\bE \wh{P}_B^\perp \be_1 \in \RR^n$ is an i.i.d Gaussian vector with covariance matrix $V_{11}\bI_{n}$ and independent of $\bX$, where 
	\[
	    V_{11} := \be_1^T P_B^\perp\Sigma_{E}P_B^\perp\be_1.
	\]
	This implies that 
	\beq\nonumber
	\sqrt{n}I_{51} ~ \big | ~ \bX \sim N\left(0,    \wh{\bomega}_1^T\frac{1}{n}\bX^T\bX\wh{\bomega}_1~ V_{11} \right).
	\eeq
	We further note that 
	\beq\label{bd_V11}
	    V_{11} = [\se]_{11} - \be_1^T P_B\se \be_1-\be_1^T P_B\se P_B^{\perp} \be_1 = [\se]_{11} + \cO(1/\sqrt{m})
	\eeq
	by using $\|P_B\be_1\|_2 = \cO(1/\sqrt m)$ deduced from (\ref{bd_bias}). Hence, also by (\ref{bd_omegaXXomega}) and (\ref{bd_Omega_11}), 
	\beq\label{bd_I51}
	        	\sqrt{n}I_{51} = \zeta + o_\PP(1)
	\eeq
	where 
	\beq\label{def_zeta}
	    \zeta | \bX \sim  N\left(0,    \wh{\bomega}_1^T\frac{1}{n}\bX^T\bX\wh{\bomega}_1~ [\Sigma_E]_{11} \right).
	\eeq
	
	For the second term, we know
	\beq\nonumber
	|I_{52}|\leq |\wh{\bomega}_1^T\frac{1}{n}\bX^T\bE(P_B - \wh{P}_B)\be_1|
	&\leq \frac{1}{n}\norm{\bE^T\bX\wh{\bomega}_1}_2\norm{(\wh{P}_B - P_B)\be_1}_2.
	\eeq
	Using the same arguments of bounding $\norm{\wh{\bomega}_1^T\bX^T\bW}_2$ as above, one can establish that
	\begin{align*}
	    \PP\left\{
	        \norm{\wh{\bomega}_1^T\bX^T\bE}_2^2 > \gamma_e^2(\wh\bomega_1^T\bX^T\bX\wh\bomega_1)\left(
	        \sqrt{\tr(\se)} + \sqrt{2\|\se\|_{\op}t}
	        \right)^2
	    \right\} \le e^{-t},\quad \forall t>0.
	\end{align*}
	Hence, by $\|\se\|_{\op}=\cO(1)$, (\ref{bd_omegaXXomega}) and (\ref{bd_Omega_11}), 
	\[
	    \norm{\wh{\bomega}_1^T\frac{1}{n}\bX^T\bE}_2 = \cO_{\PP}\left(
	     \sqrt{m\over n}
	    \right).
	\]
	Finally, invoke Lemma \ref{lemma_PB_error} to obtain
	\beq\label{bd_I52}
	 |I_{52}| =\cO_{\PP}\left(\eta_n\over \sqrt n\right).
	\eeq
    \end{itemize}	
    
    Collecting (\ref{bd_I1}), (\ref{bd_I2}), (\ref{bd_I3}), (\ref{bd_I4}), (\ref{bd_I51}) and (\ref{bd_I52}) and using 
    $$
        \bar \Theta_{11} = \Theta_{11} - \bTheta_{1\cdot}^TP_B\be_1 \overset{(\ref{bd_bias})}{=} \Theta_{11} + \cO(M_n/m)
    $$ 
    conclude
    \begin{align*}
        \sqrt{n}\left(\wt\Theta_{11} - \Theta_{11}\right) &= \zeta + \Delta 
    \end{align*}
    where $\zeta$ satisfies (\ref{def_zeta}) and 
    \beq\nonumber
        \Delta &= \cO_{\PP}\left(
        \xi_n\sqrt{\log p} + \left({M_n\sqrt{n} \over m} + \sqrt{s_\Omega \log p} + \sqrt{n} \norm{\bA_{1\cdot}}_2+1\right)\eta_n
        \right)  + \cO\left({M_n\sqrt{n}\over m}\right) + o_{\PP}(1).
    \eeq
    By $M_n\sqrt{n} = o(m)$, (\ref{def_xi}) and (\ref{def_eta_bar}), after a bit algebra, we conclude 
    \begin{align*}
        \Delta 
        &= \cO_{\PP}\left(
        {s_n\log p \over \sqrt n} + \left({s_nM_n\sqrt{\log p} \over m} + \sqrt{(s_n\vee s_\Omega)\log p}+  \sqrt{n} \norm{\bA_{1\cdot}}_2+  1\right) \eta_n
        \right)+ o_\PP(1)\\
        &= \cO_{\PP}\left(\left( \sqrt{(s_n\vee s_\Omega)\log p}+  \sqrt{n} \norm{\bA_{1\cdot}}_2+1\right) \left(
        \sqrt{\log m\over n} + r_n
        \right)
        \right) + o_\PP(1)\\
        &= \cO_{\PP}\left(\sqrt{(s_n\vee s_\Omega)\log(p)\log(m)\over n}\right)\\
        &\quad + \cO_{\PP}\left(\norm{\bA_{1\cdot}}_2\sqrt{\log m}+\left( \sqrt{(s_n\vee s_\Omega)\log p}+ \sqrt{n} \norm{\bA_{1\cdot}}_2\right) r_n
        \right) + o_\PP(1)\\
        &= o_\PP(1)
    \end{align*}
    where we use $s_n\log p = o(\sqrt n)$ in the second line, use $\log m =o(n)$ and $r_n = o(1)$ in the third equality and use 
    $
        (s_n\vee s_\Omega)\log(p)\log(m) = o(n)
    $
    together with (\ref{cond_rn}) in the last step. 
    
    Finally,  $|\wh\bomega_1^T \wh\Sigma \wh\bomega_1 - \Omega_{11}| = o_\PP(1)$ is proved in Lemma \ref{lemma_nodewise}. The proof is complete.\qed

      \subsection{Proof of Corollary \ref{cor_ASN}}\label{app_proof_cor_ASN}
        We first prove case (1). From Theorem \ref{thm_pred}, we start by simplifying the expressions of $Rem_{1,j}$, $Rem_{2,j}(\bdelta_j)$ and $Rem_{3,j}(\btheta_j)$. Recall the SVD of $\wh\Sigma = \sum_{k=1}^q\Lambda_q \bu_k\bu_k^T$ with $q = \rank(\bX)$. Pick any $1\le j\le m$ and note $\|\btheta_j\|_0 \le s_n$ We have 
        \begin{align*}
            &Rem_{1,j} = {\sigma_j^2\over n}\left(
            \sum_{k=1}^q \left(
                \Lambda_k \over \Lambda_k + \lambda_2^{(j)}
            \right)^2 + \left(\Lambda_1 \over \Lambda_1 + \lambda_2^{(j)}\right)^2\log m
            \right),\\
            &Rem_{2,j}(\bdelta_j) = \sum_{k=1}^q{\lambda_2^{(j)} \Lambda_k \over \Lambda_k + \lambda_2^{(j)}} \left(\bu_k^T\bdelta_j\right)^2,\\
            & Rem_{3,j}(\btheta_j) =  
    		{\lbdj(\Lambda_1 + \lambda_2^{(j)}) \over (\Lambda_q+\lambda_2^{(j)})^2}\left(\max_{1\le i\le p}\wh\Sigma_{ii}\right) {s_n\log (p\vee m) \over \kappa^2(s_n,4)}{\sigma_j^2\over n}.
        \end{align*}
        Taking $\lambda_2 \to \infty$ yields
        \begin{align*}
            &Rem_{1,j} = 0,\\
            &Rem_{2,j}(\bdelta_j) = \sum_{k=1}^q\Lambda_k \left(\bu_k^T\bdelta_j\right)^2 = \bdelta_j^T \wh\Sigma \bdelta_j,\\
             & Rem_{3,j}(\btheta_j) = \left(\max_{1\le i\le p}\wh\Sigma_{ii}\right) {s_n\log (p\vee m) \over \kappa^2(s_n,4)}{\sigma_j^2\over n}.
        \end{align*}
        An application of Lemma \ref{lem_bernstein} together with 
        $$
            \bdelta_j^T \Sigma \bdelta_j \le \|\bdelta_j\|_2^2 \|\Sigma\|_{\op}\le \|\bA\|_{\op}^2\|\bB_j\|_2^2\|\Sigma\|_{\op} \lesssim \|\bA\|_{\op}^2\|\Sigma\|_{\op}
        $$ yields 
        \begin{align*}
            \PP\left\{
            \bdelta_j^T \wh\Sigma \bdelta_j \le \|\bA\|_{\op}^2\|\Sigma\|_{\op}\left(1 + \sqrt{\log m \over n}\right)
            \right\}\ge 1-2p^{-2}.
        \end{align*}
        Taking the union bounds over $1\le j\le m$ and 
        invoking Assumptions \ref{ass_B_Sigma} and $\cE_{\bX}$ in (\ref{def_event_X}) conclude
        \[
            r_n = \cO\left(\|\bA\|_{\op}^2 + {s_n\log(p\vee m)\over n}\right)
        \]
        with probability tending to one. This proves the rate in (\ref{rate_rnj_case1}). In this case, condition (\ref{cond_rn}) reduces to 
        \[
            \norm{\bA_{1\cdot}}_2 \sqrt{\log m}+\left(\|\bA_{1\cdot}\|_2 \sqrt n + \sqrt{(s_n\vee s_\Omega)\log p}\right)\left(
                \|\bA\|_{\op}^2 + {s_n\log(p\vee m)\over n}
            \right)  = o(1).
        \]
        Provided that $\|\bA_{1\cdot}\|_2 = o(\sqrt{(s_n\vee s_\Omega)\log p/n})$, 
        \[
            \norm{\bA_{1\cdot}}_2 \sqrt{\log m} = o\left(
             (s_n\vee s_\Omega)\log p\log m\over n
            \right) = o(1).
        \]
        and 
        \[
            \sqrt{(s_n\vee s_\Omega)\log p}\left(
                \|\bA\|_{\op}^2 + {s_n\log(p\vee m)\over n}
            \right) = o(1)
        \]
        is ensured by (\ref{cond_A_op}) and $(s_n\vee s_\Omega)\log^2(p\vee m) = o(n)$.
        
        To prove case (2), by repeating the proof of Corollary 8 in \cite{bing2020adaptive}, one can deduce that 
        \[
            Rem_{1,j} + Rem_{2,j}(\bdelta_j) + Rem_{3,j}(\btheta_j) \lesssim \sqrt{(\tr(\wh\Sigma) + \Lambda_1 s_n) \|\bdelta_j\|_2^2\log (p\vee m)\over n}+{s_n \over n}.
        \]
        Since $\tr(\wh\Sigma)=\cO_\PP(p)$, $\|\bdelta_j\|_2^2 \lesssim \|\bA\|_{\op}^2 = \cO(1/p)$ and $\Lambda_1 = \cO_\PP(p)$ by using Lemma \ref{lem_op_norm}, $\max_{1\le j\le p}\Sigma_{jj} = \cO(1)$ and  $\|\Sigma\|_{\op} =\cO(p)$,  we conclude
        \[
            r_n = \cO\left(
             \sqrt{s_n \log (p\vee m)\over n}+{s_n\log (p\vee m) \over n}
            \right).
        \]
        Immediately, $\|\bA_{1\cdot}\|_2 \le \|\bA\|_{\op}$ and condition (\ref{cond_rn}) holds under $\|\bA\|_{\op}^2 = \cO(1/p)$ and $s_n(s_n\vee s_\Omega)\log^2(p\vee m) = o(n)$.
        \qed

    \subsection{Proof of Proposition \ref{prop_sigma_E}: consistency of the estimation of $\sigma_{E_1}^2$}\label{app_proof_prop_sigma_E}
        We work on the event that 
        \[
            \left\{\lambda_K(\bH_0) \gtrsim c_H\right\} \bigcap \left\{{1\over n}\|\bX\wh\bF_1 - \bX\bF_1\|_2^2 \lesssim r_{n,1}\right\}
        \]
        which, according to Lemma \ref{lemma_technical} and Theorem \ref{thm_pred}, holds with probability tending to one.
        Recall from (\ref{def_est_epsilon}) that 
        \[
            \wh\bepsilon_1 = \bepsilon_1 + \bDelta_1 = \bW\bB_1 + \bE_1 + \bDelta_1 = \wt\bW\wt\bB_1 + \bE_1 + \bDelta_1
        \]
        with $\bDelta_1 = \bX\wh\bF_1 - \bX\bF_1$, $\wt\bW = \bW\bH_0^{-1}$ and $\wt\bB = \bH_0\bB$ defined in (\ref{def_B_tilde}).
        By definition (\ref{def_est_variance}), after a bit algebra,
        \begin{align*}
            \wh\sigma_{E_1}^2 - \sigma_{E_1}^2 &=  {1\over n}\bE_1^T\bE_1 - \sigma_{E_1}^2+ {1\over n}\bDelta_1^T\bDelta_1 + {2\over n}\bDelta_1^T(\wt\bW\wt\bB_1 - \wh\bW\wh\bB_1) + {2\over n}\bDelta_1^T\bE_1\\
            &\quad +  {1\over n}(\wt\bW\wt\bB_1 - \wh\bW\wh\bB_1)^T(\wt\bW\wt\bB_1 - \wh\bW\wh\bB_1) + {2\over n}(\wt\bW\wt\bB_1 - \wh\bW\wh\bB_1)^T\bE_1.
        \end{align*}
        We study each terms on the right hand side separately. First, an application of Lemma \ref{lem_bernstein} together with $\sigma_{E_1}^2\le C_E$ gives 
        \[
            \left| {1\over n}\bE_1^T\bE_1 - \sigma_{E_1}^2\right| = \cO_\PP\left(\sqrt{1 /n}
            \right),
        \]
        which further implies
        \[
            {1\over \sqrt n}\|\bE_1\|_2 = \cO_\PP(1).
        \]
        We thus have
        \beq\label{bd_term_1}
            &\left|
                {1\over n}\bE_1^T\bE_1 - \sigma_{E_1}^2+ {1\over n}\bDelta_1^T\bDelta_1 + {2\over n}\bDelta_1^T\bE_1
            \right|\\
            &\le \left|
                {1\over n}\bE_1^T\bE_1 - \sigma_{E_1}^2\right|+ {1\over n}\|\bDelta_1\|_2^2 + {2\over n}\|\bDelta_1\|_2\|\bE_1\|_2
          = \cO_{\PP}(n^{-1/2} + r_n).
        \eeq 
        To bound the other terms, notice that 
        \[
            \|\wt\bW\wt\bB_1 - \wh\bW\wh\bB_1\|_2 \le \|\wt\bW-\wh\bW\|_\op \|\wh\bB_1\|_2 + \|\wt\bW\|_\op\|\wh\bB_1 - \wt\bB_1\|_2.
        \]
        By Lemma \ref{lem_W_eigens}, part (B) of Lemma \ref{lemma_technical}, Theorem \ref{thm_rates_B} and Lemma \ref{lem_W_frob}, we have 
        \beq\nonumber
                {1\over n}\|\wt\bW\wt\bB_1 - \wh\bW\wh\bB_1\|_2 = \cO_\PP\left(\sqrt{\log m\over n} +  r_n\right).
        \eeq 
        This leads to 
        \beq \label{bd_term_2}
            &\left|
                 {2\over n}\bDelta_1^T(\wt\bW\wt\bB_1 - \wh\bW\wh\bB_1)+{1\over n}(\wt\bW\wt\bB_1 - \wh\bW\wh\bB_1)^T(\wt\bW\wt\bB_1 - \wh\bW\wh\bB_1)\right.\\
            &\quad \left.+ {2\over n}(\wt\bW\wt\bB_1 - \wh\bW\wh\bB_1)^T\bE_1
            \right| = \cO_\PP\left(
                \sqrt{\log m\over n} + r_n
            \right).
        \eeq 
        Collecting (\ref{bd_term_1}) and (\ref{bd_term_2}) completes the proof. \qed

        \bigskip
        
        The following lemma provides overall control of $\wh\bW - \wt\bW$ in the operator norm. 
        
        \begin{lemma}\label{lem_W_frob}
        Under conditions of Theorem \ref{thm_rates_B}, with probability tending to one,
        \[
            {1\over \sqrt n}\|\wt\bW-\wh\bW\|_\op \lesssim  \sqrt{r_n} + \sqrt{\log m \over n \wedge m}.
        \]
    \end{lemma}    
    \begin{proof}
        We work on the event that parts (A) -- (C) of Lemma \ref{lemma_technical} hold intersecting with $\cE_B$ in (\ref{def_event_B}) and $\cE_F$ in (\ref{def_event_F}). Recalling that $\wt\bB$ is defined in (\ref{def_B_tilde}) and $\wt\bW = \bW\bH_0^{-1}$.
        Observe that 
        $$
            \wh\bW = \wh\bepsilon\wh\bB^T(\wh\bB\wh\bB^T)^{-1} = \wt\bW\wt\bB\wh\bB^T(\wh\bB\wh\bB^T)^{-1}+ (\wh\bepsilon-\bepsilon)\wh\bB^T(\wh\bB\wh\bB^T)^{-1}
        $$
        with $\bepsilon = \bW\bB = \wt\bW\wt\bB$. This gives 
        \beq \nonumber
            \wh\bW - \wt\bW &=  \wt\bW(\wt\bB-\wh\bB)\wh\bB^T(\wh\bB\wh\bB^T)^{-1}+ (\wh\bepsilon-\bepsilon)\wh\bB^T(\wh\bB\wh\bB^T)^{-1}.
        \eeq
        For the first term, 
        \[
            {1\over \sqrt n}\|\wt\bW(\wt\bB-\wh\bB)\wh\bB^T(\wh\bB\wh\bB^T)^{-1}\|_\op \le c_H^{-1}{1\over \sqrt{n}}\|\bW\|_\op{\|\wt\bB-\wh\bB\|_{\op}\over \lambda_K(\wh\bB)}.
        \]
        Invoking Lemma \ref{lem_W_eigens} and (\ref{bd_B_diff_op}) yields 
        \[
             {1\over \sqrt n}\|\wt\bW(\wt\bB-\wh\bB)\wh\bB^T(\wh\bB\wh\bB^T)^{-1}\|_\op = \cO_{\PP}\left( \eta_n \right)
        \]
        with $\eta_n$ defined in (\ref{def_eta_bar}).  Similarly, the second term can be bounded by 
        \[
         {1\over \sqrt n}\|(\wh\bepsilon-\bepsilon)\wh\bB^T(\wh\bB\wh\bB^T)^{-1}\|_\op \lesssim {1\over \sqrt{n}}\|\bX\wh\bF-\bX\bF\|_{F} {1\over \lambda_K(\wh\bB)} =\cO_{\PP}(\sqrt{r_n}).  
         \]
         Combining these two bounds completes the proof. 
    \end{proof}

       \subsection{Proof of Theorem \ref{thm_B_asn}: The asymptotic normality of $\wh B_j$}\label{app_proof_thm_B_asn}
       
       We work on the event $\cE_F\cap \cE_D$ in (\ref{def_event_F}) -- (\ref{def_event_D})  intersecting with $\{\lambda_K(\bH_0) \gtrsim 1\}$ which holds with probability tending to one. From (\ref{display_B_hat_BH}), for any $j\in [m]$, one has 
       \begin{align}\label{decomp_Bhat_j_B_j}\nonumber
            \sqrt{n}\left(\wh \bB_j - \bH_0\bB_j\right) & = {1\over m\sqrt n}\bD_K^{-2}\wh \bB \bB^T\bW^T\bE_j\\ 
            &\quad + \underbrace{{1\over m\sqrt n}\bD_K^{-2}\wh \bB\left(
            \bE^T\bW\bB_j +\bE^T\bE_j + \bepsilon^T\bDelta_j + \bDelta^T\bepsilon_j + \bDelta^T\bDelta_j
            \right)}_{R}.
       \end{align}
       Let  
       \beq\label{def_H2}
        \bH_2 = \bB\wh\bB^T (\wh\bB\wh\bB^T)^{-1} = {1\over m}\bB\wh\bB^T \bD_K^{-2},
       \eeq
       such that 
       \[   
            {1\over m\sqrt n}\bD_K^{-2}\wh \bB \bB^T\bW^T\bE_j = {1\over \sqrt n}\bH_2^T \bW^T\bE_j. 
       \]
       First notice that, since $\bW$ and $\bE$ are independent, the classical central limit theorem yields 
       \[
            {1\over \sqrt n}\bW^T\bE_j \overset{d}{\longrightarrow} N_K\left(\b0, \sigma_{E_j}^2\sw\right),\qquad \textrm{as }n\to \infty.
       \]
       Following \cite{bai2020simpler}, define 
       \beq\label{def_Q}
        \bQ = \Lambda_0R_0\Sigma_B^{-1/2}
       \eeq
       where $\Sigma_B = m^{-1}\bB\bB^T$ and $\Sigma_B^{1/2}\sw\Sigma_B^{1/2}$ has the eigen-decomposition $R_0\Lambda_0 R_0^T$.
       Since Lemma \ref{lem_H2} proves $\bH_2 \to \bQ^{-1}$ in probability, together with the fact $(\bQ^T)^{-1}\sw\bQ^{-1}= \bI_K$, Slutsky's theorem ensures 
       \[
             {1\over \sqrt n}\bH_2^T\bW^T\bE_j \overset{d}{\longrightarrow} N_K\left(\b0, \sigma_{E_j}^2\bI_K\right),\qquad \textrm{as }n\to \infty.
       \]
       
       It remains to show $R$ in (\ref{decomp_Bhat_j_B_j}) is of order $o_\PP(1)$. 
       By (\ref{bd_quad_Deltas}), one has 
       \begin{align}\label{bd_R1}\nonumber
          &{1\over m\sqrt n}\|\bD_K^{-2}\wh \bB\left(
            \bepsilon^T\bDelta_j + \bDelta^T\bepsilon_j + \bDelta^T\bDelta_j
            \right)\|_2\\\nonumber
            &= {1\over \sqrt{nm}}\|\bD_K^{-1}\bV_K^T\left(
            \bepsilon^T\bDelta_j + \bDelta^T\bepsilon_j + \bDelta^T\bDelta_j
            \right)\|_2\\\nonumber
            & \lesssim \sqrt{nr_n}\sqrt{Rem_{1,j} + Rem_{2,j}(\bdelta_j) + Rem_{3,j}(\btheta_j)} + r_{n,1}\sqrt{n} +  \sqrt{r_{n,2}\log (m)}+ r_{n,3}\\
            & = \sqrt{nr_n}\sqrt{Rem_{1,j} + Rem_{2,j}(\bdelta_j) + Rem_{3,j}(\btheta_j)} + r_{n,1}\sqrt{n} + o(1)
       \end{align}
       with probability $1-8m^{-1}$, provided that $r_n\sqrt{\log m} = o(1)$. In addition, recalling that $\wt\bB = \bH_0\bB$ and $\cE_D$, one has
       \begin{align*}
           {1\over m\sqrt n}\|\bD_K^{-2}\wh \bB
            \bE^T\bW\bB_j\|_2 &\lesssim {1\over m\sqrt n}\left(\|\wt\bB \bE^T\bW\bB_j\|_2 + \|\wh\bB-\wt\bB\|_{\op}\|\bE^T\bW\bB_j\|_2\right)\\
            &\lesssim {1\over m\sqrt n}\left(\|\bB \bE^T\bW\bB_j\|_2 + \|\wh\bB-\wt\bB\|_{\op}\|\bE^T\bW\bB_j\|_2\right).
       \end{align*}
       Since an application of Lemma \ref{lem_bernstein} with an union bound over $1\le k\le K$ yields
       \[
            {1\over m\sqrt n}\|\bB \bE^T\bW\bB_j\|_2 \le {1\over m\sqrt n}\left(
             n\log (m)\bB_j^T\sw \bB_j \sum_{k=1}^K\bB_{k\cdot}^T \se \bB_{k\cdot} 
            \right)^{1/2} \lesssim \sqrt{\log m \over m}
       \]
       with probability $1-2m^{-1}$, and similar arguments yield 
       \[
        {1\over \sqrt{nm}}\|\bE^T\bW\bB_j\|_2 \lesssim \max_{\ell\in[m]}{1\over \sqrt n}|\bE_{\ell}^T\bW \bB_j|\lesssim \sqrt{\log m}
       \]
       with probability $1-2m^{-1}$,  invoke (\ref{bd_B_diff_op}) to conclude 
       \beq\label{bd_R2}
             {1\over m\sqrt n}\|\bD_K^{-2}\wh \bB
            \bE^T\bW\bB_j\|_2  = o_\PP(1)
       \eeq
       provided that $r_n \sqrt{\log m} = o(1)$, $\log m = o(\sqrt m)$ and $\log^2(m) = o(\sqrt n)$. 
       Finally, by Lemma \ref{lem_quad_terms}, we have
       \beq\label{bd_R3}
            {1\over m\sqrt n}\|\bD_K^{-2}\wh \bB\bE^T\bE_j\|_2 &\lesssim {1\over m\sqrt n}\left(
            \|\bB\bE^T\bE_j\|_2 + \|\wh\bB-\wt\bB\|_{\op}\|\bE^T\bE_j\|_2
            \right)\\
            &\lesssim \sqrt{(n+m)\log m\over m^2} + \left(\sqrt{\log m\over n\wedge m} + r_n\right)\sqrt{(n+m)\log m\over m}\\
            & = o(1) + r_n \sqrt{n\log m\over m}
       \eeq
       with probability tending to one. The last step uses $$\sqrt{n\log m} = o(m)$$ and $r_n\sqrt{\log m} = o(1)$.  To combine the bounds, by taking $\lambda_2^{(j)} \to \infty$ for all $1\le j\le m$ and invoking $\cE_X$ in (\ref{def_event_X}), one has  
           \begin{align*}
               n Rem_{1,j} \le n r_1 = o_\PP(1),\quad Rem_{2,j}(\bdelta_j) =\cO_\PP\left(
                \|\bdelta_j\|_2^2
                \right),\qquad r_{n,2} = \cO_\PP(\|\bA\|_{\op}^2)
           \end{align*}
           and
           \[
             Rem_{3,j}(\btheta_j) \le r_{n,3} = \cO_{\PP}\left(
                    {s_n\log(p\vee m) \over n}
                \right),
           \]
           such that 
           \[
                r_n = \cO_\PP\left(
               \|\bA\|_{\op}^2 +  {s_n\log(p\vee m) \over n}
                \right) + o_\PP(n^{-1}).
           \]
           Therefore, $r_n\sqrt{\log m} = o(1)$. Also by $s_n\log(p\vee m) = o(\sqrt n)$, collecting (\ref{bd_R1}), (\ref{bd_R2}) and (\ref{bd_R3}) yields 
           \begin{align*}
                \left\|R
                \right\|_2 &= \cO_\PP\left(
                 \|\bdelta_j\|_2\sqrt{n r_n} +\sqrt{r_ns_n\log(p\vee m)} + r_n\sqrt{n\log m\over m}
                \right) + o_\PP(1)\\
                &= \cO_\PP\left(
                 \|\bdelta_j\|_2\sqrt{n\|\bA\|_{\op}^2 + s_n\log(p\vee m)} +\|\bA\|_{\op} \sqrt{s_n\log(p\vee m)}\right.\\
                 &\quad\qquad \left.+ \|\bA\|_{\op}^2\sqrt{n\log m\over m}
                \right) + o_\PP(1)\\
                &=\cO_\PP\left(
                \|\bA\|_{\op}\left[
                 \|\bdelta_j\|_2\sqrt{n} + \sqrt{s_n\log(p\vee m)}\right]+ \|\bA\|_{\op}^2\sqrt{n\log m\over m}
                \right) + o_\PP(1)
           \end{align*}
           Invoke condition (\ref{cond_r_asn}) to complete the proof.

    \section{Technical lemmas}

        \subsection{Lemmas used in the proof of Theorem \ref{thm_rates_B}}

         The following lemma provides upper and lower bounds of the eigenvalues of $n^{-1}\bW^T\bW$. 
    \begin{lemma}\label{lem_W_eigens}
        Under Assumptions \ref{ass_error} and \ref{ass_B_Sigma}, assume $K\log n\le Cn$ for some large constant $C>0$. Then 
        \[
            \PP\left\{
        c_W \lesssim \lambda_K\left({1\over n}\bW^T\bW\right) \le \lambda_1\left({1\over n}\bW^T\bW\right) \lesssim C_W
        \right\} \ge 1-2e^{-n}.
        \]
    \end{lemma}
    \begin{proof}
        First, an application of Lemma \ref{lem_op_norm_diff} yields 
        \[
            \PP\left\{
             \left\|{1\over n}\bW^T\bW - \sw\right\|_{\op} \lesssim \|\sw\|_{\op} \left(\sqrt{K\log n\over n} + {K\log n\over n}\right)
            \right\}  \ge 1-2e^{-n}.
        \]
        As Weyl's inequality leads to 
        \[
                \left|\lambda_k\left({1\over n}\bW^T\bW\right) - \lambda_k(\sw)\right| 
                \le  \left\|{1\over n}\bW^T\bW - \sw\right\|_{\op},\quad \forall 1\le k\le K,
        \]
        use $c_W\le \lambda_K(\sw) \le \lambda_1(\sw) \le C_W$ and $K\log n \le Cn$ to complete the proof.
    \end{proof}

        The following lemma shows that the event $\cE_D$ in (\ref{def_event_D}) holds with probability tending to one, thereby providing upper and lower bounds for the singular values of $\wh \bepsilon/\sqrt{nm}$.
        
        \begin{lemma}\label{lem_D_K}
            Under conditions of Theorem \ref{thm_rates_B}, one has
            $$
            \lim_{n\to\infty}\PP(\cE_D) = 1.
            $$
        \end{lemma}
        \begin{proof}
            Recall that $\bD_K$ contains the $K$ largest singular value of $\wh\bepsilon / \sqrt{nm}$. From 
            \[
                \wh\bepsilon = \bW\bB + \bE + \bDelta
            \]
            with $\bDelta = \bX\bF- \bX\wh\bF$, using Weyl's inequality gives 
            \begin{align*}
                    \left|\lambda_k(\bD_K) - {1\over \sqrt{nm}}\lambda_k(\bW\bB)\right| & ~ =~  
                    \left|{1\over \sqrt{nm}}\lambda_k(\wh\bepsilon) - {1\over \sqrt{nm}}\lambda_k(\bW\bB)\right|\\
                    &~ \le~ {1\over \sqrt{nm}}\|\bE\|_{\op} + {1\over \sqrt{nm}}\|\bX\wh\bF - \bX\bF\|_{\op},
            \end{align*}
            for all $1\le k\le K$. 
            On the one hand, by Assumption \ref{ass_B_Sigma} and Lemma \ref{lem_W_eigens},
            \[
                \sqrt{c_Wc_B} \lesssim {1\over \sqrt{nm}}\lambda_K(\bW\bB)\le   {1\over \sqrt{nm}}\lambda_1(\bW\bB) \lesssim \sqrt{C_WC_B}
            \]
            with probability at least $1-2n^{-c'n}$. 
            On the other hand, 
            invoke Lemma \ref{lem_op_norm} to obtain 
            \[
                \PP\left\{{1\over nm}\|\bE^T\bE\|_{\op} \le {\g_e^2\over m}\left(
                 \sqrt{\tr(\se) \over n} + \sqrt{6\|\se\|_{\op}} 
                \right)^2\right\} \ge 1-e^{-n}.
            \]
            Using $\tr(\se) \le m\|\se\|_{\op} \le C_E m$ and  $\|\se\|_{\op}\le C_E$ implies
            \[
                    {1\over nm}\|\bE^T\bE\|_{\op} = o_{\PP}(1).
            \]
            Since Assumption \ref{ass_initial} ensures 
            \[
               {1\over nm}\|\bX\wh\bF - \bX\bF\|_{\op}^2= \cO_{\PP}\left(r_n\right) = o_{\PP}(1),
            \]
            we conclude that, with probability tending to one, 
            \[
                \sqrt{c_Wc_B} \lesssim \lambda_k(\bD_K) \lesssim \sqrt{C_Wc_B},\quad \forall 1\le k\le K.
            \]
             The proof is complete.
        \end{proof}

        \medskip

        \begin{lemma}\label{lem_quad_terms}
        Under Assumptions \ref{ass_error} and \ref{ass_B_Sigma},  with probability greater than $1-8m^{-1}$, the following holds, uniformly over $1\le j\le m$,
        \begin{align*}
            &\|\bE^T\bW \bB_j\|_2 \lesssim \sqrt{nm\log m},  \\
            &\|\bE_j^T\bW \bB\|_2  \lesssim \sqrt{nm\log m},  \\
            &\|\bE_j^T\bE\|_2  \lesssim  \sqrt{n(n+m)\log m}.
        \end{align*}
        Furthermore, if $\|\se\|_{\i,1} \le C$ for some constant $C>0$, then with probability $1-2m^{-1}$, uniformly over $1\le j\le m$,
        \[
            \|\bB\bE^T\bE_j\|_2 \lesssim \sqrt{n(n+m)\log m}.
        \]
        \end{lemma}
        \begin{proof}
            Write $\bar \bE = \bE \se^{-1/2}$ and $\bar\bW = \bW\sw^{-1/2}$. We have 
            \[
                \|\bE^T\bW \bB_j\|_2^2  \le \|\se\|_{\op} \sum_{\ell=1}^m\left(
                \bar\bE_\ell^T \bW \bB_j
                \right)^2.
            \]
            Notice that $\bar E_{i\ell}$ is $\g_e$ sub-Gaussian and $\bW_{i\cdot}^T\bB_j$ is $\g_w\sqrt{\bB_j^T\sw\bB_j}$ sub-Gaussian,
            for all $1\le i\le n$.
            An application of Lemma \ref{lem_bernstein} together with union bounds over $1\le \ell \le m$ gives 
            \[
                \PP\left\{
                    \|\bE^T\bW \bB_j\|_2 \lesssim \sqrt{\|\se\|_{\op}\bB_j^T\sw\bB_j}\sqrt{nm\log m}
                \right\} \ge 1-2m^{-1}.
            \]
            By similar arguments, 
            $$
            \|\bE_j^T\bW \bB\|_2^2 \le \|\bE_j^T\bar\bW\|_2^2\|\bB^T\sw\bB\|_{\op} \le K\|\bE_j^T\bar\bW\|_\i^2\|\bB^T\sw\bB\|_{\op}.
            $$ 
            Since $E_{ij}$ is $\g_e\sqrt{[\se]_{jj}}$ sub-Gaussian for $1\le i\le n$, apply Lemma \ref{lem_bernstein} to bound $|\bE_j^T\bar\bW_k|$ and take union bounds over $1\le k\le K$ to obtain
            \[
                \PP\left\{
                    \|\bE_j^T\bW \bB\|_2 \lesssim \sqrt{\|\bB^T\sw\bB\|_{\op}[\se]_{jj}}\sqrt{nK\log m}
                \right\} \ge 1-2m^{-1}.
            \]
            The result follows by $\|\bB^T\sw\bB\|_{\op} \lesssim m$ from Assumption \ref{ass_B_Sigma}.
            Finally, 
            \beq\label{bd_EjE}
                \|\bE_j^T\bE\|_2^2 \le \|\se\|_{\op} \left(
                (\bE_j^T \bar\bE_j)^2 + \sum_{\ell \ne j}(\bE_j^T\bar\bE_\ell)^2
                \right).
            \eeq
            For the first term, for any $1\le i\le n$, notice that
            \[
                \EE\left[
                \bE_{ij}\bar\bE_{ij}
                \right] = \EE\left[
                \bE_{ij}\bE_{i\cdot}^T
                \right]\se^{-1/2}\be_j = \be_j^T\se^{1/2}\be_j.
            \]
            An application of Lemma \ref{lem_bernstein} gives
            \[
                 \PP\left\{
                    |\bE_j^T \bar\bE_j - n\be_j^T\se^{1/2}\be_j| \lesssim \sqrt{[\se]_{jj}}\sqrt{n\log m}
                \right\} \ge 1-2m^{-1},
            \]
            which implies 
            \beq\label{bd_EjEj}
            |\bE_j^T \bar\bE_j| \lesssim n\be_j^T\se^{1/2}\be_j + \sqrt{[\se]_{jj}}\sqrt{n\log m}\lesssim n\sqrt{\log m}
            \eeq
            with the same probability. Similarly, applying Lemma \ref{lem_bernstein} again to $\bE_j^T\bar\bE_\ell$ with union bounds over $j\ne \ell \in [m]$ yields 
             \[
                 \PP\left\{
                    |\bE_j^T \bar\bE_\ell| \lesssim \sqrt{[\se]_{jj}}\sqrt{n\log m}
                \right\} \ge 1-2m^{-1}.
            \]
            Combining this with (\ref{bd_EjE}) and (\ref{bd_EjEj}) concludes 
            \[
                \|\bE_j^T\bE\|_2^2 \lesssim n^2\log m + nm\log m
            \]
            with probability at least $1-4m^{-1}$.
            
            Finally, by similar arguments, one can show that, with probability $1-2m^{-1}$
            \[
                |\bB_{k\cdot}^T \bE^T\bE_j| \lesssim n\bB_{k\cdot}^T \se \be_j + \sqrt{n\log(m)[\se]_{jj} \bB_{k\cdot}^T \se \bB_{k\cdot}}
            \]
            uniformly over $1\le k\le K$ and $1\le j\le m$, and therefore, with the same probability,
            \begin{align*}
            \|\bB \bE^T\bE_j\|_2^2 &\lesssim \sum_{k=1}^K
            \left[
            n^2(\bB_{k\cdot}^T \se \be_j)^2 + n\log(m)[\se]_{jj} \bB_{k\cdot}^T \se \bB_{k\cdot}
            \right]\\
            &= n^2 \be_j^T \se \bB^T\bB\se \be_j + n\log(m) [\se]_{jj}\tr(\bB\se\bB)\\
            &\le  n^2\|\se\|_{\i,1}^2\|\bB\|_{2,\i}^2 + n\log(m)[\se]_{jj} \|\bB\|_F^2 \|\se\|_{\op}\\
            &\lesssim n^2 + nm\log(m)
            \end{align*}
            by invoking Assumption \ref{ass_B_Sigma} and using $\|\se\|_{\i,1}\le C$ in the last step.
            This completes the proof.
        \end{proof}
        
        \bigskip
        
        Recalling from (\ref{def_r_n_k}), 
        Assumption \ref{ass_initial} implies $r_{n,k} \le r_n = o_{\PP}(1)$, for $k\in \{1,2,3\}$.
        
        \begin{lemma}\label{lem_quad_terms_Delta}
        Under conditions of Theorem \ref{thm_rates_B},  on the event $\cE_F$ defined in (\ref{def_event_F}), the following holds with probability greater than $1-8m^{-1}$, uniformly over $1\le j\le m$.
        \begin{align*}
            {1\over n\sqrt m}\|\bepsilon_j^T\bDelta\|_2 ~~&\lesssim ~ r_{n,1} + \sqrt{r_{n,2}\log(m) \over n} + r_{n,3}\sqrt{1\over n},\\
            {1\over n\sqrt{m}}\|\bDelta_j^T\bepsilon\|_2 ~~&\lesssim Rem_{1,j} + \sqrt{\log (m) Rem_{2,j}(\bdelta_j)\over n}+ {Rem_{3,j}(\btheta_j)\over \sqrt n}\\
            & \lesssim  ~ r_{n,1} + \sqrt{r_{n,2}\log (m) \over n}+ r_{n,3}\sqrt{1\over n},\\
            {1\over n\sqrt m}\|\bDelta_j^T \bDelta\|_2 ~&\lesssim   \sqrt{r_n}\sqrt{Rem_{1,j} + Rem_{2,j}(\bdelta_j) + Rem_{3,j}(\btheta_j)},
        \end{align*}
        with $r_n$ defined in Assumption \ref{ass_initial}.
        \end{lemma}
        \begin{proof}
            Since $\bDelta = \wh\bepsilon - \bepsilon = \bX\bF - \bX\wh\bF$, on the event $\cE_F$, we immediately have
            \begin{equation}\label{bd_Delta_Delta}
                \|\bDelta_j^T \bDelta\|_2^2 \le \sum_{\ell=1}^m \|\bDelta_j\|_2^2 \|\bDelta_\ell\|_2^2 \lesssim  n m ~ r_n \|\bDelta_j\|_2^2.
            \end{equation}
            To study the other two terms, first note that $\btheta_j$ and $\bdelta_j$ are identifiable under conditions of Theorem \ref{thm_rates_B}. From Lemma \ref{lem_solution} and  $\btheta_j + \bdelta_j = \bF_j$, for any $1\le j\le m$, we have 
            \begin{align*}
                \bDelta_j = \bX\wh\bF_j - \bX\bF_j = P_{\lambda_2^{(j)}}\bepsilon_j - Q_{\lambda_2^{(j)}}\bX\bdelta_j + Q_{\lambda_2^{(j)}}\bX(\wh\btheta^{(j)} - \btheta_j).
            \end{align*}
            Then 
            \begin{align*}
                \|\bepsilon^T\bDelta_j\|_2 \le \left\|
                \bepsilon^T P_{\lambda_2^{(j)}}\bepsilon_j
                \right\|_2 + \left\|
                \bepsilon^T Q_{\lambda_2^{(j)}}\bX\bdelta_j
                \right\|_2 + \left\|
                \bepsilon^T Q_{\lambda_2^{(j)}}\bX(\wh\btheta^{(j)} - \btheta_j)
                \right\|_2.
            \end{align*}
            By Cauchy-Schwarz inequality, we have 
            \begin{align*}
                \left\|
                \bepsilon^T P_{\lambda_2^{(j)}}\bepsilon_j
                \right\|_2^2 \le \sum_{\ell=1}^m \left(\bepsilon_\ell P_{\lambda_2^{(j)}}\bepsilon_\ell\right) \left(\bepsilon_j P_{\lambda_2^{(j)}}\bepsilon_j\right).
            \end{align*}
            Invoking Lemma \ref{lem_quad} gives, with probability at least $1-m^{-1}$,
            \begin{align*}
                \bepsilon_\ell P_{\lambda_2^{(j)}}\bepsilon_\ell &\lesssim \sigma_\ell^2\left(
                \sqrt{\tr\left(P_{\lambda_2^{(j)}}\right)} + \sqrt{\left\|P_{\lambda_2^{(j)}}\right\|_{\op}\log m}\right)^2\\
                &\le 2\sigma_{\ell}^2 \left(
                 \tr\left(P_{\lambda_2^{(j)}}\right) +  \left\|P_{\lambda_2^{(j)}}\right\|_{\op}\log m\right)\\
                 & \asymp n Rem_{1,j}, 
            \end{align*}
            uniformly over $1\le \ell \le m$ and $1\le j\le m$. Here $\sigma_j^2$ is defined in (\ref{def_V_eps}) and in the last step we used  
            \begin{equation}\label{eq_sigmas}
                \sigma_j^2 \asymp 1,\qquad \forall 1\le j\le m
            \end{equation}
            under Assumption \ref{ass_B_Sigma}. The above display implies, with the same probability, 
            \beq\label{rate_eps_P_eps_j}
                \left\|
                \bepsilon^T P_{\lambda_2^{(j)}}\bepsilon_j
                \right\|_2^2  \lesssim n^2 m  [Rem_{1,j}]^2.
            \eeq
            By similar lines of arguments in the proof of Lemma 14 in \cite{bing2020adaptive}, one can show that, with probability $1-m^{-1}$, 
            \beq\label{rate_eps_QX_delta}
                 \left\|
                \bepsilon^T Q_{\lambda_2^{(j)}}\bX\bdelta_j
                \right\|_2^2 \lesssim n Rem_{2,j}(\bdelta_j)  \log(m) \sum_{\ell=1}^m \sigma_\ell^2
            \eeq 
            holds uniformly over $1\le j\le m$. Finally, 
            \[
                \left\|
                \bepsilon^T Q_{\lambda_2^{(j)}}\bX(\wh\btheta^{(j)} - \btheta_j)
                \right\|_2 \le \max_{1\le i\le p}\left\|
                \bepsilon^T Q_{\lambda_2^{(j)}}\bX_i\right\|_2 \left\|\wh\btheta^{(j)} - \btheta_j\right\|_1.
            \]
            By arguments of Lemma 15 in \cite{bing2020adaptive}, with probability at least $1-(pm)^{-1}$
            \[
                \max_{1\le i\le n}\left\|
                \bepsilon^T Q_{\lambda_2^{(j)}}\bX_i\right\|_2^2 \lesssim \left(
                \sqrt{\tr(\Gamma)} + \sqrt{4\log(pm) \|\Gamma\|_{\op}} 
                \right)^2 n \max_{1\le i\le p} M^{(j)}_{ii}
            \]
            uniformly over $1\le j\le m$, with $\Gamma := \g_w^2\bB^T\sw\bB + \g_e^2\se$ and $M^{(j)} = n^{-1}\bX^TQ_{\lambda_2^{(j)}}^2\bX$. Furthermore, the proof of Lemma 9 in \cite{bing2020adaptive} ensures that, with probability $1- (p\wedge m)^{-1}$,
            \[
                \|\wh\btheta^{(j)} - \btheta_j\|_1 \lesssim {Rem_{3,j}(\btheta_j) + Rem_{2,j}(\bdelta_j) \over \lambda_{1,j}}
            \]
            uniformly over $1\le j\le m$. By (\ref{rate_lbd1}), we conclude 
            \begin{align}\label{rate_eps_QX_beta_diff}\nonumber
               \left\|
                \bepsilon^T Q_{\lambda_2^{(j)}}\bX(\wh\btheta^{(j)} - \btheta_j)
                \right\|_2  &\lesssim \sqrt{n}\left[Rem_{3,j}(\btheta_j) + Rem_{2,j}(\bdelta_j)\right]{\sqrt{\tr(\Gamma)} + \sqrt{\|\Gamma\|_{\op}\log(pm)} \over \sigma_j\sqrt{\log (p\vee m)}}\\
                &\lesssim \sqrt{nm}\left[Rem_{3,j}(\btheta_j) + Rem_{2,j}(\bdelta_j)\right]
            \end{align}
            where the last line follows from $\tr(\Gamma) \lesssim  m$ and $\sigma_j^2 \asymp 1$ under Assumption \ref{ass_B_Sigma}. Collecting (\ref{rate_eps_P_eps_j}), (\ref{rate_eps_QX_delta}) and (\ref{rate_eps_QX_beta_diff}) concludes
            \begin{align}\label{bd_Delta_eps}
               {1\over \sqrt{nm}} \|\bDelta_j^T\bepsilon\|_2
                &\lesssim 
                    \sqrt n Rem_{1,j} + \sqrt{\log (m) Rem_{2,j}(\bdelta_j)}+ Rem_{3,j}(\btheta_j) + Rem_{2,j}(\bdelta_j).
            \end{align}
            
            We proceed to use the same arguments to bound from above 
            \[
             \|\bDelta^T\bepsilon_j\|_2^2 = \sum_{\ell=1}^m |\bDelta_\ell^T\bepsilon_j|^2.
            \]
            Since
            $$
            |\bDelta_\ell^T\bepsilon_j| \le \left|
                \bepsilon_\ell^T P_{\lambda_2^{(\ell)}}\bepsilon_j
                \right| + \left|\bdelta_j^T\bX^T Q_{\lambda_2^{(\ell)}}\bepsilon_j
                \right| + \left|
                \bepsilon_j^T Q_{\lambda_2^{(\ell)}}\bX(\wh\btheta^{(\ell)} - \btheta_j)
                \right|,
            $$
            it is straightforward to establish that 
            \begin{align}\label{bd_eps_Delta}\nonumber
             {1\over nm}\|\bDelta^T\bepsilon_j\|_2^2 
                &\lesssim {1\over m}\sum_{\ell=1}^m \left\{n[Rem_{1,\ell}]^2 +  Rem_{2,\ell}(\bdelta_j)\log(m) + \left[Rem_{3,\ell}(\btheta_j) + Rem_{2,\ell}(\bdelta_j)\right]^2
                \right\}\\
                &\lesssim n r_{n,1}^2 + r_{n,2} \log m + (r_{n,2} + r_{n,3})^2
            \end{align}
            with probability at least $1-m^{-1}$.  By collecting (\ref{bd_Delta_Delta}), (\ref{bd_Delta_eps}), (\ref{bd_eps_Delta}) and using $r_{n,2} \le r_n  = o_{\PP}(1)$ under Assumption \ref{ass_initial} to simplify the results, the proof is complete. 
        \end{proof}

     \subsection{Lemmas used in the proof of Lemma \ref{thm_Theta_simple_rates} and Theorem \ref{thm_asymp_normal}}\label{app_lemmas_Theta}
     
     The following two lemmas establish useful bounds on quantities related with $\bH_0$ and $\wh\bB$ that are used frequently in our proof. Recall that $r_{n}$ is defined in Assumption \ref{ass_initial} and $\eta_n$ is defined in (\ref{def_eta_bar}).
        
    \begin{lemma}\label{lemma_technical}
        Under Assumptions \ref{ass_error}, \ref{ass_B_Sigma} and \ref{ass_initial}, assume $M_n = o(m)$ and $\log m = o(n)$. The following holds with probability tending to one.
        \begin{enumerate}[label=(\Alph*)]
            \item  $c_H \lesssim \lambda_K(\bH_0) \le \lambda_1(\bH_0) \lesssim C_H$;
            \item $\max_{1\le j\le m}\|\wh\bB_j\|_2 \lesssim C_H\sqrt{C_B}$;
            \item $\lambda_K(\wh\bB) \gtrsim c_H\sqrt{c_B} \sqrt{m}$;
            \item $\max_{1\le j\le m} \norm{(\wt\bB-\wh\bB)\wh P_B^{\perp}\be_j}_2 \lesssim  \eta_n (C_H/c_H)\sqrt{C_B/c_B}$;
            \item $\max_{1\le j\le m}\|\wh P_B \be_j\|_2\lesssim m^{-1/2}(C_H/c_H)\sqrt{C_B/c_B}$; 
            \item $\|\bTheta \wh P_B \be_j\|_1 \lesssim m^{-1}\|\bTheta\|_{1,1} (C_H^2C_B)/(c_Hc_B)$.
        \end{enumerate}
        Here $c_H = c_W\sqrt{c_B/(C_W C_B)}$ and $C_H = C_W\sqrt{C_B/(c_Wc_B)}$ with $c_B, C_B, c_W, C_W$ defined in Assumption \ref{ass_B_Sigma}.
        \end{lemma}
        \begin{proof}
            Notice that $\eta_{n} = o(1)$ is implied by $r_{n}=o(1)$  and $\log m = o(n)$. 
            We work on the event
            \beq\label{def_event_B}
                \cE_B := \left\{
                  \max_{1\le j\le m}\|\wh\bB_j - \wt\bB_j\|_2 \lesssim \eta_{n}  
                \right\}
            \eeq
            intersecting with $\cE_D$ defined in (\ref{def_event_D}) and 
            \[
                \cE_W:= \left\{
                  c_W \lesssim \lambda_K\left({1\over n}\bW^T\bW\right) \le \lambda_1\left({1\over n}\bW^T\bW\right) \lesssim C_W
                \right\}.
            \]
            From Theorem \ref{thm_rates_B}, Lemma \ref{lem_D_K} and Lemma \ref{lem_W_eigens}, $\lim_{n\to\infty}\PP(\cE_B\cap \cE_D\cap \cE_W)= 1$. 
            
            To show (A), recall from (\ref{def_est_BW}) and (\ref{def_H0}) that 
            \[
                \bH_0^T = {1\over nm}\bW^T\bW \bB \wh\bB^T \bD_K^{-2} = {1\over n\sqrt{m}}\bW^T\bW \bB\bV_K \bD_K^{-1}. 
            \]
            It implies 
            \[
                \bH_0^T\bH_0 = {1\over n}\bW^T\bW \left({1\over m}\bB\bV_K \bD_K^{-2}\bV_K^{T}\bB^T\right) {1\over n}\bW^T\bW. 
            \]
            By invoking $\cE_W$, $\cE_D$ and Assumption \ref{ass_B_Sigma}, we then have 
            \[
                \lambda_K(\bH_0^T\bH_0) \gtrsim c_W^2 c_B / (C_WC_B).
            \]
            Similarly, 
            \[
                \lambda_1(\bH_0^T\bH_0) \lesssim C_W^2 C_B / (c_Wc_B).
            \]
            This proves (A).

           Part (B) then follows immediately by 
           \begin{align*}
               \|\wh\bB_j\|_2 &\le \|\wt \bB_j\|_2 + \|\wh\bB_j-\wt\bB_j\|_2\\ 
               &\le \lambda_1(\bH_0) \|\bB_j\|_2 + \eta_n\\
               &\lesssim C_H\sqrt{C_B}
           \end{align*}
           where we used Assumption \ref{ass_B_Sigma} in the penultimate step and $\eta_n = o(1)$ in the last step.  Similarly, using Weyl's inequality again yields
           \[
            \lambda_K(\wh \bB) \ge \lambda_K(\wt\bB) - \|\wh\bB-\wt\bB\|_{\op} \gtrsim {c_W\sqrt{c_B}\over \sqrt{C_WC_B}}\lambda_K(\bB) -\eta_n\sqrt{m}\gtrsim \sqrt{m}
           \]
           where the second inequality uses $\wt\bB^T = \bH_0 \bB^T$, part (A) and  
            \begin{equation}\label{bd_B_diff_op}
                \|\wh\bB- \wt\bB\|_{\op}^2 \le \|\wh\bB- \wt\bB\|_{F}^2 \le  m\eta_n^2
            \end{equation}
           on the event $\cE_B$.
           This proves part (C). Part (D) is proved by observing that 
           \begin{align*}
               \left\|(\wt\bB-\wh\bB)\wh P_B^{\perp}\be_j\right\|_2 &\le \left\|\wt\bB_j-\wh\bB_j\right\|_2 + \left\|(\wt\bB-\wh\bB)\wh P_B\be_j\right\|_2
           \end{align*}
           and 
           \beq\nonumber
        	\left\|(\wt\bB-\wh\bB)\wh P_B\be_j\right\|_2 = & ~ \norm{(\wt{\bB} - \wh{\bB})\wh{\bB}^T(\wh{\bB}\wh{\bB}^T)^{-1}\wh{\bB}\be_j}_2\\
        	\leq& ~ \norm{\wt{\bB} - \wh{\bB}}_{\op}\norm{\wh{\bB}^T(\wh{\bB}\wh{\bB}^T)^{-1}\wh{\bB}\be_j}_2\\
        	\leq& ~\eta_n\sqrt{m} ~ [\lambda_K(\wh\bB)]^{-1}\|\wh\bB_j\|_2
        	\eeq
           together with results in (B) and (C). Similarly, 
           \[
                \|\wh P_{B}\be_j\|_2 = \norm{\wh{\bB}^T(\wh{\bB}\wh{\bB}^T)^{-1}\wh{\bB}\be_j}_2 \le [\lambda_K(\wh\bB)]^{-1}\|\wh\bB_j\|_2 \lesssim m^{-1/2}(C_H/c_H)\sqrt{C_B/c_B}.
           \]
           Finally, 
           \[
                \|\bTheta \wh P_B \be_j\|_1 \le \|\bTheta\|_{1,1}\max_{1\le \ell\le m}\left|
                \be_\ell^T \wh\bB^T (\wh\bB\wh\bB^T)^{-1}\wh\bB\be_j\right| \le \|\bTheta\|_{1,1}{\|\wh\bB\|_{\i,2}^2 \over \lambda_K(\wh\bB\wh\bB^T)}.
           \]
           Invoke (B) and (C) to complete the proof.
        \end{proof}

         \begin{lemma}\label{lemma_PB_error}
        Under conditions of Lemma \ref{lemma_technical}, one has 
        \beq\nonumber
        \max_{1\le j\le m}\norm{(P_B - \wh P_B)\be_j}_2 = \cO_{\PP}\left(\eta_n\over \sqrt{m}\right),\quad \max_{1\le j\le m}\norm{(P_B - \wh P_B)\be_j}_\infty = \cO_{\PP}\left(\eta_n\over m\right).
        \eeq
    \end{lemma}
    \begin{proof}
        We prove the results by using Lemma \ref{lemma_technical}.
        We firstly bound the $\ell_2$ norm of $(P_B - \wh P_B)\be_j$ and will provide a sketch for bound in $\ell_\infty$ norm as the proof is very similar. Recall that $\wt \bB = \bH_0\bB$.
        By triangle inequality
        \beq
        \norm{(P_B - \wh P_B)\be_j}_2
        =&\norm{(\wt\bB^T(\wt\bB\wt\bB^T)^{-1}\wt\bB - \wh\bB^T(\wh\bB\wh\bB^T)^{-1}\wh\bB)\be_j}_2\\
        \leq& 
        \norm{(\wt\bB - \wh \bB)^T(\wt\bB\wt\bB^T)^{-1}\wt\bB\be_j}_2 + \norm{\wh\bB^T[(\wt\bB\wt\bB^T)^{-1} - (\wh\bB\wh\bB^T)^{-1}]\wt\bB\be_j}_2\\
        &\quad + 
        \norm{\wh\bB^T(\wh\bB\wh\bB^T)^{-1}(\wt\bB - \wh \bB)\be_j}_2\\
        :=&I_1 + I_2 + I_3.
        \eeq 
        Now we bound each term. For $I_1$
        \beq
        \norm{(\wt\bB - \wh \bB)^T(\wt\bB\wt\bB^T)^{-1}\wt\bB\be_j}_2  \leq&\norm{\wt\bB - \wh \bB}_{\op}\norm{(\wt\bB\wt\bB^T)^{-1}}_{\op}\norm{\wt\bB\be_j}_2\\
        \lesssim &\norm{\wt\bB - \wh \bB}_{\op}\norm{(\bB\bB^T)^{-1}}_{\op}\norm{\bB_j}_2\\
        =& \cO_{\PP}\left(\frac{\eta_n}{\sqrt{m}}\right),
        \eeq 
        where the last two steps follow from Lemma \ref{lemma_technical}. Similarly we can show that $I_3 = \cO_{\PP}(\eta_n/\sqrt{m})$. It remains to bound $I_2$. Direct calculation gives 
        \beq
        &\norm{\wh\bB^T[(\wt\bB\wt\bB^T)^{-1} - (\wh\bB\wh\bB^T)^{-1}]\wt\bB\be_j}_2\\
        &= ~\norm{\wh\bB^T(\wh\bB\wh\bB^T)^{-1}[\wt\bB\wt\bB^T - \wh\bB\wh\bB^T](\wt\bB\wt\bB^T)^{-1}\wt\bB\be_j}_2\\
        &\leq ~ [\lambda_K(\wh\bB)]^{-1}\left[
        \norm{(\wt\bB -\wh\bB)^T P_B\be_j}_2 + \norm{\wh\bB(\wt\bB -\wh\bB)^T(\wt\bB\wt\bB^T)^{-1}\wt\bB\be_j}_2
        \right]\\
        &\leq ~ [\lambda_K(\wh\bB)]^{-1}\left[\norm{\wt\bB -\wh\bB}_{\op}\norm{P_B\be_j}_2
        + \norm{\wh \bB}_{\op} I_1
        \right]\\
        &= ~ \cO_{\PP}\left(\frac{\eta_n}{\sqrt{m}}\right),
        \eeq 
        where the last step follows from Lemma \ref{lemma_technical} together with the bound for $I_1$. The proof for the $\ell_2$ bound is completed by combining the above results. 
        
        To show the result in $\ell_\infty$ norm, notice that we can similarly upper bound it by three terms $I_1'$ -- $I_3'$ in $\ell_\infty$ norm instead of $
        \ell_2$ norm by substituting $\max_{j}\norm{\wt\bB_j - \wh \bB_j}_{2}$ for $\norm{\wt \bB - \wh \bB}_{\op}$. For instance, $I_1' \leq \max_{j}\norm{\wt\bB_j - \wh \bB_j}_{2}\norm{(\wt\bB\wt\bB^T)^{-1}}_{\op}\norm{\wt\bB\be_j}_2 
        = \cO_{\PP}(\eta_n/m)$. The other two terms should follow similarly. This completes the proof.
    \end{proof}

    The following lemma proves that $\cE_{\bX}$ defined in (\ref{def_event_X}) holds with probability tending to one under conditions of Theorem \ref{thm_Theta_simple_rates}.
    
    \begin{lemma}\label{lem_X}
        Under Assumption \ref{ass_X}, assume $s_n \le C n/\log p$ for some large constant $C>0$ and $\|\Cov(Z)\|_{\op} = \cO(1)$. Then 
        \[
            \lim_{n\to\infty} \PP(\cE_{\bX}) = 1.
        \]
    \end{lemma}
    \begin{proof}
        When the rows of $\bX\Sigma^{-1/2}$ are i.i.d. sub-Gaussian random vector with bounded sub-Gaussian constant, provided that $\lambda_{\min}(\Sigma) \ge c_0$ for some constant $c_0>0$ and $s_n\log p \le Cn$ for some large constant $C>0$, \cite{rz13} shows that $\kappa(s_n,4) \ge c$ holds with probability $1-2n^{-c'n}$. 
         \cite{rz13} also  shows that 
        \begin{equation}\label{bd_op_XsXs}
            \sup_{S\subseteq[p]: |S|\le s_n}{1\over n}\lambda_1(\bX_S^T\bX_S) = \cO_{\PP}(1)
        \end{equation}
        provided that $\sup_{S\subseteq[p]: |S|\le s_n}\Sigma_{SS} =\cO(1)$.
        By applying Lemma \ref{lem_bernstein} with an union bound over $1\le j\le m$ and invoking $\max_{1\le j\le m}\Sigma_{jj} \le C$ from Assumption \ref{ass_X}, we have
        $$
            \max_{1\le j\le m}\wh\Sigma_{jj} \le \max_{1\le j\le m}\left(\Sigma_{jj} + |\wh\Sigma_{jj}-\Sigma_{jj}|\right) \le C'
        $$ 
        with probability $1 - 2(p\vee n)^{-1}$. For $\|\bX\bTheta\|_{2,1}$, since $\bTheta_{S^c\cdot} = \b0$, we have 
        \begin{align*}
            {1\over \sqrt n}\|\bX\bTheta\|_{2,1} =  {1\over \sqrt n}\|\bX_S\bTheta_{S\cdot}\|_{2,1}& \le  {1\over \sqrt n}\|\bX_S\|_{\op} \|\bTheta_{S\cdot}\|_{2,1}\\ &\overset{(\ref{bd_op_XsXs})}{=} \cO_{\PP}\left(\sqrt{s_n}\|\bTheta\|_{\i,1}\right)\overset{(\ref{def_space_Theta})}{=}
            \cO_{\PP}\left(M_n\sqrt{s_n}\right).
        \end{align*}
        Finally,  
        \begin{equation}\label{bd_XA_op}
            {1\over \sqrt n}\|\bX\bA\|_{\op} = \cO_{\PP}(1)
        \end{equation} 
        has been proved in \citet[Lemma 12]{bing2020adaptive}.
     \end{proof}    
     
     Under Assumption \ref{ass_X}, the following Lemma characterizes the estimation error of $\wh\bomega_1$ defined in (\ref{def_est_omega}) using (\ref{formula_nodewise}), as well as the order of $\wh\tau_1^2$ in (\ref{def_tau_1}). It is proved in \cite{vandegeer2014}. Recall that $s_\Omega = \norm{\bOmega_1}_0$.
    
    \begin{lemma}\label{lemma_nodewise}
    Under Assumption \ref{ass_X}, assume $s_\Omega \log p = o(n)$. By choosing $\wt \lambda \asymp \sqrt{\log p/n}$ in (\ref{formula_nodewise}), we have $1/\wh\tau_1^2 = \cO_{\PP}(1)$, 
    \[
        |\wh\bomega_1^T\wh\Sigma\wh\bomega_1 - \Omega_{11}| = \cO_{\PP}\left(\sqrt{\frac{s_{\Omega}\log p}{n}}\right),\qquad
        \norm{\be_1 - \wh\Sigma \wh\bomega_1}_{\i} = \cO_{\PP}\left(\sqrt{\frac{\log p}{n}}\right).
    \]
    \end{lemma}
    
    The following lemma provides upper bounds for $\|\be_1-\wh\Sigma\wh\bomega_1)^T\bA\|_2$.
    
    \begin{lemma}\label{lemma_nodewise_A}
        Under conditions of Lemma \ref{lemma_nodewise} and $\|\Cov(Z)\|_{\op}=\cO(1)$, one has 
        \[
            \|(\be_1 - \wh\Sigma\wh\bomega_1)^T\bA\|_2 = \cO_{\PP}\left(
                \sqrt{s_\Omega\log p\over n}
            \right)
        \]
    \end{lemma}
    \begin{proof}
       Use $\be_1 = \Sigma\bomega_1$ to obtain 
       \begin{equation}\label{eqn_decomp}
           (\be_1 - \wh\Sigma\wh\bomega_1)^T\bA = \bomega_1^T(\Sigma - \wh\Sigma)\bA + (\bomega_1-\wh\bomega_1)^T\wh\Sigma\bA. 
       \end{equation}
       For the first term, plugging $\bA = \Sigma^{-1}\Cov(X,Z)$ into the expression yields
       \begin{align*}
           \|\bomega_1^T(\Sigma - \wh\Sigma)\bA\|_2^2 = \sum_{k=1}^K\left(
           \bomega_1^T\Sigma^{1/2}\left(\bI_p - {1\over n} \bar\bX^T\bar\bX\right)
           \Sigma^{-1/2}\Cov(X,Z)\be_k\right)^2
       \end{align*}
       where $\bar\bX = \bX\Sigma^{-1/2}$. Notice that 
       \[
        \bomega_1^T\Sigma^{1/2}\left(\bI_p - {1\over n} \bar\bX^T\bar\bX\right)
           \Sigma^{-1/2}\Cov(X,Z)\be_k = {1\over n}\sum_{i=1}^n\left(\EE[U_i^TV_i] - U_iV_i
           \right)
       \]
       where $U_i = \bar\bX_{i\cdot}^T\Sigma^{1/2}\bomega_1$ is $\sqrt{\Omega_{11}}$ sub-Gaussian and $V_i = \bar\bX_{i\cdot}^T\Sigma^{-1/2}\Cov(X,Z)\be_k$ is $$
        \sqrt{\be_k^T \Cov(Z,X)\Sigma^{-1}\Cov(X,Z)\be_k} \le \sqrt{\Cov(Z_k)}
        $$ sub-Gaussian. 
       An application of Lemma \ref{lem_bernstein} with an union bound over $1\le k\le K$ gives 
       \[
          \left|\bomega_1^T\Sigma^{1/2}\left(\bI_p - {1\over n} \bar\bX^T\bar\bX\right)
           \Sigma^{-1/2}\Cov(X,Z)\be_k\right| = \cO\left(
           \sqrt{\Omega_{11}\Cov(Z_k) \over  n}
           \right)
       \]
       uniformly over $1\le k\le K$,
       with probability $1-O(n^{-1})$. Using (\ref{bd_Omega_11}) and $\|\Cov(Z)\|_{\op}=\cO(1)$ further yields 
       \beq\label{bd_Sigma_diff}
        \|\bomega_1^T(\Sigma - \wh\Sigma)\bA\|_2 = \cO_{\PP}\left(
        1/\sqrt{n}
        \right).
       \eeq
       Regarding the second term in (\ref{eqn_decomp}), one has 
       $$
        \|(\bomega_1-\wh\bomega_1)^T\wh\Sigma\bA\|_2 \le {1\over \sqrt{n}}\|\bX\bA\|_{\op}{1\over \sqrt n}\|\bX(\wh\bomega_1 - \bomega_1)\|_2 \overset{(\ref{bd_XA_op})}{=}  \cO_{\PP}(1) \cdot {1\over \sqrt n}\|\bX(\wh\bomega_1 - \bomega_1)\|_2.
       $$
       Recall from (\ref{def_est_omega}) that 
       \beq\label{bd_omega_diff}
            \wh\bomega_1^T = \wh\tau_1^{-2} \begin{bmatrix}
           1 & -\wh\bgamma_1^T
           \end{bmatrix}.
       \eeq
       Following \cite{vandegeer2014}, we define $\bgamma_1 = \argmin_{\bgamma\in\RR^{p-1}}\EE[\|\bX_1-\bX_{-1}\bgamma\|_2^2]$ and $\tau_1^2 = \EE[\|\bX_1-\bX_{-1}\bgamma_1\|_2^2]/n = \Omega_{11}^{-1}$ such that 
       $$
           \bomega_1^T = \tau_1^{-2} \begin{bmatrix}
           1 & -\bgamma_1^T
           \end{bmatrix}.
        $$
        Triangle inequality yields
       \begin{align*}
           {1\over \sqrt n}\|\bX(\wh\bomega_1 - \bomega_1)\|_2 &\le  {1\over \sqrt n}{\|\bX_{-1}(\wh\bgamma_1- \bgamma_1)\|_2 \over \wh\tau_1^2} +  {1\over \sqrt n}\|\bX_1 - \bX_{-1}\bgamma_1\|_2\left|{1\over \wh\tau_1^2}- {1\over \tau_1^2}\right|. 
       \end{align*}
       Using the results in \cite{vandegeer2014} yields 
       \[
       {1\over \sqrt n}\|\bX_{-1}(\wh\bgamma_1- \bgamma_1)\|_2 = \cO_{\PP}\left(\sqrt{s_\Omega\log p\over n}\right),\quad \left|{1\over \wh\tau_1^2}- {1\over \tau_1^2}\right| =  \cO_{\PP}\left(\sqrt{s_\Omega\log p\over n}\right).
       \]
       Together with 
       \[
        {1\over \sqrt n}\|\bX_1 - \bX_{-1}\bgamma_1\|_2 = \tau_1^2 {1\over \sqrt n}\|\bX\bomega_1\|_2 = \cO_{\PP}\left(\tau_1^2\sqrt{\bomega_1^T \Sigma \bomega_1}\right) = \cO_{\PP}(1)
       \]
       from (\ref{bd_Omega_11}), we conclude 
       \[
        {1\over \sqrt n}\|\bX(\wh\bomega_1 - \bomega_1)\|_2 = \cO_{\PP}\left(\sqrt{s_\Omega\log p\over n}\right).
       \]
       The proof is completed by combining the above display with (\ref{bd_Sigma_diff}) and (\ref{bd_omega_diff}).
    \end{proof}

        \subsection{Lemmas used in the proof of Theorem \ref{thm_B_asn}}

       Recall that $\bH_2 = \bB\wh\bB^T(\wh\bB\wh\bB^T)^{-1}$ and $\bQ$ is defined in (\ref{def_Q}). 
       The following lemma shows that $\bH_2$ converges to $\bQ^{-1}$ in probability.

      \begin{lemma}\label{lem_H2}
           Under conditions of Theorem \ref{thm_B_asn}, $\bH_2$ converges to $\bQ^{-1}$ in probability. 
       \end{lemma}
       \begin{proof}
            We prove the result by the same reasoning as \citet[Lemmas 1 \& 3]{bai2020simpler}. We first prove 
            \beq\label{WhatWconverge}
                \bH_1 = {1\over n}\wh\bW^T\bW \to \bQ,\quad \textrm{in probability,}
            \eeq
            and then show $\bH_2 = \bH_1^{-1} + o_\PP(1)$.
            Following the argument in \citet[Lemma 1]{bai2020simpler} and by expanding $\wh\bepsilon = \bW\bB+\bE+\bDelta$ with $\bDelta = \wh\bepsilon - \bepsilon$, we arrive at 
            \begin{align*}
                &{1\over n} \bW^T\wh\bW \bD_K^2\\
                & = {\bW^T\bW \over n} {\bB\bB^T \over m}{\bW^T\wh \bW \over n} + {\bW^T\bE\bE^T\wh\bW \over n^2m} + {\bW^T\bE\bB^T\over nm}{\bW^T\wh\bW\over n} + {\bW^T\bW\over n}{\bB\bE^T\wh\bW\over nm}\\
                &\quad +{1\over n^2m}\left(
                    \bW^T\bDelta \bepsilon^T \wh\bW + \bW^T\bDelta \bDelta^T \wh\bW + \bW^T\bepsilon \bDelta^T \wh\bW 
                \right).
            \end{align*}
            With $\wt\bW = \bW\bH_0^{-1}$, notice 
            \[
                {\bB\bE^T\wh\bW\over nm} = {\bB\bE^T\wt \bW\over nm} + {\bB\bE^T(\wh\bW-\wt\bW)\over nm}
            \]
            and 
            \[
                    {\bW^T\bE\bE^T\wh\bW \over n^2m} = {\bW^T\bE\bE^T\wt \bW \over n^2m} + {\bW^T\bE\bE^T(\wh\bW-\wt\bW) \over n^2m}.
            \]
            By arguments in \citet[Lemma 1]{bai2020simpler} and Lemma \ref{lem_W_frob}, one has 
            \[
                 {\bW^T\bE\bE^T\wh\bW \over n^2m} + {\bW^T\bE\bB^T\over nm}{\bW^T\wh\bW\over n} + {\bW^T\bW\over n}{\bB\bE^T\wh\bW\over nm} = o_\PP(1).
            \]
            Furthermore, by Lemma \ref{lem_quad_terms_Delta} and Lemma \ref{lem_W_eigens}, 
            \[
                {1\over n^2m}
                    \|\bW^T\bDelta \bepsilon^T \wh\bW\|_F \le {1\over \sqrt n}\|\bW\|_{\op} {1\over nm}\|\bDelta\bepsilon^T\|_F =\cO_\PP\left({1\over n\sqrt m}\max_{j\in[m]}\|\bDelta\bepsilon_j\|_2\right) = o_\PP(1).
            \]
            Using similar arguments yields
            \[
                {1\over n^2m}\left(
                    \bW^T\bDelta \bepsilon^T \wh\bW + \bW^T\bDelta \bDelta^T \wh\bW + \bW^T\bepsilon \bDelta^T \wh\bW 
                \right) =  o_\PP(1),
            \]
            and, therefore, 
            \[
                {\bW^T\wh\bW \over n} \bD_K^2 = {\bW^T\bW \over n} {\bB\bB^T \over m}{\bW^T\wh \bW \over n} + o_\PP(1).
            \]
            Finally, recalling $\Lambda_0$ from (\ref{def_Q}), note that $\bD_K^2\to \Lambda_0$ in probability. To see this, since \[
                \lambda_j(\Lambda_0) = {1\over m}\lambda_j(\bB\se \bB^T),
            \]
            for any $1\le j\le K$, 
            Weyl's inequality yields
            \[
                \left|\lambda_j(\bD_K^2) - \lambda_j(\Lambda_0)\right| \le {1\over m}\left\|
                   {1\over n}\wh\bepsilon^T \wh \bepsilon -  \bB \Sigma_W \bB^T 
                \right\|_{\op}.
            \]
            By the proof of Theorem \ref{thm_rates_B} together with Lemma \ref{lem_W_eigens}, it is easy to derive 
            \[
                \left|\lambda_j(\bD_K^2) - \lambda_j(\Lambda_0)\right| = o_\PP(1),\qquad \forall j\in [K],
            \]
            such that $\bD_K \to \Lambda_0$ in probability. Then
            the arguments in \citet[Lemma 1]{bai2020simpler} yield (\ref{WhatWconverge}). 
            It remains to prove 
            \[
                \bH_2^{-1}  =  \bH_1 + o_\PP(1).
            \]
            We prove this by using the same arguments in \citet[Lemma 3]{bai2020simpler} of showing that 
            $\bH_0 = \bH_1 + o_\PP(1)$ and $\bH_0 = \bH_2^{-1}+o_\PP(1)$, where we recall that 
            \[
                \bH_0^T = {1\over n}\bW^T\bW {1\over m}\bB\wh\bB^T\bD_K^{-2}.
            \]
            
            To prove $\bH_0= \bH_2^{-1}+o_\PP(1)$, notice that 
            \[
                \bD_K^{-1}\wh\bB\left({1\over nm}\wh\bepsilon^T \wh\bepsilon\right) \wh\bB^T \bD_K^{-1} = m \bD_K^{2}.
            \]
            Further expanding the left hand side by $\wh\bepsilon = \bW\bB+\bE+\bDelta$ with $\bDelta = \wh\bepsilon -\bepsilon$ yields
            \begin{align*}
                m\bD_K^2 &= \bD_K^{-1}{1\over m}\wh\bB\bB^T{1\over n}\bW^T\bW \bB\wh\bB^T \bD_K^{-1} + 
                2\bD_K^{-1}\left({1\over m}\wh\bB\bB^T\right)\left({1\over n}\bW^T\bE\wh\bB^T\right)\bD_K^{-1}\\
                &\quad + \bD_K^{-1}\left({1\over nm}\wh\bB\bE^T\bE\wh\bB^T\right)\bD_K^{-1} + 2\bD_K^{-1}\left({1\over nm}\wh\bB\bDelta^T \bepsilon\wh\bB^T\right)\bD_K^{-1}\\
                &\quad + \bD_K^{-1}\left({1\over nm}\wh\bB\bDelta^T \bDelta\wh\bB^T\right)\bD_K^{-1}.
            \end{align*}
            Since $\bH_2 = \bB\wh\bB^T(\wh\bB\wh\bB^T)^{-1} = \bB\wh\bB^T / m$, we conclude 
            \begin{align*}
                \bH_0^{-1} &= \bH_2 + 2\bH_0^{-1}\bD_K^{-1}\left({1\over m}\wh\bB\bB^T\right)\left({1\over nm}\bW^T\bE\wh\bB^T\right)\bD_K^{-1}\\
                &\quad + \bH_0^{-1}\bD_K^{-1}\left({1\over nm^2}\wh\bB\bE^T\bE\wh\bB^T\right)\bD_K^{-1} + 2\bH_0^{-1}\bD_K^{-1}\left({1\over nm^2}\wh\bB\bDelta^T \bepsilon\bB^T\right)\bD_K^{-1}\\
                &\quad + \bH_0^{-1}\bD_K^{-1}\left({1\over nm^2}\wh\bB\bDelta^T \bDelta\bB^T\right)\bD_K^{-1}.
            \end{align*}
            To show the last four terms on the right hand side are negligible, by Lemma \ref{lem_D_K} and \ref{lemma_technical}, one has 
            \begin{align*}
              &\left\|\bH_0^{-1}\bD_K^{-1}\left({1\over m}\wh\bB\bB^T\right)\left({1\over nm}\bW^T\bE\wh\bB^T\right)\bD_K^{-1}\right\|_F\\
                &\lesssim \left\|
                {1\over m}\wh\bB\bB^T
                \right\|_{\op} \left\|{1\over nm}\bW^T\bE\right\|_F \|\wh\bB\|_{\op}^2\\
                & \lesssim {1\over n\sqrt{m}}\left\|\bW^T\bE\right\|_F
            \end{align*}
            with probability tending to one. Since 
            \[
                 {1\over n\sqrt{m}}\left\|\bW^T\bE\right\|_F \le  {\sqrt{K}\over n}\max_{k\in[K],j\in [m]}\|\bW_k^T\bE_j\|_2 = \cO_\PP\left(\sqrt{\log m \over n}\right) = o_\PP(1)
            \]
            from Lemma \ref{lem_bernstein} with an union bound over $k\in[K]$ and $j\in [m]$ and $\log m = o(n)$, we have 
            \beq\label{bd_WE_F}
            {1\over n\sqrt{m}}\left\|\bW^T\bE\right\|_F = o_\PP(1).
            \eeq
            By similar arguments, we have 
            \[
                \left\|
                \bH_0^{-1}\bD_K^{-1}\left({1\over nm^2}\wh\bB\bE^T\bE\wh\bB^T\right)\bD_K^{-1}\right\|_{F} = \cO_\PP\left( {1\over nm}\|\bE\|_{\op}^2 \right) =  \cO_\PP\left({n + m \over nm}\right) = o_\PP(1)
            \]
            by also using Lemma \ref{lem_op_norm} and $\tr(\se) = \cO(m)$. Furthermore, invoke Lemma \ref{lem_quad_terms_Delta} to obtain
            \begin{align*}
                \|\bH_0^{-1}\bD_K^{-1}\left({1\over nm^2}\wh\bB\bDelta^T \bepsilon\wh\bB^T\right)\bD_K^{-1}\|_F &\lesssim {1\over nm}\|\bDelta^T \bepsilon\|_F \le {1\over n\sqrt{m}}\max_{j\in [m]}\|\bDelta^T \bepsilon_j\|_2 = o(1)
            \end{align*}
            and 
            \[
                \|\bH_0^{-1}\bD_K^{-1}\left({1\over nm^2}\wh\bB\bDelta^T \bDelta\wh\bB^T\right)\bD_K^{-1}\|_F \lesssim {1\over n\sqrt{m}}\max_{j\in [m]}\|\bDelta^T \bDelta_j\|_2 = o(1)
            \]
            with probability tending to one. Collecting terms concludes 
            $
                \bH_0^{-1} = \bH_2 + o_\PP(1),
            $
            or equivalently, 
            $\bH_0 = \bH_2^{-1} + o_\PP(1).$
            
            We proceed to show $\bH_0 = \bH_1 + o_\PP(1)$. From the basic equality 
            $\wh\bepsilon = \bW\bB+\bE+\bDelta$ and $\wh\bepsilon \wh\bB^T\bD_K^{-2} = m\sqrt{n}\bU_K = m\wh \bW$, we have 
            \begin{align*}
                 {1\over n} \bW^T\wh\bW & = {1\over nm}\bW^T \wh\bepsilon\wh\bB^T\bD_K^{-2}\\
                 &= {1\over n}\bW^T\bW {1\over m}\bB\wh\bB^T\bD_K^{-2} + {1\over nm}\bW^T\bE\wh\bB^T\bD_K^{-2} + {1\over nm}\bW^T\bDelta \wh\bB^T\bD_K^{-2},
            \end{align*}
            which leads to 
            \begin{align*}
                 \bH_1 & = \bH_0 + {1\over nm}\bD_K^{-2}\wh\bB\bE^T\bW + {1\over nm}\bD_K^{-2}\wh\bB\bDelta^T\bW.
            \end{align*}
            Previous arguments and (\ref{bd_WE_F}) give
            \[
                {1\over nm}\|\bD_K^{-2}\wh\bB\bE^T\bW\|_F = \cO_\PP\left({1\over n\sqrt m}\|\bE^T\bW\|_F\right) = o_\PP(1)
            \]
            and
            \[
                {1\over nm}\|\bD_K^{-2}\wh\bB\bDelta^T\bW\|_F = \cO_\PP\left(
                {1\over n\sqrt m}\|\bDelta^T\bW\|_F
                \right).
            \]
            Invoke Lemma \ref{lem_W_eigens} and Assumption \ref{ass_initial} to conclude 
            \[
                {1\over n\sqrt m}\|\bDelta^T\bW\|_F \le {1\over \sqrt n}\|\bW\|_{\op} {1\over \sqrt{nm}}\|\bDelta\|_F = o_\PP(1).
            \]
            We have finished the proof of $\bH_1 = \bH_0 + o_\PP(1) = \bH_2^{-1} + o_\PP(1)$, completing the proof. 
       \end{proof}

    \section{Auxiliary lemmas}\label{sec_proof_aux}
	
	The following lemma is used in our analysis. The tail inequality is for a quadratic form of sub-Gaussian random vectors. It is a slightly simplified version of Lemma 30 in \cite{Hsu2014} and is proved in \cite{bing2020adaptive}.
	\begin{lemma}\label{lem_quad}
		Let $\xi\in \RR^d$ be a $\gamma_\xi$ sub-Gaussian random vector. For all symmetric positive semi-definite matrices $H$, and all $t\ge 0$, 
		\[
		\PP\left\{
		\xi^T H \xi > \gamma_\xi^2\left(
		\sqrt{{\rm tr}(H)}+ \sqrt{2\|H\|_{\op}t}
		\right)^2
		\right\} \le e^{-t}.
		\] 
	\end{lemma}

	The following lemma provides an upper bound on the operator norm of $\bG H \bG^T$ where  $\bG\in \cR^{n\times d}$ is a random matrix and its rows are independent sub-Gaussian random vectors. It is proved in \cite{bing2020prediction}.
	\begin{lemma}\label{lem_op_norm}
	    Let $\bG$ be $n$ by $d$ matrix whose rows are independent $\gamma$ sub-Gaussian  random vectors with identity covariance matrix. Then for all symmetric positive semi-definite matrices $H$, 
		\[
		\PP\left\{{1\over n}\| \bG H \bG^T \|_{{\rm op}} \le \gamma^2\left( \sqrt{{\rm tr}(H) \over n} + \sqrt{6\|H\|_{\op}}
		\right)^2\right\} \ge  1 -  e^{-n}
		\]
	\end{lemma}
	
	Another useful concentration inequality of the operator norm of the random matrices with i.i.d. sub-Gaussian rows is stated in the following lemma. This is an immediate result of \citet[Remark 5.40]{vershynin_2012}.
	
	\begin{lemma}\label{lem_op_norm_diff}
 		Let $\bG$ be $n$ by $d$ matrix whose rows are i.i.d. $\gamma$ sub-Gaussian  random vectors with covariance matrix $\Sigma_Y$. Then for every $t\ge 0$, with probability at least  $1-2e^{-ct^2}$,
 		\[
 		\left\|	{1\over n}\bG^T \bG - \Sigma_Y\right\|_{{\rm op}}\le \max\left\{\delta, \delta^2\right\} \left\|\Sigma_Y\right\|_{{\rm op}},
 		\]
 		with $\delta = C\sqrt{d/n}+ t/\sqrt n$ where $c = c(\gamma)$ and $C=C(\gamma)$ are positive constants depending on $\gamma$.
 	\end{lemma}

    The deviation inequalities of the inner product of two random vectors with independent sub-Gaussian elements are well-known; we state the one in \cite{bing2020inference} for completeness. 
	
	\begin{lemma}\cite[Lemma 10]{bing2020inference}\label{lem_bernstein}
		Let $\{X_t\}_{t=1}^n$ and $\{Y_t\}_{t=1}^n$ be any two sequences, each with zero mean independent $\gamma_x$ sub-Gaussian and $\gamma_y$ sub-Gaussian elements. Then, for some absolute constant $c>0$, we have 
		\[
		\PP\left\{{1\over n}\left|\sum_{t=1}^n\left(X_t Y_t - \EE[X_t Y_t]\right)\right| \le \gamma_x \gamma_y t \right\}\ge 1-2\exp\left\{-c\min\left( t^2,t \right)n\right\}.
		\]
		In particular, when $\log N\le n$, one has
		\[
		\PP\left\{{1\over n}\left|\sum_{t=1}^n\left(X_t Y_t - \EE[X_t Y_t]\right)\right| \le C\sqrt{\log N \over n} \right\}\ge 1-2N^{-c}
		\]
		where $c \ge 2$ and $C = C(\gamma_x,\gamma_y,c)$ are some positive constants.
	\end{lemma}

\end{document}